\documentclass[12pt,a4paper]{article}%
\usepackage{a4}
\usepackage{amsmath}
\usepackage{amsfonts}
\usepackage{amssymb}
\usepackage{graphicx}
\usepackage{hyperref}%
\providecommand{\U}[1]{\protect\rule{.1in}{.1in}}
\newtheorem{theorem}{Theorem}[section]
\newtheorem{definition}[theorem]{Definition}
\newtheorem{lemma}[theorem]{Lemma}

\newtheorem{condition}[theorem]{Conditions}
\newtheorem{proposition}[theorem]{Proposition}
\newtheorem{corollary}[theorem]{Corollary}

\newtheorem{remark}[theorem]{Remark}
\newenvironment{proof}[1][Proof]{\textbf{#1.} }{\ \rule{0.5em}{0.5em}}
\def\phantom#1{}
\begin{document}

\title{ The Nested Off-shell Bethe ansatz\\and $O(N)$ Matrix Difference Equations}
\author{Hrachya M. Babujian\thanks{Address: Yerevan Physics Institute, Alikhanian
Brothers 2, Yerevan, 375036 Armenia} \thanks{E-mail: babujian@yerphi.am} ,
Angela Foerster\thanks{Address: Instituto de F\'{\i}sica da UFRGS, Av. Bento
Gon\c{c}alves 9500, Porto Alegre, RS - Brazil} \thanks{E-mail:
angela@if.ufrgs.br} , and Michael Karowski\thanks{E-mail:
karowski@physik.fu-berlin.de}\\Institut f\"{u}r Theoretische Physik, Freie Universit\"{a}t Berlin,\\Arnimallee 14, 14195 Berlin, Germany }
\date{{\small \today}}
\maketitle

\begin{abstract}
A system of $O(N)$-matrix difference equations is solved by means of the
off-shell version of the nested algebraic Bethe ansatz. In the nesting process
a new object, the $\Pi$-matrix, is introduced to overcome the complexities of
the $O(N)$ group structure. The proof of the main theorem is presented in
detail. In particular, the cancellation of all \textquotedblleft unwanted
terms\textquotedblright\ is shown explicitly. The highest weight property of
the solutions is proved.

\end{abstract}


\section{Introduction}

\label{s1} $O(N)$ Gross-Neveu and $O(N)$ $\sigma-$models are asymptotically
free quantum field theories which attract high interest, since they share some
common features with QCD. Since perturbation theory fails for these models,
exact results, such as exact generalized form factors are desirable and
welcome. The concept of a generalized form factor was introduced in
\cite{KW,BKW}, where several consistency equations were formulated.
Subsequently this approach was developed further and investigated in different
models by Smirnov \cite{Sm}. Generalized form factors are matrix elements of
fields with many particle states. To construct these objects explicitly one
has to solve generalized Watson's equations which are matrix difference
equations. To solve these equations the so called \textquotedblleft off-shell
Bethe ansatz" is applied \cite{BKZ2,BFKZ,BFK1}. The conventional Bethe ansatz
introduced by Bethe \cite{Bethe} is used to solve eigenvalue problems and its
algebraic formulation was developed by Faddeev and coworkers (see e.g.
\cite{TF}). The off-shell Bethe ansatz has been introduced in \cite{B3} to
solve the Knizhnik-Zamolodchikov equations which are differential equations.
In \cite{Re} a variant of this technique has been formulated to solve matrix
difference equations of the form
\[
K(u_{1},\dots,u_{i}+\kappa,\dots,u_{n})=\,K(u_{1}\dots,u_{i},\dots
,u_{n})Q(u_{1},\dots,u_{n},;i)~,~~(i=1,\dots,n)\,,
\]
where $K({\underline{u}})$ is a co-vector valued function, $Q({\underline{u}%
},i)$ are matrix valued functions and $\kappa$ is a constant to be specified.
We use here a co-vector formulation because this is more convenient for the
application to the form factor program. For higher rank internal symmetry
groups the nested version of this Bethe ansatz has to be applied. The nested
Bethe ansatz as a method to solve eigenvalue problems was introduced by Yang
\cite{Yang1} and further developed by Sutherland \cite{Su,sut}.

In this article we will solve the $O(N)$ difference equations combining the
nested Bethe ansatz with the off-shell Bethe ansatz. This procedure is similar
to the $SU(N)$ \cite{BKZ} case, where also a nesting procedure is used.
However, the algebraic formulation for $O(N)$ is much more intricate because
the R-matrix exhibits an extra new term. In addition, for $SU(N)$ we can use
the same R-matrix at every level, while for the group $O(N)$ the R-matrix
changes after each level. Therefore in our construction a new object, called
$\Pi$-matrix, is introduced in order to overcome these difficulties. This
provides a systematic formulation of techniques introduced by Tarasov
\cite{Tar} and also used in \cite{MR}. In \cite{dVK} a different procedure was
used to solve the $O(N)$ on-shell Bethe ansatz for even $N$.

The results of this article will be applied in \cite{BFK7} to calculate exact
form factors of the $O(N)$ $\sigma$- and Gross-Neveu models. We should mention
that the first computation of form factors for $O(3)$ $\sigma$-model is due to
\cite{Sm} (see also \cite{Ba,BH}). There are also new developments concerning
the connection between 2d Conformal Field Theory (CFT) and integrable models
with $N=2$ Super Yang Mills (SYM) theories in different higher dimensions.
First, there is a surprising relation between 2d-conformal blocks and the
instanton partition function in $N=2$, 4d-SYM theory \cite{AGT} (Alday,
Gaiotto, Tachikawa - AGT relation) and this is a particular version of the
AdS/CFT correspondence which is a more general part of the gauge/string
duality. There is also a q-deformation of the AGT relation which connects the
$N=2$ 5d-SYM theory and the q-deformed conformal blocks \cite{MMSS}. This last
relation offers new insights and gives the intriguing hope that the form
factor program can be used to obtain a deeper understanding of this
connection. The solution of the difference equations is the first step to
obtain the exact form factors and therefore important physical relations and
correlation functions for integrable models. In fact, difference equations
play a significant role in various contexts of mathematical physics (see e.g.
\cite{FR} and references therein).

The article is organized as follows. In Section \ref{s2} we recall some
results and fix the notation concerning the $O(N)$ R-matrix, the monodromy
matrix and some commutation rules. We also introduce a new object, which we
call the $\Pi$-matrix and which is a central element in our construction of
the nested off-shell Bethe vector. In Section \ref{s3} we introduce the nested
generalized Bethe ansatz to solve a system of $O(N)$ difference equations and
present the solutions in terms of \textquotedblleft Jackson-type Integrals".
We introduce a new type of monodromy matrix fulfilling a new type of
Yang-Baxter relation and which is adapted to the difference problem. In
particular this yields a relatively simple proof of our main result, which is
Theorem \ref{TN}. In Section \ref{s4} we prove the highest weight property of
the solutions and calculate the weights. The appendices provide the more
complicated proofs of the results we have obtained. In particular, in Appendix
\ref{a2} we determine all \textquotedblleft unwanted terms\textquotedblright%
\ in the Bethe ansatz and show that they cancel.

\section{General setting and notion of the $\Pi$-matrix}

\label{s2}

\subsection{The $O(N)$ - R-matrix}

Let ${V^{1\dots n}}$ be the tensor product space
\begin{equation}
{V^{1\dots n}}=V_{1}\otimes\dots\otimes V_{n}\,, \label{2.1}%
\end{equation}
where the vector spaces $V_{i}\cong\mathbf{C}^{N},~(i=1,\dots,n)$ are copies
of the fundamental vector representation space of $O(N)$ with the (real) basis
vectors%
\[
\,|\,\alpha\,\rangle_{r}\in{V}_{i},~~(\alpha=1,\dots,N)\,.
\]
It is straightforward to generalize the results of this paper to the case
where the $V_{i}$ are vector spaces for other representations. We denote the
canonical basis vectors of ${V^{1\dots n}}$ by
\begin{equation}
\,|\,\alpha_{1},\dots,\alpha_{n}\,\rangle\in{V^{1\dots n}},~~(\alpha
_{i}=1,\dots,N)\,. \label{2.2}%
\end{equation}
A vector $v^{1\dots n}\in{V^{1\dots n}}$ is given in terms of its components
by
\begin{equation}
v^{1\dots n}=\sum_{\underline{\alpha}}\,|\,\alpha_{1},\dots,\alpha
_{n}\,\rangle_{r}\,v^{\alpha_{1},\dots,\alpha_{n}}. \label{2.3}%
\end{equation}
A matrix acting in ${V^{1\dots n}}$ is denoted by
\begin{equation}
{A_{1\dots n}}~:~{V^{1\dots n}}\rightarrow{V^{1\dots n}}. \label{2.4}%
\end{equation}
We will also use the dual space ${V}_{1\dots n}=\left(  {V^{1\dots n}}\right)
^{\dagger}$.

The $O(N)$ spectral parameter dependent R-matrix was found by
Zamolodchikov-Zamolodchikov \cite{ZZ}\footnote{We use here the normalization
$R=S/\sigma_{2}$ and the parameterization $u=\theta/i\pi\nu$ which is more
convenient for our purpose.}. It acts on the tensor product of two
(fundamental) representation spaces of $O(N)$. It may be written as
\begin{equation}
R_{12}(u_{12})=\left(  \mathbf{1}_{12}+c(u_{12})\,\mathbf{P}_{12}%
+d(u_{12})\,\mathbf{K}_{12}\right)  :\,V^{12}\rightarrow V^{21}\,, \label{2.5}%
\end{equation}
where $\mathbf{P}_{12}$ is the permutation operator, $\mathbf{K}_{12}$ the
annihilation-creation operator and $u_{12}=u_{1}-u_{2}$. Here and in the
following we associate to each space $V_{i}$ a variable (spectral parameter)
$u_{i}\in\mathbb{C}$. The components of the R-matrix are
\begin{equation}
R_{\alpha\beta}^{\delta\gamma}(u_{12})=\delta_{\alpha}^{\gamma}\,\delta
_{\beta}^{\delta}+\delta_{\alpha}^{\delta}\,\delta_{\beta}^{\gamma}%
\,c(u_{12})+\delta^{\gamma\delta}\delta_{\alpha\beta}\,d(u_{12})=~%
\begin{array}
[c]{c}%
\unitlength2mm\begin{picture}(6,7) \put(1,1){\line(1,1){4}} \put(5,1){\line(-1,1){4}} \put(.5,-.5){$ \alpha$} \put(5,-.5){$ \beta$} \put(5,5.5){$\gamma$} \put(.5,5.5){$ \delta$} \put(-.7,2){$u_1$} \put(4.8,2){$ u_2$} \end{picture}
\end{array}
~, \label{2.6}%
\end{equation}
from which $\mathbf{P}_{12}$ and $\mathbf{K}_{12}$ can be read off. The
functions
\begin{equation}
c(u)=\frac{-1}{u},~~d(u)=\frac{1}{u-1/\nu},~~\nu=\dfrac{2}{N-2} \label{2.7}%
\end{equation}
are obtained as the rational solution of the Yang-Baxter equation
\begin{gather}
R_{12}(u_{12})\,R_{13}(u_{13})\,R_{23}(u_{23})=R_{23}(u_{23})\,R_{13}%
(u_{13})\,R_{12}(u_{12})\label{1.3}\\%
\begin{array}
[c]{c}%
\unitlength4mm\begin{picture}(9,4) \put(0,1){\line(1,1){3}} \put(0,3){\line(1,-1){3}} \put(2,0){\line(0,1){4}} \put(4.3,2){$=$} \put(6,0){\line(1,1){3}} \put(6,4){\line(1,-1){3}} \put(7,0){\line(0,1){4}} \put(.2,.5){$\scriptstyle 1$} \put(1.3,0){$\scriptstyle 2$} \put(3,.2){$\scriptstyle 3$} \put(5.5,.2){$\scriptstyle 1$} \put(7.2,0){$\scriptstyle 2$} \put(8.4,.4){$\scriptstyle 3$} \end{picture}~~,
\end{array}
\end{gather}
where we have employed the usual notation \cite{Yang1}. We will also use%
\[
\tilde{R}(u)=R(u)/a(u)
\]
with
\[
a(u)=1+c(u)=\frac{u-1}{u}\,.
\]
The \textquotedblleft unitarity" of the R-matrix reads as
\[
\tilde{R}_{21}(u_{21})\,\tilde{R}_{12}(u_{12})=1~:~~~~~%
\begin{array}
[c]{c}%
\unitlength3mm\begin{picture}(9,4) \put(1,0){\line(1,1){2}} \put(3,0){\line(-1,1){2}} \put(1,2){\line(1,1){2}} \put(3,2){\line(-1,1){2}} \put(7,0){\line(0,1){4}} \put(9,0){\line(0,1){4}} \put(4.5,1.7){$=$} \put(.2,0){$\scriptstyle 1$} \put(3.2,0){$\scriptstyle 2$} \put(6.2,0){$\scriptstyle 1$} \put(8.2,0){$\scriptstyle 2$} \end{picture}
\end{array}
~
\]
and the three eigenvalues of the R-matrix are%
\begin{equation}
R_{\pm}(u)=1\pm c(u)=\frac{u\mp1}{u},~R_{0}=1+c(u)+Nd(u)=\frac{u+1}{u}%
\frac{u+1/\nu}{u-1/\nu}\,. \label{EV}%
\end{equation}
The crossing relation may be written as
\begin{gather}
R_{12}(u_{12})=\mathbf{C}^{2\bar{2}}\,R_{\bar{2}1}(\hat{u}_{12})\,\mathbf{C}%
_{\bar{2}2}=\mathbf{C}^{1\bar{1}}\,R_{2\bar{1}}(\hat{u}_{12})\,\mathbf{C}%
^{\bar{1}1}\label{Rc}\\%
\begin{array}
[c]{c}%
\unitlength3mm\begin{picture}(4,5) \put(0,1){\line(1,1){4}} \put(4,1){\line(-1,1){4}} \put(0,-.5){$1$} \put(3.7,-.5){$2$} \end{picture}
\end{array}
\quad=\quad%
\begin{array}
[c]{c}%
\unitlength3mm\begin{picture}(6,5) \put(1,1){\line(1,1){4}} \put(4,1){\line(-1,2){2}} \put(1,5){\oval(2,8)[lb]} \put(5,1){\oval(2,8)[tr]} \put(3.5,-.5){$1$} \put(5.7,-.5){$2$} \end{picture}
\end{array}
\quad=\quad%
\begin{array}
[c]{c}%
\unitlength3mm\begin{picture}(6,5) \put(2,1){\line(1,2){2}} \put(5,1){\line(-1,1){4}} \put(1,1){\oval(2,8)[lt]} \put(5,5){\oval(2,8)[br]} \put(0,-.5){$1$} \put(2,-.5){$2$} \end{picture}
\end{array}
~\,,~~~~
\end{gather}
where $\hat{u}=1/\nu-u$. Here $\mathbf{C}^{1\bar{1}}$ and $\mathbf{C}%
_{1\bar{1}}$ are the charge conjugation matrices. Their matrix elements are
$\mathbf{C}^{\alpha\bar{\beta}}=\mathbf{C}_{\alpha\bar{\beta}}=\delta
_{\alpha\beta}$ where $\bar{\beta}$ denotes the anti-particle of $\beta$. In
the real basis used up to now the particles are chargeless which means that
$\bar{\beta}=\beta$ and $\mathbf{C}$ is diagonal.

In the following we will use instead of the real basis $|\alpha\rangle
_{r}\,,~\left(  \alpha=1,2,\dots,N\right)  \,$the complex basis given by
\begin{align*}
|\alpha\rangle &  =\frac{1}{\sqrt{2}}\left(  |2\alpha-1\rangle_{r}%
+i|2\alpha\rangle_{r}\right) \\
|\bar{\alpha}\rangle &  =\frac{1}{\sqrt{2}}\left(  |2\alpha-1\rangle
_{r}-i|2\alpha\rangle_{r}\right)
\end{align*}
for $\alpha=1,2,\dots,\left[  N/2\right]  $. If $N$ is odd there is in
addition $|0\rangle=|\bar{0}\rangle=|N\rangle_{r}$. The weight vector
$w=\left(  w_{1},\dots,w_{\left[  N/2\right]  }\right)  $ and the charges of
the one-particle states are given by%
\[%
\begin{array}
[c]{llllll}%
\text{for} & |\,\alpha\,\rangle & : & w_{k}=\delta_{k\alpha} & , & Q=1\\
\text{for} & |\,\bar{\alpha}\,\rangle & : & w_{k}=-\delta_{k\alpha} & , &
Q=-1\\
\text{for} & |\,0\,\rangle & : & w_{k}=0 & , & Q=0\,.
\end{array}
\]

\begin{remark}
For even $N$ this means that we consider $O(N)$ as a subgroup of $U(N/2)$ and
the charge $Q$ is its $U(1)$ charge. For $N=3$ we may identify the particles
$1,\bar{1},0$ with the pions $\pi_{\pm},\pi_{0}$.
\end{remark}

The highest weight eigenvalue of the R-matrix is%
\[
R_{11}^{11}(u)=R_{+}(u)=a(u).
\]

We order the states as: $1,2,\dots,(0),\dots,\bar{2},\bar{1}$ ($0$ only for
$N$ odd). Then the charge conjugation matrix in the complex basis is of the
form%
\begin{gather}
\mathbf{C}^{\delta\gamma}=\delta^{\delta\bar{\gamma}}\,,~\mathbf{C}%
_{\alpha\beta}=\delta_{\alpha\bar{\beta}}\label{Cc}\\
\mathbf{C}\mathbf{=}\left(
\begin{array}
[c]{ccccc}%
0 & \cdots & 0 & \cdots & 1\\
\vdots & \ddots & \vdots & \cdot & \vdots\\
0 & \cdots & 1 & \cdots & 0\\
\vdots & \cdot & \vdots & \ddots & \vdots\\
1 & \cdots & 0 & \cdots & 0
\end{array}
\right)  .\nonumber
\end{gather}
The annihilation-creation matrix in (\ref{2.5}) may be written as%
\[
\mathbf{K}_{\alpha\beta}^{\delta\gamma}=\mathbf{C}^{\delta\gamma}%
\mathbf{C}_{\alpha\beta}\,.
\]

\subsection{The monodromy matrix}

We consider a state with $n$ particles and as is usual in the context of the
algebraic Bethe ansatz we define \cite{FST,TF} the monodromy matrix by%
\begin{equation}
T_{1\dots n,0}(\underline{u},u_{0})=R_{10}(u_{10})\,\cdots R_{n0}(u_{n0})=%
\begin{array}
[c]{c}%
\unitlength3mm\begin{picture}(9,4.5) \put(0,2){\line(1,0){9}} \put(2,0){\line(0,1){4}} \put(7,0){\line(0,1){4}} \put(1,0){$1$} \put(5.8,0){$ n$} \put(8.2,.7){$ 0$} \put(3.6,3){$\dots$} \end{picture}
\end{array}
\label{T}%
\end{equation}
with $\underline{u}=u_{1},\dots,u_{n}$. It is a matrix acting in the tensor
product of the \textquotedblleft quantum space\textquotedblright\ $V^{1\dots
n}=V_{1}\otimes\cdots\otimes V_{n}$ and the \textquotedblleft auxiliary
space\textquotedblright\ $V_{0}$. All vector spaces $V_{i}$ are isomorphic to
a space $V$ whose basis vectors label all kinds of particles. Here
$V\cong\mathbb{C}^{N}$ is the space of the vector representation of $O(N)$.

Suppressing the indices $1\ldots n$ we write the monodromy matrix as
(following the notation of Tarasov \cite{Tar})
\begin{equation}
T_{\alpha}^{\alpha^{\prime}}=~\alpha^{\prime}%
\begin{array}
[c]{c}%
\unitlength1.5mm\begin{picture}(10,4) \put(0,2){\line(1,0){10}} \put(2,0){\line(0,1){4}} \put(8,0){\line(0,1){4}} \put(3.5,.8){$\dots$} \end{picture}
\end{array}
\,\alpha~=\left(
\begin{array}
[c]{ccc}%
A_{1} & \left(  B_{1}\right)  _{\mathring{\alpha}} & B_{2}\\
\left(  C_{1}\right)  ^{\mathring{\alpha}^{\prime}} & \left(  A_{2}\right)
_{\mathring{\alpha}}^{\mathring{\alpha}^{\prime}} & \left(  B_{3}\right)
^{\mathring{\alpha}^{\prime}}\\
C_{2} & \left(  C_{3}\right)  _{\mathring{\alpha}} & A_{3}%
\end{array}
\right)  \label{T1}%
\end{equation}
where $\alpha,\alpha^{\prime}$ assume the values $1,2,\dots,(0),\dots,\bar
{2},\bar{1}$ corresponding to the basis vectors of the auxiliary space
$V\cong\mathbb{C}^{N}$ and $\mathring{\alpha},\mathring{\alpha}^{\prime}$
assume the values $2,\dots,(0),\dots,\bar{2}$ corresponding to the basis
vectors of $\mathring{V}\cong\mathbb{C}^{N-2}$. We will also use the notation
$A=A_{1},~B=B_{1},~C=C_{1}$ and $D=A_{2}$ which is an $\left(  N-2\right)
\times\left(  N-2\right)  $ matrix in the auxiliary space. The Yang-Baxter
algebra relation for the R-matrix (\ref{1.3}) yields
\begin{gather}
T_{1\dots n,a}(\underline{u},u_{a})\,T_{1\dots n,b}(\underline{u}%
,u_{b})\,R_{ab}(u_{ab})=R_{ab}(u_{ab})\,T_{1\dots n,b}(\underline{u}%
,u_{b})\,T_{1\dots n,a}(\underline{u},u_{a})\label{TTS}\\
\unitlength4mm\begin{picture}(20,5.5) \put(.3,0){$1$} \put(4.3,0){$n$} \put(0,2.3){$b$} \put(0,4.3){$a$} \put(7.5,2.3){$b$} \put(7.5,.5){$a$} \put(2.5,3){$\dots$} \put(1,0){\line(0,1){5}} \put(0,2){\line(1,0){8}} \put(5,0){\line(0,1){5}} \put(0,2){\oval(14,4)[rt]} \put(8,2){\oval(2,4)[lb]} \put(9.3,2.3){$=$} \put(13.3,0){$1$} \put(17.3,0){$n$} \put(11,2){$b$} \put(11,4.2){$a$} \put(18.5,3.4){$b$} \put(18.5,1.4){$a$} \put(15.5,2){$\dots$} \put(14,0){\line(0,1){5}} \put(18,0){\line(0,1){5}} \put(11,3){\line(1,0){8}} \put(19,3){\oval(14,4)[lb]} \put(11,3){\oval(2,4)[rt]} \end{picture}\,.~~~~\nonumber
\end{gather}

\subsection{A lemma}

In our approach of the algebraic Bethe ansatz the following lemma replaces
commutation rules of the entries of the monodromy matrix. In the conventional
approach one derives them from the Yang-Baxter algebra relations (\ref{TTS})
and uses them for the algebraic Bethe ansatz.

\begin{lemma}
\label{l2}For the monodromy matrix the following identity holds%
\begin{multline}
T_{1\dots n,a}(\underline{u},v)=\mathbf{1}_{1}\dots\mathbf{1}_{n}%
\mathbf{1}_{a}+\sum_{i=1}^{n}c(u_{i}-v)R_{1a}(u_{1i})\dots\mathbf{P}_{ia}\dots
R_{na}(u_{ni})\\
+\sum_{j=1}^{n}d(u_{i}-v)R_{1a}(\hat{u}_{i1})\dots\mathbf{K}_{ia}\dots
R_{na}(\hat{u}_{in}) \label{Sb}%
\end{multline}
with $\hat{u}=1/\nu-u$, or in terms of pictures%
\begin{multline*}%
\begin{array}
[c]{c}%
\unitlength3mm\begin{picture}(10,4) \put(0,2){\line(1,0){10}} \put(2,0){\line(0,1){4}} \put(5,0){\line(0,1){4}} \put(8,0){\line(0,1){4}} \put(1.2,0){$\scriptstyle 1$} \put(4.4,0){$\scriptstyle i$} \put(7,0){$\scriptstyle n$} \put(9,1){$\scriptstyle 0$} \put(3,1){$\scriptstyle \dots$} \put(6,1){$\scriptstyle \dots$} \end{picture}%
\raisebox{0.0524in}{\phantom{\rule{1.3513in}{0.6349in}}}%
\end{array}
=%
\begin{array}
[c]{c}%
{\phantom{\rule{1.0795in}{0.5069in}}}%
\unitlength3mm\begin{picture}(10,4) \put(0,2){\line(1,0){2}} \put(3,0){\line(0,1){4}} \put(8,0){\line(0,1){4}} \put(2.2,0){$\scriptstyle 1$} \put(7,0){$\scriptstyle n$} \put(1,1){$\scriptstyle 0$} \put(5,2){$\dots$} \end{picture}
\end{array}
\\
+\sum_{i=1}^{n}c(u_{i}-v)~%
\begin{array}
[c]{c}%
\raisebox{-0.0166in}{\phantom{\rule{1.3754in}{0.6823in}}}%
\unitlength3mm\begin{picture}(10,4) \put(0,0){\oval(10,4)[rt]} \put(10,4){\oval(10,4)[lb]} \put(2,0){\line(0,1){4}} \put(8,0){\line(0,1){4}} \put(1.2,0){$\scriptstyle 1$} \put(4.4,0){$\scriptstyle i$} \put(7,0){$\scriptstyle n$} \put(9,1){$\scriptstyle 0$} \put(3,1){$\scriptstyle \dots$} \put(6,1){$\scriptstyle \dots$} \end{picture}
\end{array}
+\sum_{i=1}^{n}d(u_{i}-v)~%
\begin{array}
[c]{c}%
{\phantom{\rule{1.3363in}{0.6274in}}}%
\unitlength3mm\begin{picture}(10,4) \put(10,0){\oval(10,4)[lt]} \put(0,4){\oval(10,4)[rb]} \put(2,0){\line(0,1){4}} \put(8,0){\line(0,1){4}} \put(1.2,0){$\scriptstyle 1$} \put(4.4,0){$\scriptstyle i$} \put(7,0){$\scriptstyle n$} \put(9,1){$\scriptstyle 0$} \put(3,1){$\scriptstyle \dots$} \put(6,1){$\scriptstyle \dots$} \end{picture}
\end{array}
.
\end{multline*}

\end{lemma}

\begin{proof}
The R-matrix $R(u)$ (see (\ref{2.6}) and (\ref{2.7})) is meromorphic and has
simple poles at $u=0$ and $u=1/\nu$ due to the form of $c(u)$ and $d(u)$ such
that%
\begin{align*}
\operatorname*{Res}_{v=u_{i}}T_{1\dots n,a}(\underline{u},v)  &
=\operatorname*{Res}_{v=u_{i}}c(u_{i}-v)R_{1a}(u_{1i})\dots\mathbf{P}%
_{ia}\dots R_{na}(u_{ni})\\
\operatorname*{Res}_{v=u_{i}-1/\nu}T_{1\dots n,a}(\underline{u},v)  &
=\operatorname*{Res}_{v=u_{i}-1/\nu}d(u_{j}-v)R_{1a}(\hat{u}_{i1}%
)\dots\mathbf{K}_{ia}\dots R_{ma}(\hat{u}_{in})
\end{align*}
holds. The claim follows by Liouville's theorem because $\lim
\limits_{v\rightarrow\infty}T_{1\dots n,a}(\underline{u},v)=\mathbf{1}%
_{1}\dots\mathbf{1}_{n}\mathbf{1}_{a}$.
\end{proof}

Similarly we have for the crossed monodromy matrix%
\[
T_{a,1\dots n}(v,\underline{u})=R_{an}(v-u_{n})\dots R_{a1}(v-u_{1})
\]
the relation%
\begin{multline}
T_{a,1\dots n}(v,\underline{u})=\mathbf{1}_{a}\mathbf{1}_{1}\dots
\mathbf{1}_{n}+\sum_{i=1}^{n}c(v-u_{i})R_{an}(u_{in})\dots\mathbf{P}_{ai}\dots
R_{a1}(u_{i1})\\
+\sum_{i=1}^{n}d(v-u_{i})R_{am}(\hat{u}_{mi})\dots\mathbf{K}_{ai}\dots
R_{a1}(\hat{u}_{1i})\,. \label{Sa}%
\end{multline}
Note that the crossing relation (\ref{Rc}) implies%
\begin{equation}
T_{a,1\dots n}(v_{a},\underline{u})=\mathbf{C}_{ba}T_{1\dots n,b}%
(\underline{u},v_{b})\mathbf{C}^{ab} \label{Tc}%
\end{equation}
with $v_{b}=v_{a}-1/\nu$.

\subsection{The Matrix $\Pi$}

\label{s2.4}The nested Bethe ansatz relies on the principle that after each
level the rank of the group (or quantum group) is reduced by one. For $SU(N)$
the rank is $N-1$ and for $O(N)$ it is $\left[  N/2\right]  $. This means that
the dimension of the vector representation (where the R-matrix usually acts)
is reduced by $1$ for the case of $SU(N)$ and by $2$ for case of $O(N)$. A
more essential difference is that for $SU(N)$ one can use in every level the
same R-matrix, because (with a suitable normalization and parameterization)
the $SU(N)$ R-matrix does not depend on $N$. In contrast for $O(N)$ the
R-matrix changes after each level, because it depends on $N$. Therefore we
need a new object called matrix $\Pi$, which maps the $O(N)$ R-matrix to the
$O(N-2)$ one. We use the notation%
\begin{align}
\mathring{R}(u)  &  =R(u,N-2)=\mathbf{1}+\mathbf{P}c(u)+\mathbf{K}\mathring
{d}(u)\label{R'}\\
\mathring{d}(u)  &  =\frac{1}{u-1/\mathring{\nu}}=\frac{1}{u-1/\nu+1}\nonumber
\end{align}
with $\mathring{\nu}=2/(N-4)$. The components of the R-matrix $\mathring
{R}(u)$ will be denoted by
\[
\mathring{R}_{\mathring{\alpha}\mathring{\beta}}^{\mathring{\delta}%
\mathring{\gamma}}(u)\,,~\mathring{\alpha},\mathring{\beta},\mathring{\gamma
},\mathring{\delta}=2,3,\dots,(0),\dots,\bar{3},\bar{2}\,.
\]
In addition to ${V^{1\dots m}}=V_{1}\otimes\dots\otimes V_{m}$ (\ref{2.1}) we
introduce
\begin{equation}
{\mathring{V}^{1\dots m}}={\mathring{V}}_{1}\otimes\dots\otimes{\mathring{V}%
}_{m}\,,
\end{equation}
where the vector spaces ${\mathring{V}}_{i}\cong\mathbf{C}^{N-2}%
,~(i=1,\dots,m)$ are considered as fundamental (vector) representation spaces
of $O(N-2)$. The space $V_{i}$ is spanned by the complex basis vectors
$|1\rangle,|2\rangle,\dots,|\bar{2}\rangle,|\bar{1}\rangle$ and ${\mathring
{V}}_{i}$ by $|2\rangle,\dots,|\bar{2}\rangle$ .

\begin{definition}
We define the map%
\[
\Pi_{1\dots m}:{V^{1\dots m}}\rightarrow{\mathring{V}^{1\dots m}}%
\]
recursively by $\Pi_{1}=\pi_{1}$ and%
\begin{equation}
\Pi_{1\dots m}(\underline{u})=\left(  \pi_{1}\Pi_{2\dots m}\right)  \bar
{e}_{a}T_{1\dots m,a}(\underline{u},u_{a})\bar{e}^{a} \label{Pi}%
\end{equation}
with the projector $\pi:V\rightarrow{\mathring{V}\subset V}$, the monodromy
matrix (\ref{T}) and $u_{a}=u_{1}-1$. The vector $\bar{e}^{a}\in V_{a}$
(acting in the auxiliary space of $T_{1\dots m,a}$) and the co-vector $\bar
{e}_{a}\in\left(  V_{a}\right)  ^{\dagger}$ correspond to the state $\bar{1}$
and have the components $\bar{e}^{\alpha}=\delta_{\bar{1}}^{\alpha},~\bar
{e}_{\alpha}=\delta_{\alpha}^{\bar{1}}\,.$This definition may be depicted as%
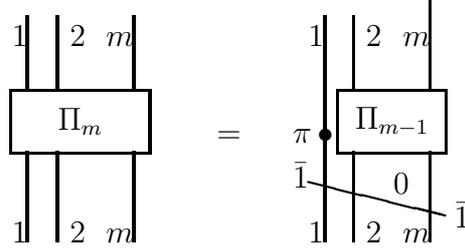
\[%
{\phantom{\rule{2.307in}{1.3214in}}}%
\begin{array}
[c]{c}%
\unitlength4mm\begin{picture}(6,7.5)\thicklines \put(.5,5){\line(0,1){2.5}} \put(1.5,5){\line(0,1){2.5}} \put(4,5){\line(0,1){2.5}} \put(0,3){\framebox(4.5,2){$\Pi_m$}} \put(.5,0){\line(0,1){3}} \put(1.5,0){\line(0,1){3}} \put(4,0){\line(0,1){3}} \put(0,6.5){$1$} \put(0,0){$1$} \put(1.9,0){$2$} \put(1.9,6.5){$2$} \put(3.1,6.5){$m$} \put(3.1,0){$m$} \end{picture}
\end{array}
=~~~~%
\begin{array}
[c]{c}%
\unitlength4mm\begin{picture}(6,7.5)\thicklines \put(1.5,5){\line(0,1){2.5}} \put(4,5){\line(0,1){3}} \put(-.5,3.3){$\pi\,\bullet$} \put(1,3){$\framebox(3.5,2){$\Pi_{m-1}$}$} \put(.55,0){\line(0,1){7.5}} \put(1.5,0){\line(0,1){3}} \put(4,0){\line(0,1){3}} \put(0,2){\line(4,-1){4.5}} \put(-.5,1.8){$\bar1$} \put(4.8,.5){$\bar1$} \put(0,6.5){$1$} \put(0,0){$1$} \put(1.9,0){$2$} \put(1.9,6.5){$2$} \put(3.1,6.5){$m$} \put(3.1,0){$m$} \put(2.8,1.6){$0$} \end{picture}
\end{array}
\]

\end{definition}

\begin{lemma}
\label{L1}In particular for $m=2$ the matrix $\Pi_{12}(u_{1},u_{2})$ may be
written as%
\begin{equation}
\Pi_{12}(\underline{u})=\pi_{1}\pi_{2}+f(u_{12})\mathbf{\mathring{C}}^{12}%
\bar{e}_{1}e_{2} \label{Pi2}%
\end{equation}
with $e_{2}=\mathbf{C}_{2a}\bar{e}^{a}$ ($e_{a}=\delta_{\alpha}^{1}$) and
$f(u)=-d(1-u)$. It satisfies the fundamental relation%
\begin{equation}
\mathring{R}_{12}(u_{12})\Pi_{12}(u_{1},u_{2})=\Pi_{21}(u_{2},u_{1}%
)R_{12}(u_{12})\,, \label{RPi2}%
\end{equation}
where $\mathring{R}_{12}$ is the $O(N-2)$ R-matrix.
\end{lemma}

\begin{proof}
Equation (\ref{Pi2}) can be easily derived. We calculate (\ref{Pi}) for $m=2$
with $u_{a}=u_{1}-1$
\begin{align*}
\Pi_{12}(\underline{u})  &  =\pi_{1}\pi_{2}\bar{e}_{a}T_{12,a}(\underline
{u},u_{a})\bar{e}^{a}\\
&  =\pi_{1}\pi_{2}\bar{e}_{a}R_{1a}(u_{1}-u_{a})R_{2a}(u_{2}-u_{a})\bar{e}%
^{a}\\
&  =\pi_{1}\pi_{2}\bar{e}_{a}\left(  \mathbf{1}_{1}\mathbf{1}_{a}%
+c(1)\mathbf{P}_{1a}\right)  \left(  \mathbf{1}_{2}\mathbf{1}_{a}%
+d(u_{21}+1)\mathbf{K}_{2a}\right)  \bar{e}^{a}\\
&  =\pi_{1}\pi_{2}+c(1)d(u_{21}+1)\mathbf{\mathring{C}}^{12}\bar{e}_{1}%
e_{2}\,.
\end{align*}
Use has been made of $c(1)=-1$ and $\pi_{1}\pi_{2}\bar{e}_{a}\mathbf{P}%
_{1a}\mathbf{K}_{2a}\bar{e}^{a}=\mathbf{\mathring{C}}^{12}\bar{e}_{1}e_{2}$.
Equation (\ref{RPi2}) is derived for all components. Obviously%
\[
\mathring{R}_{\mathring{\alpha}\mathring{\beta}}^{\mathring{\beta}^{\prime
}\mathring{\alpha}^{\prime}}(u)\mathbf{\mathring{C}}^{\mathring{\alpha
}\mathring{\beta}}=\mathring{R}_{0}(u)\mathbf{\mathring{C}}^{\mathring{\beta
}^{\prime}\mathring{\alpha}^{\prime}}%
\]
holds, where the scalar R-matrix eigenvalue is (see (\ref{EV}))%
\[
\mathring{R}_{0}(u)=a(u)+\left(  N-2\right)  \mathring{d}(u)\,.
\]
Therefore the relations
\[
\left(  \mathring{R}_{12}(u_{12})\Pi_{12}(u_{12})\right)  _{\alpha\beta
}^{\mathring{\beta}^{\prime}\mathring{\alpha}^{\prime}}=\mathring
{R}\,_{\mathring{\alpha}\mathring{\beta}}^{\mathring{\beta}^{\prime}%
\mathring{\alpha}^{\prime}}(u_{12})\pi_{\alpha}^{\mathring{\alpha}}\pi_{\beta
}^{\mathring{\beta}}+f(u_{12})\mathring{R}_{0}(u_{12})\mathbf{\mathring{C}%
}^{\mathring{\beta}^{\prime}\mathring{\alpha}^{\prime}}\delta_{\alpha}%
^{\bar{1}}\delta_{\beta}^{1}%
\]
and%
\[
\left(  \Pi_{21}(u_{21})R_{12}(u_{12})\right)  _{\alpha\beta}^{\mathring
{\beta}^{\prime}\mathring{\alpha}^{\prime}}=\pi_{\beta^{\prime}}%
^{\mathring{\beta}^{\prime}}\pi_{\alpha^{\prime}}^{\mathring{\alpha}^{\prime}%
}R_{\alpha\beta}^{\beta^{\prime}\alpha^{\prime}}(u_{12})+f(u_{21}%
)\mathbf{\mathring{C}}^{\mathring{\beta}^{\prime}\mathring{\alpha}^{\prime}%
}R_{\alpha\beta}^{\bar{1}1}(u_{12})
\]
are valid. The claim of the lemma is then equivalent to%
\[%
\begin{array}
[c]{lll}%
(i):\mathring{d}(u)=d(u)+f(-u)d(u) & \text{for} & \alpha\text{ or }\beta
\neq1,\bar{1}\\
(ii):0=d(u)+f(-u)\left(  1+d(u)\right)  & \text{for} & \alpha=1,~\text{ }%
\beta=\bar{1}\\
(iii):f(u)\mathring{R}_{0}(u)=d(u)+f(-u)\left(  c(u)+d(u)\right)  & \text{for}
& \alpha=\bar{1},~\text{ }\beta=1.
\end{array}
\]
These equations may be easily checked with the amplitudes (\ref{2.7}).
\end{proof}

The matrix $\Pi_{12}$ may be depicted as%
\[
\Pi_{\alpha\beta}^{\mathring{\alpha}\mathring{\beta}}(u_{1},u_{2})=%
\begin{array}
[c]{c}%
\raisebox{0.0075in}{\phantom{\rule{1.7984in}{0.91in}}}%
\end{array}%
\begin{array}
[c]{c}%
\unitlength3.6mm\begin{picture}(2.5,6)\thicklines \put(0,1){\line(0,1){4}} \put(-.27,2.7){$\bullet$} \put(2,1){\line(0,1){4}} \put(1.73,2.7){$\bullet$} \put(-.2,0){$\alpha$} \put(1.8,0){$\beta$} \put(-.2,5.4){$\mathring{\alpha}$} \put(1.8,5.4){$\mathring{\beta}$} \end{picture}
\end{array}
+f(u_{12})~~%
\begin{array}
[c]{c}%
\unitlength3.6mm\begin{picture}(3.5,6)\thicklines \put(0,1){\line(0,1){2}} \put(2,1){\line(0,1){2}} \put(1,4){\line(1,1){1}} \put(1,4){\line(-1,1){1}} \put(-.2,0){$\alpha$} \put(1.8,0){$\beta$} \put(-.2,3.3){$\bar1$} \put(1.8,3.3){$1$} \put(-.2,5.4){$\mathring{\alpha}$} \put(1.8,5.4){$\mathring{\beta}$} \end{picture}
\end{array}
\]
These results can be extended to general $m$, as presented below.

\begin{lemma}
\label{LM}The matrix $\Pi_{1\dots m}(\underline{u})$ satisfies

\begin{itemize}
\item[(a)] in addition to (\ref{Pi}) the recursion relation%
\begin{equation}
\Pi_{1\dots m}(\underline{u})=\left(  \Pi_{1\dots m-1}\pi_{m}\right)  \bar
{e}_{b}T_{1\dots m,b}(\underline{u},u_{b})\bar{e}^{b} \label{Pi'}%
\end{equation}
with $u_{b}=u_{m}-1/\nu+1$, and

\item[(b)] the fundamental relation%
\begin{equation}
\mathring{R}_{ij}(u_{ij})\Pi_{\dots ij\dots}(\underline{u})=\Pi_{\dots
ji\dots}(\underline{u})R_{ij}(u_{ij})\,. \label{RPi}%
\end{equation}

\item[(c)] The matrix $\bar{e}_{0}T_{1\dots m,0}(\underline{u},u_{0})\bar
{e}^{0}$ acts on $\Pi_{1\dots m}(\underline{u})$ as the unit matrix for
arbitrary $u_{0}$%
\begin{equation}
\Pi_{1\dots m}(\underline{u})\bar{e}_{0}T_{1\dots m,0}(\underline{u}%
,u_{0})\bar{e}^{0}=\Pi_{1\dots m}(\underline{u})\,. \label{PiT}%
\end{equation}

\item[(d)] Special components of $\Pi$ satisfy%
\begin{align}
\Pi_{1\alpha_{2}\dots\alpha_{m}}^{\mathring{\alpha}_{1}\dots\mathring{\alpha
}_{m}}(u_{1},\dots,u_{m})  &  =0\label{m1}\\
\Pi_{\mathring{\alpha}\alpha_{2}\dots\alpha_{m}}^{\mathring{\alpha}_{1}%
\dots\mathring{\alpha}_{m}}(u_{1},\dots,u_{m})  &  =\delta_{\mathring{\alpha}%
}^{\mathring{\alpha}_{1}}\Pi_{\alpha_{2}\dots\alpha_{m}}^{\mathring{\alpha
}_{2}\dots\mathring{\alpha}_{m}}(u_{2},\dots,u_{m})\label{m3}\\
\Pi_{\alpha_{1}\dots\alpha_{m-1}\bar{1}}^{\mathring{\alpha}_{1}\dots
\mathring{\alpha}_{m}}(u_{1},\dots,u_{m})  &  =0\label{m4}\\
\Pi_{\alpha_{1}\dots\alpha_{m-1}\mathring{\alpha}}^{\mathring{\alpha}_{1}%
\dots\mathring{\alpha}_{m}}(u_{1},\dots,u_{m})  &  =\Pi_{\alpha_{1}\dots
\alpha_{m-1}}^{\mathring{\alpha}_{1}\dots\mathring{\alpha}_{m-1}}(u_{1}%
,\dots,u_{m-1})\delta_{\mathring{\alpha}}^{\mathring{\alpha}_{m}}\,.
\label{m5}%
\end{align}
with $\mathring{\alpha}\neq1,\bar{1}$.
\end{itemize}
\end{lemma}

\noindent The proof of this Lemma is presented in appendix \ref{a1}.

The recursion relations (\ref{Pi}) and (\ref{Pi'}) can be rewritten as (see
also lemma \ref{L1} for $m=2$)%
\begin{align}
\Pi_{1\dots m}(\underline{u})  &  =\pi_{1}\Pi_{2\dots m}+\sum_{j=2}%
^{m}f(u_{1j})\mathring{R}_{jj-1}\cdots\mathring{R}_{j2}\mathbf{\mathring{C}%
}^{1j}\Pi_{2\dots\hat{\jmath}\dots m}\bar{e}_{1}e_{j}R_{jm}\cdots
R_{jj+1}\label{Pia}\\
&  =\Pi_{1\dots m-1}\pi_{m}+\sum_{j=1}^{m-1}f(u_{jm})\mathring{R}_{j+1j}%
\cdots\mathring{R}_{m-1j}\mathbf{\mathring{C}}^{jm}\Pi_{1\dots\hat{\jmath
}\dots m-1}\bar{e}_{j}R_{1j}\cdots R_{j-1j}e_{m}\, \label{Pia'}%
\end{align}
or in terms of pictures%
\begin{align*}%
\begin{array}
[c]{c}%
\unitlength4mm\begin{picture}(6,7.5)\thicklines \put(.5,5){\line(0,1){2.5}} \put(1.5,5){\line(0,1){2.5}} \put(4,5){\line(0,1){2.5}} \put(0,3){\framebox(4.5,2){$\Pi_m$}} \put(.5,0){\line(0,1){3}} \put(1.5,0){\line(0,1){3}} \put(4,0){\line(0,1){3}} \end{picture}
\end{array}
&  =~~~~%
\begin{array}
[c]{c}%
\unitlength4mm\begin{picture}(6,7.5)\thicklines \put(.5,5){\line(0,1){2.5}} \put(1.5,5){\line(0,1){2.5}} \put(4,5){\line(0,1){3}} \put(-.7,3.3){$\pi~\bullet$} \put(1,3){$\framebox(3.5,2){$\Pi_{m-1}$}$} \put(.5,0){\line(0,1){7.5}} \put(1.5,0){\line(0,1){3}} \put(4,0){\line(0,1){3}} \end{picture}
\end{array}
+~~%
{\displaystyle\sum\limits_{j=2}^{m}}
f(u_{1j})~~%
\begin{array}
[c]{c}%
\unitlength4mm\begin{picture}(6,7.5)\thicklines \put(1.5,5){\line(0,1){2.5}} \put(1.1,5.5){\line(1,1){1.8}} \put(1.1,5.5){\line(-1,2){1}} \put(4,5){\line(0,1){2.5}} \put(1,3){$\framebox(3.5,2){$\Pi_{m-2}$}$} \put(.1,3.6){$\bar1$} \put(.25,0){\line(0,1){3.5}} \put(1.5,0){\line(0,1){3}} \put(2.75,0){\line(0,1){1}} \put(2.75,1){\line(2,1){2.5}} \put(5.1,3.6){$1$} \put(5.25,2.25){\line(0,1){1}} \put(4,0){\line(0,1){3}} \put(-.3,6.5){$1$} \put(-.3,0){$1$} \put(.9,0){$2$} \put(2.8,6.5){$j$} \put(2.2,0){$j$} \put(4.4,6.5){$m$} \put(4.4,0){$m$} \end{picture}
\end{array}
\\%
\begin{array}
[c]{c}%
\unitlength4mm\begin{picture}(6,7.5)\thicklines \put(.5,5){\line(0,1){2.5}} \put(3,5){\line(0,1){2.5}} \put(4,5){\line(0,1){2.5}} \put(0,3){\framebox(4.5,2){$\Pi_m$}} \put(.5,0){\line(0,1){3}} \put(3,0){\line(0,1){3}} \put(4,0){\line(0,1){3}} \end{picture}
\end{array}
&  =~~%
\begin{array}
[c]{c}%
\unitlength4mm\begin{picture}(6,7.5)\thicklines \put(.5,5){\line(0,1){2.5}} \put(3,5){\line(0,1){2.5}} \put(3.8,3.3){$\bullet~\pi$} \put(0,3){$\framebox(3.5,2){$\Pi_{m-1}$}$} \put(.5,0){\line(0,1){3}} \put(3,0){\line(0,1){3}} \put(4,0){\line(0,1){7.5}} \end{picture}
\end{array}
+~~\sum_{j=1}^{m-1}f(u_{jm})%
\begin{array}
[c]{c}%
\unitlength4mm\begin{picture}(6,7.5)\thicklines \put(1.5,5){\line(0,1){2.5}} \put(4.5,5.5){\line(-1,1){1.8}} \put(4.5,5.5){\line(1,2){1}} \put(4,5){\line(0,1){2.5}} \put(1,3){$\framebox(3.5,2){$\Pi_{m-2}$}$} \put(.1,3.6){$\bar1$} \put(.25,2.25){\line(0,1){1}} \put(1.5,0){\line(0,1){3}} \put(2.75,0){\line(0,1){1}} \put(.1,3.6){$\bar1$} \put(2.75,1){\line(-2,1){2.5}} \put(5.4,3.6){$1$} \put(5.5,0){\line(0,1){3.5}} \put(4,0){\line(0,1){3}} \put(.8,6.5){$1$} \put(.8,0){$1$} \put(2.2,6.5){$j$} \put(2.1,0){$j$} \put(5.4,6.5){$m$} \put(4.5,0){$m$} \end{picture}
\end{array}
\end{align*}
%

{\phantom{\rule{3.8004in}{1.0056in}}}%
{\phantom{\rule{3.6691in}{0.9973in}}}%

In particular%
\begin{multline}
\Pi_{\bar{1}\alpha_{2}\dots\alpha_{m}}^{\mathring{\alpha}_{1}\dots
\mathring{\alpha}_{m}}(u_{1},\dots,u_{m})=\label{m6}\\
\sum_{j=2}^{m}f(u_{1j})\left(  \mathring{R}_{jj-1}\cdots\mathring{R}%
_{j2}\mathbf{\mathring{C}}^{1j}\Pi_{2\dots\hat{\jmath}\dots m}e_{j}%
R_{jm}\cdots R_{jj+1}\right)  _{\alpha_{2}\dots\alpha_{m}}^{\mathring{\alpha
}_{1}\dots\mathring{\alpha}_{m}}%
\end{multline}%
\begin{multline}
\Pi_{\alpha_{1}\dots\alpha_{m-1}1}^{\mathring{\alpha}_{1}\dots\mathring
{\alpha}_{m}}(u_{1},\dots,u_{m})=\label{m7}\\
\sum_{j=1}^{m-1}f(u_{jm})\left(  \mathring{R}_{j+1j}\cdots\mathring{R}%
_{m-1j}\mathbf{\mathring{C}}^{jm}\Pi_{1\dots\hat{\jmath}\dots m-1}\bar{e}%
_{j}R_{1j}\cdots R_{j-1j}\right)  _{\alpha_{1}\dots\alpha_{m-1}}%
^{\mathring{\alpha}_{1}\dots\mathring{\alpha}_{m}}\,.
\end{multline}

\section{The $O(N)$ - difference equation}

\label{s3} Let $K_{1\dots n}({\underline{u}})\in{V}_{1\dots n}$ be a co-vector
valued function of ${\underline{u}}=u_{1},\dots,u_{n}$ with values in
${V}_{1\dots n}\,$. The components of this vector are denoted by
\[
K_{\alpha_{1}\dots\alpha_{n}}({\underline{u}})~,~~(\alpha_{i}=1,2,\dots
,(0),\dots,\bar{2},\bar{1}).
\]
The following symmetry and periodicity properties of this function are
supposed to be valid:

\begin{condition}
\label{cond}~

\begin{itemize}
\item[\textrm{(i)}] The symmetry property under the exchange of two
neighboring spaces $V_{i}$ and $V_{j}$ and the variables $u_{i}$ and $u_{j}$,
at the same time, is given by
\begin{equation}
K_{\dots ij\dots}(\dots,u_{i},u_{j},\dots)=K_{\dots ji\dots}(\dots,u_{j}%
,u_{i},\dots)\tilde{R}_{ij}(u_{ij})\,, \label{3.1}%
\end{equation}
where $\tilde{R}(u)=R(u)/a(u)$ and $R(u)$ is the $O(N)$ R-matrix.

\item[\textrm{(ii)}] The \textbf{system of matrix difference equations} holds
\begin{equation}
\fbox{\rule[-3mm]{0cm}{8mm} $K_{1\dots n}(\dots,u_{i}^{\prime},\dots
)=K_{1\dots n}(\dots,u_{i},\dots)Q_{1\dots n}({\underline{u}};i)~,~~(i=1,\dots
,n)$ } \label{3.2}%
\end{equation}
with $u_{i}^{\prime}=u_{i}+2/\nu$. The matrix $Q_{1\dots n}({\underline{u}%
};i)\in End({V^{1\dots n}})$ is defined as the trace
\begin{equation}
Q_{1\dots n}({\underline{u}};i)=\operatorname*{tr}\nolimits_{0}\tilde
{T}_{Q,1\dots n,0}(\underline{u},i) \label{3.3}%
\end{equation}
of a modified monodromy matrix%
\[
\tilde{T}_{Q,1\dots n,0}(\underline{u},i)=\tilde{R}_{10}(u_{1}-u_{i}^{\prime
})\,\cdots\mathbf{P}_{i0}\cdots\tilde{R}_{n0}(u_{n}-u_{i})\,.
\]

\end{itemize}
\end{condition}

\noindent The Yang-Baxter equations for the R-matrix guarantee that these
properties are compatible. The shift of $2/\nu$ in eq.~(\ref{3.2}) could be
replaced by an arbitrary $\kappa$. For the application to the form factor
problem, however, it is fixed to be equal to $2/\nu$ in order to be compatible
with crossing symmetry. The properties (i) and (ii) may be depicted as
\begin{align*}
\mathrm{(i)}  &  ~~~~~~%
\begin{array}
[c]{c}%
\unitlength3mm\begin{picture}(8,4) \put(4,3){\oval(8,2)} \put(4,3){\makebox(0,0){$K$}} \put(1,0){\line(0,1){2}} \put(3,0){\line(0,1){2}} \put(5,0){\line(0,1){2}} \put(7,0){\line(0,1){2}} \put(1.5,1){$\scriptstyle \dots$} \put(5.5,1){$\scriptstyle \dots$} \end{picture}
\end{array}
~~=~~%
\begin{array}
[c]{c}%
\unitlength3mm\begin{picture}(8,4) \put(4,3){\oval(8,2)} \put(4,3){\makebox(0,0){$K$}} \put(1,0){\line(0,1){2}} \put(3,0){\line(1,1){2}} \put(5,0){\line(-1,1){2}} \put(7,0){\line(0,1){2}} \put(1.5,1){$\scriptstyle \dots$} \put(5.5,1){$\scriptstyle \dots$} \end{picture}~
\end{array}
,\\
\mathrm{(ii)}  &  ~~~~~~%
\begin{array}
[c]{c}%
\unitlength3mm\begin{picture}(8,5) \put(4,3){\oval(7,2)} \put(4,3){\makebox(0,0){$K$}} \put(2,0){\line(0,1){2}} \put(4,0){\line(0,1){2}} \put(6,0){\line(0,1){2}} \put(2.5,1){$\scriptstyle \dots$} \put(4.5,1){$\scriptstyle \dots$} \end{picture}
\end{array}
~~=~~%
\begin{array}
[c]{c}%
\unitlength3mm\begin{picture}(8,5) \put(4,3){\oval(7,2)} \put(4,3){\makebox(0,0){$K$}} \put(2,0){\line(0,1){2}} \put(6,0){\line(0,1){2}} \put(2.5,.5){$\scriptstyle \dots$} \put(4.5,.5){$\scriptstyle \dots$} \put(6,2){\oval(4,2)[b]} \put(4,2){\oval(8,6)[t]} \put(2,2){\oval(4,2)[lb]} \put(2,0){\oval(4,2)[rt]} \end{picture}
\end{array}
\end{align*}
with the graphical rule that a line changing the "time direction" changes the
spectral parameters $u\rightarrow u\pm1/\nu$ as follows
\[%
\begin{array}
[c]{c}%
\unitlength3mm\begin{picture}(12,2) \put(2,0){\oval(2,4)[t]} \put(9,2){\oval(2,4)[b]} \put(0,0){$\scriptstyle u$} \put(3.3,0){$\scriptstyle u-1/\nu$} \put(7,1){$\scriptstyle u$} \put(10.3,1){$\scriptstyle u+1/\nu$} \end{picture}~.
\end{array}
\]
Instead of the Yang-Baxter relation~(\ref{TTS}) the modified monodromy matrix
${\tilde{T}_{Q}}$ satisfies the Zapletal rules \cite{BKZ,BKZ2}. We have for
$i=1,\dots,n$
\begin{equation}
{\tilde{T}_{Q}}({\underline{u}};i)\,{T_{0}}({\underline{u}}^{\prime
},v)\,R_{i0}(u_{i}-v)=R_{i0}(u_{i}^{\prime}-v)\,{T_{0}}({\underline{u}%
},v)\,{\tilde{T}_{Q}}({\underline{u}};i) \label{3.8}%
\end{equation}
with ${\underline{u}}^{\prime}=u_{1},\dots,u_{i}^{\prime},\dots,u_{n}$ and
$u_{i}^{\prime}=u_{i}+2/\nu$. The $Q_{1\dots n}({\underline{u}};i)$ satisfy
the commutation rules
\begin{multline}
Q_{1\dots n}(\dots u_{i}\dots u_{j}\dots;i)\,Q_{1\dots n}(\dots u_{i}^{\prime
}\dots u_{j}\dots;j)\label{3.4}\\
=Q_{1\dots n}(\dots u_{i}\dots u_{j}\dots;j)\,Q_{1\dots n}(\dots u_{i}\dots
u_{j}^{\prime}\dots;i)\,.
\end{multline}

The following Proposition is obvious

\begin{proposition}
\label{p3.1} Let the vector valued function $K_{1\dots n}({\underline{u}}%
)\in{V}_{1\dots n}$ satisfy (i). Then for all $i=1,\dots,n$ the relations
(3.2) are equivalent to each other and also equivalent to the following
periodicity property under cyclic permutation of the spaces and the variables
\begin{equation}
K_{\alpha_{1}\alpha_{2}\dots\alpha_{n}}(u_{1}^{\prime},u_{2},\dots
,u_{n})=K_{\alpha_{2}\dots\alpha_{n}\alpha_{1}}(u_{2},\dots,u_{n},u_{1})\,.
\label{3.5}%
\end{equation}

\end{proposition}

\begin{remark}
The equations (\ref{3.1},\ref{3.5}) imply Watson's equations and crossing
relations for the form factors \cite{BFK}.
\end{remark}

Because of proposition \ref{p3.1} we mainly consider $Q_{1\dots n}%
({\underline{u},i})$ for $i=1$%
\begin{equation}
Q_{1\dots n}({\underline{u}})=\operatorname*{tr}\nolimits_{0}{\tilde{T}%
}_{Q,1\dots n,0}(\underline{u})=\prod_{k=2}^{n}\frac{1}{a(v_{ki}%
)}\operatorname*{tr}\nolimits_{0}{T}_{Q,1\dots n,0}(\underline{u})
\label{3.3'}%
\end{equation}
with $T_{Q,1\dots n,0}(\underline{u})=T_{Q,1\dots n,0}(\underline{u},1)$. In
analogy to eq.~(\ref{T1}) we introduce (suppressing the indices $1\dots n$)
\begin{equation}
{T}_{Q}({\underline{u}})\equiv\left(
\begin{array}
[c]{ccc}%
A_{Q}({\underline{u}}) & B_{Q}({\underline{u}}) & B_{Q,2}({\underline{u}})\\
C_{Q}({\underline{u}}) & D_{Q}({\underline{u}}) & B_{Q,3}({\underline{u}})\\
C_{Q,2}({\underline{u}}) & C_{Q,3}({\underline{u}}) & A_{Q,3}({\underline{u}})
\end{array}
\right)  . \label{3.9}%
\end{equation}

\subsection{The off-shell Bethe ansatz}

We will express the co-vector valued function $K_{\underline{\alpha}%
}({\underline{u}})$ in terms of the co-vectors
\begin{equation}
\Psi_{\underline{\alpha}}(\underline{u},\underline{v})=L_{\underline
{\mathring{\beta}}}(\underline{v})\,\Phi_{\underline{\alpha}}^{\underline
{\mathring{\beta}}}(\underline{u},\underline{v})=\left(  L(\underline{v}%
)\Phi(\underline{u},\underline{v})\right)  _{\underline{\alpha}}\,,
\label{Psi}%
\end{equation}
where summation over $\mathring{\beta}_{1},\dots,\mathring{\beta}%
_{m},~\mathring{\beta}_{i}=2,\dots,0,\dots,\bar{2}$ is assumed and
$L_{\underline{\mathring{\beta}}}(\underline{v})$ is a co-vector valued
function with values in ${\mathring{V}}_{1\dots m}\simeq\mathbb{C}%
^{N-2}\otimes\cdots\otimes\mathbb{C}^{N-2}$. We assume that for $L_{\mathring
{\beta}}(\underline{v})$ the higher level conditions of \ref{cond} hold with
$R$ and $Q$ replaced by $\mathring{R}$ and $\mathring{Q}$ (which means $N$ is
replaced by $N-2$)%
\begin{align}
(\mathrm{i})^{(1)}  &  :~~~L_{\dots ij\dots}(\dots,v_{i},v_{j},\dots)=L_{\dots
ji\dots}(\dots,v_{j},v_{i},\dots)\tilde{\mathring{R}}_{ij}(v_{ij}%
)\label{3.1a}\\
(\mathrm{ii})^{(1)}  &  :~~~L_{1\dots m}(\dots,v_{i}^{\prime},\dots)=L_{1\dots
n}(\dots,v_{i},\dots)\mathring{Q}_{1\dots m}({\underline{v}},i)\,.
\label{3.2a}%
\end{align}
with\footnote{Note that the shift $2/\nu$ is the same in the higher level
off-shell Bethe ansatz.} $v_{i}^{\prime}=v_{i}+2/\nu$.

The \textbf{Bethe ansatz states} are\footnote{The $\Phi_{\underline{\alpha}%
}^{\underline{\mathring{\beta}}}$ are generalizations of the states introduced
by Tarasov in \cite{Tar}.}%
\begin{equation}
\Phi_{\underline{\alpha}}^{\underline{\mathring{\beta}}}(\underline
{u},\underline{v})=\Pi_{\underline{\beta}}^{\underline{\mathring{\beta}}%
}(\underline{v})\left(  \Omega T_{1}^{\beta_{m}}(\underline{u},v_{m})\dots
T_{1}^{\beta_{1}}(\underline{u},v_{1})\right)  _{\underline{\alpha}}=~~~%
\begin{array}
[c]{c}%
{\phantom{\rule{1.6563in}{1.4751in}}}%
\unitlength4mm\begin{picture}(10,8.8)(0,.7)
\thicklines\put(3.8,1){$u_1$}
\put(8.3,1){$u_n$}
\put(6,3.8){$v_i$}
\put(5.0,6.2){1}
\put(8.0,6.2){1}
\put(9.5,1.8){1}
\put(9.5,3.2){1}
\put(9.5,4.8){1}
\put(5.8,5.5){$\dots$}
\put(.3,5.5){$_{\dots}$}
\put(1.6,5.5){$_{\dots}$}
\put(5,1){\line(0,1){5}}
\put(8,1){\line(0,1){5}}
\put(9,6){\oval(15.5,5)[lb]}
\put(9,6){\oval(18,8)[lb]}
\put(9,6){\oval(13,2)[lb]}
\put(-.2,6){$\framebox(3,1){$\Pi$}$}
\put(1.,8.6){$\underline{\mathring{\beta}}$}
\put(6,.7){$\underline{\alpha}$}
\put(0,7){\line(0,1){1}}
\put(1.25,7){\line(0,1){1}}
\put(2.5,7){\line(0,1){1}}
\end{picture}%
\end{array}
\,. \label{Phi}%
\end{equation}

\begin{remark}
The condition (\ref{3.1a}) implies the symmetry%
\begin{equation}
\Psi_{\underline{\alpha}}(\underline{u},\dots v_{i},v_{j}\dots)=\Psi
_{\underline{\alpha}}(\underline{u},\dots v_{j},v_{i}\dots)\,. \label{3.1c}%
\end{equation}

\end{remark}

The reference state $\Omega$ (\textquotedblleft
pseudo-vacuum\textquotedblright)\ is the highest weight co-vector (with
weights $w=(n,0,\dots,0)$)
\begin{equation}
\Omega_{\underline{\alpha}}=\delta_{\alpha_{1}}^{1}\cdots\delta_{\alpha_{n}%
}^{1}\,. \label{omega}%
\end{equation}
It satisfies
\begin{gather}
\Omega T({\underline{u},v})=\Omega\left(
\begin{array}
[c]{ccc}%
a_{1}(\underline{u},v) & 0 & 0\\
\ast & a_{2}(\underline{u},v) & 0\\
\ast & \ast & a_{3}(\underline{u},v)
\end{array}
\right)  \,,\label{1.40}\\
a_{1}(\underline{u},v)=\prod_{k=1}^{n}a(u_{k}-v)\,,~a_{2}(\underline
{u},v)=1\,,~a_{3}(\underline{u},v)=\prod_{k=1}^{n}\left(  1+d(u_{k}-v)\right)
\,.\nonumber
\end{gather}
We also have for $T_{Q}({\underline{u}})=T_{Q}({\underline{u},1})$%
\begin{equation}
\Omega T_{Q}({\underline{u}})=\Omega\prod_{k=2}^{n}a(u_{k1})\left(
\begin{array}
[c]{ccc}%
1 & 0 & 0\\
\ast & 0 & 0\\
\ast & 0 & 0
\end{array}
\right)  . \label{1.41}%
\end{equation}

The system of difference equations (\ref{3.2}) can be solved by means of a
nested \textquotedblleft off-shell" Bethe ansatz. The first level is given by
the \textbf{off-shell Bethe ansatz}%
\begin{equation}
\fbox{$\rule[-0.2in]{0in}{0.5in}\displaystyle~K_{\underline{\alpha}%
}(\underline{u})=\sum_{\underline{v}}\,g(\underline{u},\underline{v}%
)\,\Psi_{\underline{\alpha}}(\underline{u},\underline{v})$~,} \label{3.12}%
\end{equation}
where the state $\Psi$ is defined by (\ref{Psi}) and (\ref{Phi}) and the
scalar function $g(\underline{u},\underline{v})$ is%
\begin{equation}
g(\underline{u},\underline{v})=\prod_{i=1}^{n}\prod_{j=1}^{m}\psi(u_{i}%
-v_{j})\prod_{1\leq i<j\leq m}\tau(v_{i}-v_{j})\,. \label{3.15}%
\end{equation}
The functions $\psi(u)$ and $\tau(v)$ satisfy the functional equations%
\begin{equation}
\psi(u^{\prime})=a(u)\psi(u)\,,~\tau(v^{\prime})a(v^{\prime})=a(-v)\tau(v).
\label{3.17}%
\end{equation}
with $u^{\prime}=u+2/\nu$ The summation over {$\underline{v}$} is specified
by
\begin{equation}
{\underline{v}}=(v_{1},\dots,v_{m})=(\tilde{v}_{1}-2l_{1}/\nu,\dots,\tilde
{v}_{m}-2l_{m}/\nu)\,,~l_{i}\in\mathbf{Z\,}, \label{3.13}%
\end{equation}
where the $\tilde{v}_{i}$ are arbitrary constants.

\bigskip

\noindent The sums (\ref{3.12}) are also called \textquotedblleft Jackson-type
Integrals" (see e.g. \cite{Re} and references therein). Solutions of
(\ref{3.17}) are%
\begin{align}
\psi(u)  &  =\frac{\Gamma(-\frac{1}{2}\nu+\frac{u}{2}\nu)}{\Gamma(\frac{u}%
{2}\nu)}\label{3.18}\\
\tau(v)  &  =v\frac{\Gamma(\frac{1}{2}\nu+\frac{v}{2}\nu)}{\Gamma(1-\frac
{1}{2}\nu+\frac{v}{2}\nu)}\,. \label{3.20}%
\end{align}

We are now in a position to formulate the main result of this paper.

\begin{theorem}
\label{TN} Let the co-vector valued function $K_{1\dots n}({\underline{u}}%
)\in{{V_{1\dots n}}}$ be given by the Bethe ansatz (\ref{3.12}) and let
$g({\underline{x}},{\underline{u}})$ be of the form (\ref{3.15}). If in
addition the co-vector valued function ${{L}}_{1\dots m}(${$\underline{v}$%
}$)\in{{\mathring{V}_{1\dots m}}}$ satisfies the properties $(i)^{(1)}%
$\textrm{ }and\textrm{ }$(ii)^{(1)}$, i.e. equations (\ref{3.1}) and
(\ref{3.2}) for $O(N-2)$, then $K_{1\dots n}(\underline{u})$ satisfies the
equations (\ref{3.1}) and (\ref{3.2}) for $O(N)$, i.e. $K_{1\dots
n}(\underline{u})$ is a solution of the set of difference equations.
\end{theorem}

\noindent The proof of this theorem can be found in appendix \ref{a2}.

Iterating (\ref{3.12}), (\ref{Psi}) and theorem \ref{TN} we obtain the nested
off-shell Bethe ansatz with levels $k=1,\dots,\left[  \left(  N-1\right)
/2\right]  -1$. The ansatz for level $k$ reads
\begin{align}
K_{1\dots n_{k-1}}^{(k-1)}\left(  \underline{u}^{(k-1)}\right)   &
=\sum\limits_{\underline{u}^{(k)}}\,g^{(k-1)}(\underline{u}^{(k-1)}%
,\underline{u}^{(k)})\,\Psi_{1\dots n_{k-1}}^{(k-1)}(\underline{u}%
^{(k-1)},\underline{u}^{(k)})\label{3.19}\\
\Psi_{1\dots n_{k-1}}^{(k-1)}(\underline{u}^{(k-1)},\underline{u}^{(k)})  &
=\left(  K^{(k)}(\underline{u}^{(k)})\Phi^{(k-1)}(\underline{u}^{(k-1)}%
,\underline{u}^{(k)})\right)  _{1\dots n_{k-1}}\,,\nonumber
\end{align}
where $\Phi^{(k)}$ is the Bethe ansatz state (\ref{Phi}) and $g^{(k)}$ the
function (\ref{3.15}) for $O(N-2k)$. The highest levels differ from that of
theorem \ref{TN}, but they are given by the $O(3)$-problem for $N$ odd or the
$O(4)$-problem for $N$ even. These two case are investigated below.

\ifcase0\or

\begin{anote}
The start of the iteration may be given by a $k_{max}$ with $1\leq
k_{max}<\left[  \left(  N-1\right)  /2\right]  $ where $\Psi^{(k_{max}-1)}$
trivial%
\begin{equation}
\Psi_{\alpha_{1}\dots\alpha_{n_{k-1}}}^{(k_{max}-1)}=\delta_{\alpha_{1}%
}^{k_{max}}\dots\delta_{\alpha_{n_{k-1}}}^{k_{max}}~,~\text{and }%
n_{k}=0~\text{for }k\geq k_{max} \label{3.21}%
\end{equation}
which is the reference state of level $k_{max}-1$ and trivially fulfills the
Conditions \ref{cond}. If no state $\Psi^{(k-1)}$ is trivial the nesting
starts for $k=\left[  \left(  N-1\right)  /2\right]  $ with the $O(3)$-problem
for $N$ odd or the $O(4)$-problem for $N$ even. These two case are
investigated below.
\end{anote}

\begin{corollary}
\label{c3.2 copy(1)} The system of $O(N)$ matrix difference equations
(\ref{3.2}) is solved by the nested Bethe ansatz (\ref{3.19}) with
(\ref{3.21}) and ${K}_{1\dots n}({\underline{u}})={K^{(0)}}_{1\dots
n}({\underline{u}})\,$.
\end{corollary}

\fi

\begin{corollary}
\label{c3.2} The system of $O(N)$ matrix difference equations (\ref{3.2}) is
solved by the nested Bethe ansatz (\ref{3.19}) with ${K}_{1\dots
n}({\underline{u}})={K^{(0)}}_{1\dots n}({\underline{u}})\,$.
\end{corollary}

\subsubsection{The off-shell Bethe ansatz for $O(3)$}

The $O(3)$ R-matrix is
\[
R(u)=\mathbf{1}+c(u)\,P+d(u)\,K,~~c(u)=\frac{-1}{u},~d(u)=\frac{1}{u-1/2}\,.
\]
The solution of the difference equations (\ref{3.1})-(\ref{3.3}) is again
given by the off-shell Bethe ansatz (\ref{3.12})-(\ref{3.20}). The Bethe
vector $\Psi$ is expressed in terms of the co-vectors (\ref{Phi})
\[
\Psi_{\underline{\alpha}}(\underline{u},\underline{v})=L(\underline{v}%
)\,\Phi_{\underline{\alpha}}(\underline{u},\underline{v})\,,
\]
where the scalar function $L(\underline{v})$ has to satisfy%
\begin{equation}%
\begin{array}
[c]{c}%
L(\dots,v_{i},v_{j},\dots)=L(\dots,v_{j},v_{i},\dots)\tilde{\mathring{R}%
}(v_{ij})\\
L(v_{1}^{\prime},v_{2},\dots,v_{m})=L(v_{2},\dots,v_{m},v_{1})
\end{array}
\label{3.25}%
\end{equation}
with $\tilde{\mathring{R}}(v)=\frac{\left(  v+1\right)  \left(  v-1/2\right)
}{\left(  v-1\right)  \left(  v+1/2\right)  }$. For $N=3$ i.e. $\nu=2$ we have%
\[
\psi(u)=\frac{1}{u-1}\,,~\tau(v)=v^{2}.
\]
and $v^{\prime}=v+1$. The the minimal solution of the equations (\ref{3.25})
is%
\begin{align*}
L(\underline{v})  &  =\prod\limits_{1\leq i<j\leq m}L(v_{ij})\\
L(v)  &  =\frac{\pi}{4}\frac{\left(  v-1/2\right)  }{v\left(  v-1\right)
}\tan\pi v\,.
\end{align*}
The $O(3)$ weight of the state $\Psi_{\underline{\alpha}}(\underline
{u},\underline{v})$ is%
\begin{equation}
w=n-m\,. \label{w3}%
\end{equation}

\subsubsection{The off-shell Bethe ansatz for $O(4)$}

The $O(4)$ R-matrix is
\[
R(u)=\mathbf{1}+c(u)\,P+d(u)\,K,~~c(u)=\frac{-1}{u},~~d(u)=\frac{1}{u-1}\,.
\]
We could apply theorem \ref{TN} and write the off-shell Bethe ansatz for
$O(4)$ in terms of an $O(2)$ problem. However, the latter cannot be solved by
the Bethe ansatz because the R-matrix is diagonal (note that $R_{1\bar{1}%
}^{1\bar{1}}=0$). But there is another way to solve the $O(4)$ problem. The
group isomorphism $O(4)\simeq SU(2)\otimes SU(2)$ reflects itself in terms of
the corresponding R-matrices. Indeed, the $O(4)$ R-matrix can be written as a
tensor product of two $SU(2)$ R-matrices, or more precisely%
\begin{gather*}
\left(  \tilde{R}^{O(4)}\right)  _{\alpha\beta}^{\delta\gamma}\Gamma
_{AB}^{\alpha}\Gamma_{CD}^{\beta}=\Gamma_{C^{\prime}D^{\prime}}^{\delta}%
\Gamma_{A^{\prime}B^{\prime}}^{\gamma}\left(  \tilde{R}_{+}^{SU(2)}\right)
_{AC}^{C^{\prime}A^{\prime}}\left(  \tilde{R}_{-}^{SU(2)}\right)
_{BD}^{D^{\prime}B^{\prime}}\\%
\begin{array}
[c]{c}%
\unitlength1mm\begin{picture}(30,23)(3,0) \put(13,12){\line(1,1){9}} \put(22,12){\line(-1,1){9}} \put(13,12){\line(0,-1){8}} \put(13,12){\line(-2,-1){8}} \put(22,12){\line(0,-1){9}} \put(22,12){\line(5,-2){8}} \put(2,4){$A$} \put(11,0){$B$} \put(31,4){$D$} \put(21,0){$C$} \put(11,14){$\alpha$} \put(22,14){$\beta$} \put(11,22){$\delta$} \put(22,22){$\gamma$} \end{picture}
\end{array}
=%
\begin{array}
[c]{c}%
\unitlength1mm\begin{picture}(27,26)(0,0) \put(5,6){\line(1,1){12}} \put(17,18){\line(1,0){5}} \put(22,18){\line(0,-1){5}} \put(22,13){\line(-1,-1){12}} \put(17,1){\line(-1,1){12}} \put(5,13){\line(0,1){5}} \put(5,18){\line(1,0){5}} \put(10,18){\line(1,-1){12}} \put(22,18){\line(1,1){5}} \put(5,18){\line(-1,1){5}} \put(2,4){$A$} \put(7,0){$B$} \put(23,4){$D$} \put(19,0){$C$} \put(7,19){$D'$} \put(19,19){$A'$} \put(.8,13){$C'$} \put(23,13){$B'$} \put(0,24){$\delta$} \put(27,24){$\gamma$} \put(7.7,8.4){$\bullet$} \put(17.7,8.5){$\bullet$} \end{picture}
\end{array}
.
\end{gather*}
The $SU(2)$ R-matrices $\tilde{R}_{\pm}^{SU(2)}(u)=R_{\pm}^{SU(2)}(u)/a(u)$
correspond to the spinor representations of $O(4)$ with positive (negative)
chirality%
\[
R_{\pm}^{SU(2)}=\mathbf{1}+c(u)\,\mathbf{P\,,}%
\]
where the amplitude $c(u)=-1/u$ is again given by (\ref{2.7}). The relative
R-matrix for states of different chirality is trivial $\tilde{R}=\mathbf{1}$.
The intertwiners are%
\[
\Gamma_{AB}^{\alpha}=\left(  \gamma_{+}\gamma^{\alpha}\mathbf{C}\right)  _{AB}%
\]
with the $O(4)$ gamma matrices $\gamma^{\alpha}$, $\gamma_{+}=\tfrac{1}%
{2}\left(  1+\gamma^{5}\right)  $ and the charge conjugation matrix
$\mathbf{C}$. For more details see \cite{ShWi,KT1}. In the complex basis of
the $O(4)$ and the fundamental representations of the $SU(2)$ the states have
the $O(4)$ weights
\begin{equation}%
\begin{tabular}
[c]{cc}%
vector states & $O(4)$ weights\\\hline
$1$ & $(1,0)$\\
$2$ & $(0,1)$\\
$\bar{2}$ & $(0,-1)$\\
$\bar{1}$ & $(-1,0)$%
\end{tabular}
\ \ \ \ \ \ \ \ \qquad%
\begin{tabular}
[c]{cc}%
spinor states & $O(4)$ weights\\\hline
$\uparrow_{+}$ & $(\frac{1}{2},\frac{1}{2})$\\
$\downarrow_{+}$ & $(-\frac{1}{2},-\frac{1}{2})$\\
$\uparrow_{-}$ & $(\frac{1}{2},-\frac{1}{2})$\\
$\downarrow_{-}$ & $~~(-\frac{1}{2},\frac{1}{2})~.$%
\end{tabular}
\ \ \ \ \ \ \ \ \label{spin}%
\end{equation}
Because of weight conservation the intertwiner matrix is diagonal in this
basis and is calculated to be%
\begin{equation}
\Gamma_{AB}^{\alpha}=\left(
\begin{array}
[c]{cccc}%
\Gamma_{\uparrow_{+}\uparrow_{-}}^{1} & 0 & 0 & 0\\
0 & \Gamma_{\uparrow_{+}\downarrow_{-}}^{2} & 0 & 0\\
0 & 0 & \Gamma_{\downarrow_{+}\uparrow_{-}}^{\bar{2}} & 0\\
0 & 0 & 0 & \Gamma_{\downarrow_{+}\downarrow_{-}}^{\bar{1}}%
\end{array}
\right)  =\left(
\begin{array}
[c]{cccc}%
-1 & 0 & 0 & 0\\
0 & 1 & 0 & 0\\
0 & 0 & 1 & 0\\
0 & 0 & 0 & 1
\end{array}
\right)  . \label{int}%
\end{equation}
We also use the dual intertwiner $\Gamma_{\alpha}^{AB}$ with%
\begin{equation}
\sum_{A,B}\Gamma_{AB}^{\alpha^{\prime}}\Gamma_{\alpha}^{AB}=\delta_{\alpha
}^{\alpha^{\prime}}\,,~\sum_{\alpha}\Gamma_{\alpha}^{A^{\prime}B^{\prime}%
}\Gamma_{AB}^{\alpha}=\delta_{A}^{A^{\prime}}\delta_{B}^{B^{\prime}}.
\label{oc}%
\end{equation}
We write the co-vector valued function $K_{\underline{\alpha}}(\underline{u})$
as%
\begin{align}
K_{\underline{\alpha}}(\underline{u})  &  =K_{\underline{A}}^{(+)}%
(\underline{u})K_{\underline{B}}^{(-)}(\underline{u})\Gamma_{\underline
{\alpha}}^{\underline{A}\,\underline{B}}\label{K4}\\%
\begin{array}
[c]{c}%
\unitlength1mm\begin{picture}(17,17)(0,0) \put(8,13){\oval(17,6)[]} \put(3,1){\line(0,1){9}} \put(13,1){\line(0,1){9}} \put(8,13){\makebox(0,0)[cc]{$K$}} \put(8,2){\makebox(0,0)[cc]{$\underline\alpha$}} \end{picture}
\end{array}
&  =%
\begin{array}
[c]{c}%
\unitlength1.5mm\begin{picture}(26,14)(0,3) \put(5.5,14){\oval(9,4)[]} \put(21.5,14){\oval(9,4)[]} \put(3,12){\line(4,-3){8}} \put(8,12){\line(4,-3){8}} \put(11,6){\line(4,3){8}} \put(16,6){\line(4,3){8}} \put(11,3){\line(0,1){3}} \put(16,3){\line(0,1){3}} \put(5.5,14){\makebox(0,0)[cc]{$K^{(+)}$}} \put(21.5,14){\makebox(0,0)[cc]{$K^{(-)}$}} \put(9.5,9){\makebox(0,0)[cc]{$\underline A$}} \put(17.6,9){\makebox(0,0)[cc]{$\underline B$}} \put(13.6,3){\makebox(0,0)[cc]{$\underline\alpha$}} \end{picture}
\end{array}
\,,
\end{align}
where $\Gamma_{\underline{\alpha}}^{\underline{A}\,\underline{B}}=\prod
_{i=1}^{n}\Gamma_{\alpha_{i}}^{A_{i}B_{i}}$. The transfer matrix
$\operatorname*{tr}\tilde{T}^{O(4)}(\underline{u},v)$ can also be decomposed
such that%
\begin{align*}
K_{\underline{\alpha}}(\underline{u})\left(  \tilde{T}^{O(4)}\right)
_{\gamma}^{\gamma}(\underline{u},v)  &  =K_{\underline{\alpha}}(\underline
{u})\Gamma_{A^{\prime}B^{\prime}}^{\gamma}\left(  \tilde{T}_{+}^{SU(2)}%
\right)  _{A}^{A^{\prime}}(\underline{u},v)\left(  \tilde{T}_{-}%
^{SU(2)}\right)  _{B}^{B^{\prime}}(\underline{u},v)\Gamma_{\gamma}^{AB}\\
&  =\left(  K_{\underline{A}}^{(+)}(\underline{u})\left(  \tilde{T}%
_{+}^{SU(2)}\right)  _{A}^{A}(\underline{u},v)\right)  \left(  K_{\underline
{B}}^{(-)}(\underline{u})\left(  \tilde{T}_{-}^{SU(2)}\right)  _{B}%
^{B}(\underline{u},v)\right)  \Gamma_{\underline{\alpha}}^{\underline
{A}\,\underline{B}}\\%
\begin{array}
[c]{c}%
\unitlength1mm\begin{picture}(17,17)(0,0) \put(8,13){\oval(17,6)[]} \put(3,1){\line(0,1){9}} \put(13,1){\line(0,1){9}} \put(8,13){\makebox(0,0)[cc]{$K$}} \put(1,5){\line(1,0){13}} \put(-1,5){$\gamma$} \put(16,5){$\gamma$} \end{picture}
\end{array}
~~~~  &  =~~~%
\begin{array}
[c]{c}%
\unitlength1.5mm\begin{picture}(26,14)(0,3) \put(5.5,14){\oval(9,4)[]} \put(21.5,14){\oval(9,4)[]} \put(3,12){\line(4,-3){8}} \put(8,12){\line(4,-3){8}} \put(11,6){\line(4,3){8}} \put(16,6){\line(4,3){8}} \put(11,3){\line(0,1){3}} \put(16,3){\line(0,1){3}} \put(5.5,14){\makebox(0,0)[cc]{$K^{(+)}$}} \put(21.5,14){\makebox(0,0)[cc]{$K^{(-)}$}} \put(0,9){$A$} \put(11.2,10){$A$} \put(14,10){$B$} \put(25,9){$B$} \put(3,10.5){\line(1,0){8}} \put(16,10.5){\line(1,0){8}} \end{picture}
\end{array}
~~\,,
\end{align*}
where (\ref{oc}) has been used. Therefore $K_{\underline{\alpha}}%
(\underline{u})$ satisfies the $O(4)$ symmetry relation (\ref{3.1}) and the
difference equation (\ref{3.2}) if the $K_{\underline{A}}^{(\pm)}%
(\underline{u})$ satisfy the corresponding $SU(2)$ relations.

The $SU(2)$ on-shell Bethe ansatz is well known and the off-shell case has
been solved in \cite{BKZ,BFK1,BFK3}
\begin{align}
K_{\underline{A}}(\underline{u})  &  =\sum_{\underline{v}}\,g(\underline
{u},\underline{v})\,\Psi_{\underline{A}}(\underline{u},\underline
{v})\label{KSU2}\\
\Psi(\underline{u},\underline{v})  &  =\Omega C(\underline{u},v_{m})\dots
C(\underline{u},v_{1})\,, \label{psiSU2}%
\end{align}
where $\sum_{\underline{v}}$ and $g(\underline{u},\underline{v})$ are given by
(\ref{3.15})-(\ref{3.20}). For $N=4$ i.e. $\nu=1$ we have%
\[
\psi(u)=\frac{\Gamma(-\frac{1}{2}+\frac{u}{2})}{\Gamma(\frac{u}{2})}%
\,,~\tau(v)=v\,.
\]
The $SU(2)$ weights of the state (\ref{psiSU2}) are $w=(n-m,m)$ and due to
(\ref{spin}) the $O(4)$ weights are%
\[
w=\left\{
\begin{array}
[c]{lll}%
(n-m,-m) &  & \text{for positive chirality spinors}\\
(n-m,m) &  & \text{for negative chirality spinors\thinspace.}%
\end{array}
\right.
\]
Therefore the $O(4)$ weights of (\ref{K4}) are (see also \cite{dVK})
\begin{equation}
w=(n-n_{-}-n_{+},n_{-}-n_{+})\,, \label{w4}%
\end{equation}
where $n_{\pm}$ are the numbers of positive (negative) chirality $C$-operators.

\section{Weights of off-shell $O(N)$ Bethe vectors}

\label{s4}

In this section we analyze some group theoretical properties of off-shell
Bethe states. We show that they are highest weight states and we calculate the
weights. The first result is not only true for the conventional Bethe ansatz,
which solves an eigenvalue problem and which is well known, but it is also
true, as we will show, for the off-shell one which solves a difference
equation (or a differential equation).

By the asymptotic expansion of the R-matrix (\ref{2.5}) and the monodromy
matrix (\ref{T}) we get for $u\rightarrow\infty$%
\begin{align}
R_{ab}(u)  &  =\mathbf{1}_{ab}-\frac{1}{u}M_{ab}+O(u^{-2})\label{4.1}\\
M_{ab}  &  =\mathbf{P}_{ab}-\mathbf{K}_{ab}\,.
\end{align}
More explicitly eq.~(\ref{T}) gives%
\begin{align}
{T_{1\dots n,a}}({\underline{u}},u)  &  =\mathbf{1}_{1\dots n,a}+\frac{1}%
{u}\,M_{1\dots n,a}+O(u^{-2})\label{4.3}\\
M_{1\dots n,a}  &  =\left(  \mathbf{P}_{1a}-\mathbf{K}_{1a}\right)
+\dots+\left(  \mathbf{P}_{na}-\mathbf{K}_{na}\right)  \,.
\end{align}
The matrix elements of $M_{1\dots n,a}$, as a matrix in the auxiliary space,
are the $O(N)$ Lie algebra generators. In the following we will consider only
operators acting in the fixed tensor product space ${V^{1\dots n}}$ of
(\ref{2.1}). Therefore we will omit the indices $1\dots n$. In terms of the
matrix elements in the auxiliary space $V_{a}$ the generators act on the basis
states as
\begin{equation}
\langle\,\alpha_{1},\dots,\alpha_{i},\dots,\alpha_{n}\,|\,M_{\alpha}%
^{\alpha^{\prime}}=\sum_{i=1}^{n}\Big(\delta_{\alpha\alpha_{i}}\,\langle
\,\alpha_{1},\dots,\alpha^{\prime},\dots,\alpha_{n}\,|-\delta_{\alpha^{\prime
}\bar{\alpha}_{i}}\,\langle\,\alpha_{1},\dots,\bar{\alpha},\dots,\alpha
_{n}\,|\Big). \label{4.4}%
\end{equation}
The diagonal elements of $M_{\alpha}^{\alpha^{\prime}}$ are the the weight
operators $W_{\alpha}$ with%
\begin{align*}
\langle\,\alpha_{1},\dots,\alpha_{i},\dots,\alpha_{n}\,|\,W_{\alpha}  &
=w_{\alpha}\langle\,\alpha_{1},\dots,\alpha_{i},\dots,\alpha_{n}\,|\\
w_{\alpha}  &  =n_{\alpha}-n_{\bar{\alpha}}%
\end{align*}
where $n_{\alpha}$ is the number of particles $\alpha$ in the state. It is
sufficient to consider only the weights%
\begin{equation}
w=\left(  w_{1},\dots,w_{\left[  N/2\right]  }\right)  \label{w}%
\end{equation}
because of $W_{\alpha}=-W_{\bar{\alpha}}\,$and $\langle\,\underline{\alpha
}\,|W_{0}=0$ for $N$ odd.

The Yang-Baxter relations (\ref{TTS}) yield for $u_{a}\rightarrow\infty$
\begin{equation}
\lbrack M_{a}+M_{ab},T_{b}(u_{b})]=0 \label{4.5}%
\end{equation}
and if additionally $u_{b}\rightarrow\infty$, we get
\begin{equation}
\lbrack M_{a}+M_{ab},M_{b}]=0, \label{4.6}%
\end{equation}
or for the matrix elements (in the real basis)%
\begin{align}
\lbrack M_{\alpha}^{\alpha^{\prime}},T_{\beta}^{\beta^{\prime}}(u)]  &
=-\delta_{\alpha}^{\beta^{\prime}}T_{\beta}^{\alpha^{\prime}}(u)+\delta
^{\alpha^{\prime}\beta^{\prime}}T_{\beta}^{\alpha}(u)+T_{\alpha}%
^{\beta^{\prime}}(u)\delta_{\beta}^{\alpha^{\prime}}-T_{\alpha^{\prime}%
}^{\beta^{\prime}}(u)\delta_{\alpha\beta}\label{4.7}\\
\lbrack M_{\alpha}^{\alpha^{\prime}},M_{\beta}^{\beta^{\prime}}]  &
=-\delta_{\alpha}^{\beta^{\prime}}M_{\beta}^{\alpha^{\prime}}+\delta
^{\alpha^{\prime}\beta^{\prime}}M_{\beta}^{\alpha}+M_{\alpha}^{\beta^{\prime}%
}\delta_{\beta}^{\alpha^{\prime}}-M_{\alpha^{\prime}}^{\beta^{\prime}}%
\delta_{\alpha\beta}. \label{4.8}%
\end{align}
Equation (\ref{4.8}) represents the structure relations of the $O(N)$ Lie
algebra and (\ref{4.7}) the $O(N)$-covariance of $T$. In particular the
transfer matrix is invariant
\begin{equation}
\lbrack M_{\alpha}^{\alpha^{\prime}},\operatorname*{tr}T(u)]=0. \label{4.9}%
\end{equation}

\begin{theorem}
\label{HW}

\begin{enumerate}
\item If the co-vector valued function
\[
K_{\underline{\alpha}}(\underline{u})=\sum_{\underline{v}}\,g(\underline
{u},\underline{v})\,L_{\underline{\mathring{\beta}}}(\underline{v}%
)\,\Phi_{\underline{\alpha}}^{\underline{\mathring{\beta}}}(\underline
{u},\underline{v})
\]
is given by the nested off-shell Bethe ansatz (\ref{3.19}) the weights
(\ref{w}) are $w=$%
\[
(w_{1},\dots,w_{\left[  N/2\right]  })=\left\{
\begin{array}
[c]{lll}%
\left(  n-n_{1},\dots,n_{\left[  N/2\right]  -1}-n_{\left[  N/2\right]
}\right)  & \text{for} & N~\text{odd}\\
\left(  n-n_{1},\dots,n_{\left[  N/2\right]  -2}-n_{-}-n_{+},n_{-}%
-n_{+}\right)  & \text{for} & N~\text{even\thinspace.}%
\end{array}
\right.
\]

\item If $K_{\underline{\alpha}}(\underline{u})$ satisfies the conditions of
theorem \ref{TN} and if $L_{\underline{\mathring{\beta}}}(\underline{v})$ is a
highest weight state, then $K_{\underline{\alpha}}(\underline{u})$ is a
highest weight state:%
\[
K(\underline{u})M_{\alpha}^{\alpha^{\prime}}=0\,~\text{for }\alpha^{\prime
}<\alpha\,.
\]
(Recall that $\alpha^{\prime}<\alpha$ is to be understood corresponding to the
ordering\newline$1,2,\dots,\left[  N/2\right]  ,(0),\overline{\left[
N/2\right]  },\dots,\bar{2},\bar{1}$ .)

\item The weights satisfy the highest weight condition%
\[
\left\{
\begin{array}
[c]{lll}%
w_{1}\geq w_{2}\geq\dots\geq w_{\left[  N/2\right]  }\geq0 & \text{for} &
N~\text{odd}\\
w_{1}\geq w_{2}\geq\dots\geq|w_{\left[  N/2\right]  }| & \text{for} &
N~\text{even.}%
\end{array}
\right.
\]

\end{enumerate}
\end{theorem}

The proof of this theorem can be found in appendix \ref{a5}. We mention that
for $N$ even the highest weight property was already discussed in appendix B
of \cite{dVK}.

\section{Conclusion}

\label{s6}

In this article we solved the $O(N)$ -matrix difference equations by means of
the off-shell algebraic nested Bethe ansatz. We introduced a new object called
$\Pi$-matrix to overcome the difficulties connected to the special
peculiarities of the $O(N)$ symmetric R-matrix structure. The highest weights
properties of\ the solutions were analyzed. We believe that our construction
can also be applied to the cases with similar group theoretical complexities,
such as $B_{n}$, $C_{n}$, $D_{n}$ Lie algebras and superalgebra $Osp(n|2m)$
(see \cite{MR}).

\paragraph{\textbf{Acknowledgment:}}

The authors have profited from discussions with A.~Fring, R. Schra\-der,
F.~Smirnov and A.~Belavin. H.B. thanks R.~Flume, R.~Poghossian and P.~Wiegmann
for valuable discussions. H.B. and M.K. were supported by Humboldt Foundation
and H.B. is also supported by the Armenian grant 11-1\_c028. A.F. acknowledges
support from DAAD (Deutscher Akademischer Austausch Dienst) and CNPq (Conselho
Nacional de Desenvolvimento Cient\'{\i}fico e Tecnol\'{o}gico).

\appendix

\section*{Appendix}

\renewcommand{\theequation}{\mbox{\Alph{section}.\arabic{equation}}} \setcounter{equation}{0}

\section{Proof of Lemma
\protect\ref{LM}%
}

\label{a1}

\begin{proof}
\begin{itemize}
\item[(a)] We prove (\ref{Pi'}) by induction: It is true for $m=2$, because
similar to (\ref{Pi2})%
\[
\pi_{1}\pi_{2}\bar{e}_{b}T_{12,b}(\underline{u},u_{b})\bar{e}^{b}=\pi_{1}%
\pi_{2}+f(u_{12})\mathbf{\mathring{C}}^{12}\bar{e}_{1}e_{2}\,.
\]
We assume (\ref{Pi'}) for $m-1$ and replace in the definition (\ref{Pi})
$\Pi_{2\dots m}$ as given by (\ref{Pi'})%
\begin{align*}
\Pi_{1\dots m}  &  =\pi_{1}\Pi_{2\dots m}\bar{e}_{a}T_{1\dots m,a}\bar{e}%
^{a}\\
&  =\pi_{1}\left(  \Pi_{2\dots m-1}\pi_{m}\bar{e}_{b}T_{2\dots m,b}\bar{e}%
^{b}\right)  \bar{e}_{a}T_{1\dots m,a}\bar{e}^{a}\\
&  =\left(  \pi_{1}\Pi_{2\dots m-1}\pi_{m}\right)  \bar{e}_{b}T_{1\dots
m,b}\bar{e}^{b}\bar{e}_{a}T_{1\dots m,a}\bar{e}^{a}\\
&  =\left(  \pi_{1}\Pi_{2\dots m-1}\pi_{m}\right)  \bar{e}_{a}T_{1\dots
m,a}\bar{e}^{a}\bar{e}_{b}T_{1\dots m,b}\bar{e}^{b}\\
&  =\left(  \pi_{1}\Pi_{2\dots m-1}\pi_{m}\right)  \bar{e}_{a}T_{1\dots
m-1,a}\bar{e}^{a}\bar{e}_{b}T_{1\dots m,b}\bar{e}^{b}\\
&  =\left(  \Pi_{1\dots m-1}\pi_{m}\right)  \bar{e}_{b}T_{1\dots m,b}\bar
{e}^{b}\,.
\end{align*}
for $u_{a}=u_{1}-1$ and $u_{b}=u_{m}-1/\nu+1$. Going from equality 2 to
equality 3
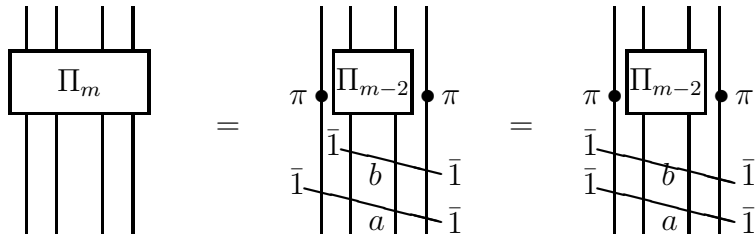
\[%
\begin{array}
[c]{c}%
\unitlength4mm\begin{picture}(6,7.5)\thicklines \put(.5,6){\line(0,1){1.5}} \put(1.5,6){\line(0,1){1.5}} \put(3,6){\line(0,1){1.5}} \put(4,6){\line(0,1){1.5}} \put(0,4){\framebox(4.5,2){$\Pi_m$}} \put(.5,0){\line(0,1){4}} \put(1.5,0){\line(0,1){4}} \put(3,0){\line(0,1){4}} \put(4,0){\line(0,1){4}} \end{picture}
\end{array}%
{\phantom{\rule{3.4255in}{1.2393in}}}%
=~~~~%
\begin{array}
[c]{c}%
\unitlength4mm\begin{picture}(6,7.5)\thicklines \put(1.5,6){\line(0,1){1.5}} \put(3,6){\line(0,1){1.5}} \put(-.5,4.3){$\pi\,\bullet$} \put(3.8,4.3){$\bullet\,\pi$} \put(1,4){$\framebox(2.5,2){$\Pi_{m-2}$}$} \put(.55,0){\line(0,1){7.5}} \put(4.,0){\line(0,1){7.5}} \put(1.5,0){\line(0,1){4}} \put(3,0){\line(0,1){4}} \put(4,0){\line(0,1){4}} \put(1.2,2.8){\line(4,-1){3.3}} \put(.7,2.6){$\bar1$} \put(4.7,1.6){$\bar1$} \put(2.1,1.6){$b$} \put(0,1.5){\line(4,-1){4.5}} \put(-.5,1.3){$\bar1$} \put(4.7,0){$\bar1$} \put(2.1,.1){$a$} \end{picture}
\end{array}
=~~~~%
\begin{array}
[c]{c}%
\unitlength4mm\begin{picture}(6,7.5)\thicklines \put(1.5,6){\line(0,1){1.5}} \put(3,6){\line(0,1){1.5}} \put(-.5,4.3){$\pi\,\bullet$} \put(3.8,4.3){$\bullet\,\pi$} \put(1,4){$\framebox(2.5,2){$\Pi_{m-2}$}$} \put(.55,0){\line(0,1){7.5}} \put(4.,0){\line(0,1){7.5}} \put(1.5,0){\line(0,1){4}} \put(3,0){\line(0,1){4}} \put(4,0){\line(0,1){4}} \put(0,2.8){\line(4,-1){4.5}} \put(-.5,2.6){$\bar1$} \put(4.7,1.6){$\bar1$} \put(2.1,1.6){$b$} \put(0,1.5){\line(4,-1){4.5}} \put(-.5,1.3){$\bar1$} \put(4.7,0){$\bar1$} \put(2.1,.1){$a$} \end{picture}
\end{array}
\]
we have $T_{2\dots m,b}$ replaced by $T_{1\dots m,b}$, which means
$R_{1b}(u_{1b})$ may be replaced by $\mathbf{1}_{1b}$, because%
\[
\bar{e}_{b}R_{1b}(u_{1b})\bar{e}_{a}R_{1a}(u_{1a})=\bar{e}_{b}\bar{e}%
_{a}R_{1a}(1)+c(u_{1b})\bar{e}_{1}\bar{e}_{a}R_{1a}(1)=\bar{e}_{b}\bar{e}%
_{a}R_{1a}(1)
\]
holds, where $\bar{e}_{1}\bar{e}_{a}R_{1a}(1)=\bar{e}_{1}\bar{e}_{a}a(1)=0$
has been used. Similarly, equality 5 holds. Equality 4 holds because the
Yang-Baxter equation for $R$ implies that $\bar{e}_{a}T_{1\dots m,a}\bar
{e}^{a}$ and $\bar{e}_{b}T_{1\dots m,b}\bar{e}^{b}$ commute.

\item[(b)] Again we prove (\ref{RPi}) by induction. For $m=2$ the claim was
proved in section \ref{s2.4}, for $m>2$ it follows for $1<i<j$ from (\ref{Pi})
and for $i<j<m$ from (\ref{Pi'}).

\item[(c)] The proof of equation (\ref{PiT}) is similar to that of a). We
commute $T(u_{a})$ and $T(u_{0})$, use $\bar{e}_{0}\pi_{1}R_{10}(u_{10}%
)R_{1a}(1)=\bar{e}_{0}\pi_{1}$ and apply induction:%
\begin{align*}
\Pi_{1\dots m}\bar{e}_{0}T_{1\dots m,0}(u_{0})\bar{e}^{0}  &  =\left(  \pi
_{1}\Pi_{2\dots m}\right)  \bar{e}_{a}T_{1\dots m,a}\bar{e}^{a}\bar{e}%
_{0}T_{1\dots m,0}(u_{0})\bar{e}^{0}\\
&  =\left(  \pi_{1}\Pi_{2\dots m}\right)  \bar{e}_{0}T_{1\dots m,0}(u_{0}%
)\bar{e}^{0}\bar{e}_{a}T_{1\dots m,a}\bar{e}^{a}\\
&  =\pi_{1}\left(  \Pi_{2\dots m}\right)  \bar{e}_{0}T_{2\dots m,0}(u_{0}%
)\bar{e}^{0}\bar{e}_{a}T_{1\dots m,a}\bar{e}^{a}\\
&  =\left(  \pi_{1}\Pi_{2\dots m}\right)  \bar{e}_{a}T_{1\dots m,a}\bar{e}%
^{a}=\Pi_{1\dots m}\,.
\end{align*}

\item[(d)] Equations (\ref{m1}) and (\ref{m3}) follow from (\ref{Pi}) and
$R_{1\alpha}^{\bar{1}\mathring{\alpha}_{1}}(1)=0,~R_{\mathring{\alpha}\alpha
}^{\bar{1}\mathring{\alpha}_{1}}(1)=\delta_{\mathring{\alpha}}^{\mathring
{\alpha}_{1}}\delta_{\alpha}^{\bar{1}}$ and analogously (\ref{m4}) and
(\ref{m5}).
\end{itemize}
\end{proof}

\section{Proof of the main theorem
\protect\ref{TN}%
}

\label{a2}

In the following we use the convention that $\alpha,\beta$ etc. take the
values $1,2,\dots,(0),\dots,\bar{2},\bar{1}$ and $\mathring{\alpha}%
,\mathring{\beta}$ etc. take the values $2,\dots,(0),\dots,\bar{2}$.%
\begin{align*}
K_{\underline{\alpha}}(\underline{u})  &  =\sum_{\underline{v}}\,g(\underline
{u},\underline{v})\,\Psi_{\underline{\alpha}}(\underline{u},\underline
{v}),~\Psi_{\underline{\alpha}}(\underline{u},\underline{v})=L_{\underline
{\mathring{\beta}}}(\underline{v})\Phi_{\underline{\alpha}}^{\underline
{\mathring{\beta}}}(\underline{u},\underline{v})\\
\Phi_{\underline{\alpha}}^{\underline{\mathring{\beta}}}(\underline
{u},\underline{v})  &  =\Pi_{\underline{\beta}}^{\underline{\mathring{\beta}}%
}(\underline{v})\left(  \Omega T_{1}^{\beta_{m}}(\underline{u},v_{m})\dots
T_{1}^{\beta_{1}}(\underline{u},v_{1})\right)  _{\underline{\alpha}}\,.
\end{align*}

\paragraph{(i)}

\begin{proof}
Property (i) in the form of (\ref{3.1}) follows directly from the Yang-Baxter
equations and the action of the R-matrix on the pseudo-ground state $\Omega$
\begin{align*}
\left(  T_{\dots ji\dots}\right)  _{1}^{\beta}(\dots u_{j},u_{i}\dots
)\,R_{ij}(u_{ij})  &  =\,R_{ij}(u_{ij})\,\left(  T_{\dots ij\dots}\right)
_{1}^{\beta}(\dots u_{i},u_{j}\dots)\\
\Omega_{\dots ij\dots}R_{ij}(u_{ij})  &  =a(u_{ij})\Omega_{\dots ij\dots}\,.
\end{align*}

\end{proof}

\paragraph{(ii)}

\begin{proof}
We prove
\begin{equation}
K_{1\dots n}(\underline{u}^{\prime})=K_{1\dots n}(\underline{u})\,Q_{1\dots
n}(\underline{u})\,, \label{3.34}%
\end{equation}
where $\underline{u}^{\prime}=(u_{1}+2/\nu,u_{2}\dots,u_{n})$. The matrix
$Q(\underline{u})=Q(\underline{u},1)$ is given by (\ref{3.3'}). Note that we
assign to the auxiliary space of $\tilde{T}_{Q}(\underline{u})$ the spectral
parameter $u_{1}$ on the right hand side and $u_{1}^{\prime}=$ $u_{1}+2/\nu$
on the left hand side. The difference equation (\ref{3.34}) may be depicted
as
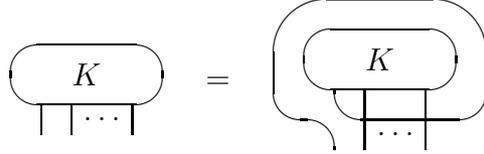
\[%
\begin{array}
[c]{c}%
\unitlength4mm\begin{picture}(5,4) \put(2.5,2){\oval(5,2)} \put(2.5,2){\makebox(0,0){$K$}} \put(1,0){\line(0,1){1}} \put(2,0){\line(0,1){1}} \put(4,0){\line(0,1){1}} \put(2.4,.5){$\dots$} \end{picture}
\end{array}
~~=~~%
\begin{array}
[c]{c}%
\unitlength4mm\begin{picture}(7,5) \put(4.5,2){\oval(5,2)[b]} \put(3.5,2){\oval(7,6)[t]} \put(1,2){\oval(2,2)[lb]} \put(1,0){\oval(2,2)[rt]} \put(3.5,3){\oval(5,2)} \put(3.5,3){\makebox(0,0){$K$}} \put(3,0){\line(0,1){2}} \put(5,0){\line(0,1){2}} \put(3.4,.5){$\dots$} \end{picture}
\end{array}
\]
where we use the rule that the rapidity of a line changes by $2/\nu$ if the
line bends by $360^{0}$ in the positive sense. In the following we will
suppress the indices $1\dots n$. We are now going to prove (\ref{3.34}) in the
form%
\begin{equation}
K(\underline{u})\left(  A_{Q}(\underline{u})+D{_{Q}{}}_{\mathring{\beta}%
}^{\mathring{\beta}}(\underline{u})+A_{3,Q}(\underline{u})\right)
=\prod_{k=2}^{n}a(u_{k1})K(\underline{u}^{\prime})\,, \label{3.42}%
\end{equation}
where $K(\underline{u})$ is a co-vector valued function as given by
eq.~(\ref{3.12}) and the Bethe ansatz state (\ref{Psi}). To analyze the left
hand side of eq.~(\ref{3.42}) we proceed as follows: We apply the trace of
$T_{Q}$ to the co-vector $K(\underline{u})$. In particular we calculate
$\Phi^{\underline{\mathring{\beta}}}(\underline{u},\underline{v}%
)T_{Q}(\underline{u})$%
\[
\Pi_{\underline{\beta}}^{\underline{\mathring{\beta}}}\Omega\,T_{1}^{\beta
_{m}}(\underline{u},v_{m})\cdots T_{1}^{\beta_{1}}(\underline{u},v_{1})\left(
T{_{Q}}(\underline{u})\right)  _{\gamma}^{\gamma^{\prime}}=~%
\begin{array}
[c]{c}%
{\phantom{\rule{1.9397in}{1.382in}}}%
\unitlength5mm\begin{picture}(11,7.2)\thicklines\put(10,5){\oval(17,2)[lb]}
\put(10,5){\oval(20,6)[lb]}
\put(4,0){\line(0,1){5}}
\put(9,0){\line(0,1){5}}
\put(10,5){\oval(14,8)[lb]}
\put(2,0){\oval(2,2)[rt]}
\put(0,6){\line(0,1){.5}}
\put(1.5,6){\line(0,1){.5}}
\put(-.3,5){\framebox(2,1){$\Pi_m$}}
\put(-.4,6.8){$\mathring{\beta}_1$}
\put(1,6.8){$\mathring{\beta}_m$}
\put(2.8,5.3){1}
\put(3.8,5.3){1}
\put(8.8,5.3){1}
\put(1.4,.8){$\gamma'$}
\put(10.2,.8){$\gamma$}
\put(10.2,1.8){1}
\put(10.2,3.8){1}
\put(4.3,0){$u_{2}$}
\put(9.2,0){$u_{n}$}
\put(7.5,2.3){$v_1$} \put(7.5,4.3){$v_m$}
\put(5.5,3){$\dots$}
\put(9.5,2.7){$\vdots$}
\put(7.5,1.2){$u_1$}
\put(2,0){$u_1'$} \end{picture}%
\end{array}
.
\]
We now proceed as usual in the algebraic Bethe ansatz and push $A_{Q}%
(\underline{u}),\,D{_{Q}{}}_{\mathring{\beta}}^{\mathring{\beta}}%
(\underline{u})$ and $A_{3,Q}(\underline{u})$ through all the $T_{1}%
^{\beta_{i}}$-operators. As usual we obtain wanted terms and unwanted terms.
We first find that the wanted contribution from $A_{Q}(\underline{u})$ already
gives the result we are looking for. Secondly the wanted contributions from
$D{_{Q}{}}_{\mathring{\beta}}^{\mathring{\beta}}(\underline{u})$ and
$A_{3,Q}(\underline{u})$ applied to $\Omega$ give zero because of
(\ref{1.41}). Thirdly the unwanted contributions from $A_{Q}(\underline{u}),$
$D{_{Q}{}}_{\mathring{\beta}}^{\mathring{\beta}}(\underline{u})$ and
$A_{3,Q}(\underline{u})$ can be written as differences which cancel after
summation over the $v_{j}$. All these three facts can be seen as follows.

The \textquotedblleft wanted terms" from $\,{A_{Q}}$ are obtained if one
writes the Zapletal commutation rule (\ref{3.8}) as%
\[
{T}_{1}^{\beta_{k}}({\underline{u}},v_{k})\,{A_{Q}}({\underline{u}})={A_{Q}%
}({\underline{u}})\,{T}_{1}^{\beta_{i}}({\underline{u}}^{\prime}%
,v_{k})\,a(u_{1}-v_{k})+uw\,.
\]
Therefore we obtain%
\[
\left(  L(\underline{v})\Pi\underline{u},\underline{v}\right)  _{\underline
{\beta}}\Omega T_{1}^{\beta_{m}}(\underline{u},v_{m})\cdots T_{1}^{\beta_{1}%
}(\underline{u},v_{1})A_{Q}(\underline{u})=w^{A}+uw^{A}%
\]
with%
\begin{align}
w^{A}(\underline{u},\underline{v})  &  =\left(  L(\underline{v})\Pi
\underline{u},\underline{v}\right)  _{\underline{\beta}}\Omega A_{Q}%
(\underline{u})\,T_{1}^{\beta_{m}}(\underline{u}^{\prime},v_{m})\cdots
T_{1}^{\beta_{1}}(\underline{u}^{\prime},v_{1})\prod_{k=1}^{m}a(u_{1}%
-v_{k})\label{wa}\\
&  =\prod_{k=2}^{n}a(u_{k1})\prod_{k=1}^{m}a(u_{1}-v_{k})L_{\mathring{\beta}%
}(\underline{v})\Phi^{\underline{\mathring{\beta}}}(\underline{u}^{\prime
},\underline{v})\nonumber
\end{align}
and similarly $w^{D}=w^{A_{3}}=0$, because of (\ref{1.41}). Inserted into
(\ref{3.12}) this yields%
\begin{align*}
\left(  K(\underline{u})Q\right)  ^{wanted}(\underline{u})  &  =\sum
_{\underline{v}}g(\underline{u},\underline{v})\prod_{k=1}^{m}a(u_{1}%
-v_{k})L_{\mathring{\beta}}(\underline{v})\Phi^{\underline{\mathring{\beta}}%
}(\underline{u}^{\prime},\underline{v})\\
&  =\sum_{\underline{v}}g(\underline{u}^{\prime},\underline{v})L_{\mathring
{\beta}}(\underline{v})\Phi^{\underline{\mathring{\beta}}}(\underline
{u}^{\prime},\underline{v})=K(\underline{u}^{\prime})
\end{align*}
because
\[
g(\underline{u},\underline{v})\prod_{k=1}^{m}a(u_{1}-v_{k})=g(\underline
{u}^{\prime},\underline{v})\,,
\]
where (\ref{3.17}) has been used. Therefore it remains to prove that the
unwanted terms cancel. This will follow from the lemma below.
\end{proof}

We apply $Q_{1\dots n}({\underline{u}})$ to the state $\Psi_{1\dots
n}(\underline{u},\underline{v})$ (suppressing the quantum space indices)%
\[
\Psi(\underline{u},\underline{v})\left(  A_{Q}(\underline{u})+D{_{Q}{}%
}_{\mathring{\beta}}^{\mathring{\beta}}(\underline{u})+A_{3,Q}(\underline
{u})\right)  =wanted+unwanted\,.
\]
The wanted contribution has been calculated above and the unwanted terms may
be written as (see subsection \ref{a3})%
\begin{align*}
unwanted  &  =uw^{A}+uw^{D}+uw^{A_{3}}\\
uw^{A}  &  =\sum_{i=1}^{m}\left(  uw_{C}^{A,i}\right)  _{\mathring{\gamma}%
}C_{Q}^{\mathring{\gamma}}+\sum_{i=1}^{m}\sum_{\underset{j\neq i}{j=1}}%
^{m}uw_{C_{2}}^{A,ij}C_{2,Q}\\
uw^{D}  &  =\sum_{i=1}^{m}\left(  uw_{C}^{D,i}\right)  _{\mathring{\gamma}%
}C_{Q}^{\mathring{\gamma}}+\sum_{i=1}^{m}\sum_{\underset{j\neq i}{j=1}}%
^{m}uw_{C_{2}}^{D,ij}C_{2,Q}+\sum_{i=1}^{m}\left(  uw_{C_{3}}^{D,i}\right)
_{\mathring{\gamma}}\left(  C_{3,Q}\right)  _{\mathring{\gamma}^{\prime}%
}\mathbf{\mathring{C}}^{\mathring{\gamma}^{\prime}\mathring{\gamma}}\\
uw^{A_{3}}  &  =\sum_{i=1}^{m}\sum_{\underset{j\neq i}{j=1}}^{m}uw_{C_{2}%
}^{A_{3},ij}C_{2,Q}+\sum_{i=1}^{m}\left(  uw_{C_{3}}^{A_{3},i}\right)
_{\mathring{\gamma}}\left(  C_{3,Q}\right)  _{\mathring{\gamma}^{\prime}%
}\mathbf{\mathring{C}}^{\mathring{\gamma}^{\prime}\mathring{\gamma}}\,.
\end{align*}

\begin{lemma}
The unwanted terms satisfy the relations%
\begin{align}
\left(  uw_{C}^{D,i}\right)  _{\mathring{\gamma}}(\underline{u},\underline
{v})\,g(\underline{u},\underline{v})  &  =-\left(  uw_{C}^{A,i}\right)
_{\mathring{\gamma}}(\underline{u},\underline{v}^{(i)})\,g(\underline
{u},\underline{v}^{(i)})\label{C1}\\
\left(  uw_{C_{3}}^{A_{3},i}\right)  _{\mathring{\gamma}}(\underline
{u},\underline{v})\,g(\underline{u},\underline{v})  &  =-\left(  uw_{C_{3}%
}^{D,i}\right)  _{\mathring{\gamma}}(\underline{u},\underline{v}%
^{(i)})\,g(\underline{u},\underline{v}^{(i)})\label{C2}\\
uw_{C_{2}}^{D,ij}(\underline{u},\underline{v}^{(j)})\,g(\underline
{u},\underline{v}^{(j)})  &  =-uw_{C_{2}}^{A,ij}(\underline{u},\underline
{v}^{(ij)})\,g(\underline{u},\underline{v}^{(ij)})-uw_{C_{2}}^{A_{3}%
,ij}(\underline{u},\underline{v})\,g(\underline{u},\underline{v}) \label{C3}%
\end{align}
with the notation%
\begin{align*}
\underline{v}  &  =v_{1},\dots,v_{m}\\
\underline{v}^{(i)}  &  =v_{1},\dots,v_{i}^{\prime},\dots,v_{m}\\
\underline{v}^{(ij)}  &  =v_{1},\dots,v_{i}^{\prime},\dots,v_{j}^{\prime
},\dots,v_{m}\\
v^{\prime}  &  =v+2/\nu\,.
\end{align*}
Equations (\ref{C1}) - (\ref{C3}) imply that all unwanted terms cancel after
summation over the $\underline{v}$ in the Jackson-type Integral (\ref{3.12}).
\end{lemma}

\begin{proof}
The first relation (\ref{C1}) is the same as the relation of the unwanted
terms for $SU(N)$. The two others are new and more complicated to derive. We
can calculate the following unwanted contributions explicitly (see subsection
\ref{a3})
\begin{align}
\left(  uw_{C}^{A,i}\right)  _{\mathring{\gamma}}(\underline{u},\underline
{v})  &  =-c(u_{1}^{\prime}-v_{i})X_{\mathring{\gamma}}^{(i)}(\underline
{u},\underline{v})\label{uwa1}\\
\left(  uw_{C}^{D,i}\right)  _{\mathring{\gamma}}(\underline{u},\underline
{v})  &  =c(u_{1}-v_{i})X_{\mathring{\gamma}}^{(i)}(\underline{u}%
,\underline{v}^{(i)})\,\chi_{i}(\underline{u},\underline{v})\label{uwd1}\\
\left(  uw_{C_{3}}^{D,i}\right)  _{\mathring{\gamma}}(\underline{u}%
,\underline{v})  &  =-f(u_{1}^{\prime}-v_{i})X_{\mathring{\gamma}}%
^{(i)}(\underline{u},\underline{v})\label{uwd3}\\
\left(  uw_{C_{3}}^{A_{3},i}\right)  _{\mathring{\gamma}}(\underline
{u},\underline{v})  &  =f(u_{1}-v_{i})X_{\mathring{\gamma}}^{(i)}%
(\underline{u},\underline{v}^{(i)})\,\chi_{i}(\underline{u},\underline{v})
\label{uwa3}%
\end{align}
with
\begin{equation}
X_{\mathring{\gamma}}^{(i)}(\underline{u},\underline{v})=L(v_{i},\underline
{v}_{i})_{\mathring{\gamma}\underline{\mathring{\beta}}_{i}}\Phi
^{\underline{\mathring{\beta}}_{i}}(\underline{u},\underline{v}_{i}%
)\prod_{\underset{k\neq i}{k=1}}^{m}a(v_{ik})a_{1}(\underline{u},v_{i})
\label{Xi}%
\end{equation}
and%
\begin{equation}
\chi_{i}(\underline{u},\underline{v})=\frac{a_{2}(\underline{u},v_{i})}%
{a_{1}(\underline{u},v_{i}^{\prime})}\prod_{\underset{k\neq i}{k=1}}^{m}%
\frac{a(v_{ki})}{a(v_{ik}^{\prime})}\,, \label{chi}%
\end{equation}
where $a_{1}$ and $a_{2}$ are defined in (\ref{1.40}). We use the short
notations $\underline{v}_{i}$ and $\underline{\mathring{\beta}}_{i}$ which
means that $v_{i}$ and $\mathring{\beta}_{i}$ are missing, respectively. The
remaining unwanted terms are
\begin{align}
uw_{C_{2}}^{A,ij}(\underline{u},\underline{v})  &  =-c(u_{1}^{\prime}%
-v_{i})X^{(ij)}(\underline{v})\label{uwa2}\\
uw_{C_{2}}^{D,ij}(\underline{u},\underline{v})  &  =\left(  c(u_{1}%
-v_{i})+f(u_{1}^{\prime}-v_{j})\right)  X^{(ij)}(\underline{v}^{(i)}%
)\,\chi_{i}(\underline{u},\underline{v})\label{uwd2}\\
uw_{C_{2}}^{A_{3},ij}(\underline{u},\underline{v})  &  =-f(u_{1}%
-v_{j})X^{(ij)}(v_{i}^{\prime},v_{j}^{\prime},\underline{v}_{ij})\,\chi
_{i}(\underline{u},\underline{v}^{(j)})\,\chi_{j}(\underline{u},\underline{v})
\label{uwa32}%
\end{align}
with%
\begin{equation}
X^{(ij)}(\underline{u},\underline{v})=f(v_{ij})a(v_{ij})L(v_{i},v_{j}%
,\underline{v}_{ij})\mathbf{\mathring{C}}^{ij}\Phi(\underline{u},\underline
{v}_{ij})\prod_{\underset{k\neq i,j}{k=1}}^{m}\left(  a(v_{ik})a(v_{jk}%
)\right)  a_{1}(v_{i})a_{1}(v_{j})\,, \label{Xij}%
\end{equation}
where again $\underline{v}_{ij}$ means that $v_{i}$ and $v_{j}$ are missing.
The claims (\ref{C1}) - (\ref{C3}) follow from the shift property of the
function $g(\underline{u},\underline{v})$ defined in (\ref{3.15})%
\[
\chi_{i}(\underline{u},\underline{v})g(\underline{u},\underline{v}%
)=g(\underline{u},\underline{v}^{(i)})
\]
which is due to the shift properties (\ref{3.17}) of the functions $\psi(v)$
and $\tau(v)$.

The states $X_{\mathring{\gamma}}^{(1)}$ and $X^{(12)}$ may be depicted as%
\[
X_{\mathring{\gamma}}^{(1)}(\underline{u},\underline{v})=%
\begin{array}
[c]{c}%
{\phantom{\rule{1.4826in}{1.1784in}}}%
\end{array}%
\begin{array}
[c]{c}%
\unitlength3.8mm\begin{picture}(10,9)(0,0) \thicklines\put(3.3,0){$u_1$} \put(5.8,0){$u_n$} \put(2.3,3.8){$v_m$} \put(1.3,2.2){$v_2$} \put(-1.1,5){$\mathring \gamma$} \put(-2.4,6.8){$v_1$} \put(5.3,5.3){$v_1$} \put(3.4,4,8){1} \put(4.4,6.2){1} \put(6.9,6.2){1} \put(8.,0){1} \put(9.2,1){1} \put(9.2,2){1} \put(9.2,3.7){1} \put(5.3,2.4){$\dots$} \put(1.,5){${\dots}$} \put(4.5,0){\line(0,1){6}} \put(7,0){\line(0,1){6}} \put(9,6){\oval(17,9)[lb]} \put(9,6){\oval(13,4)[lb]} \put(4,5){\line(1,0){3}} \put(.7,8.2){$L$} \put(.8,8.6){\oval(4.8,1.2)} \put(7,1){\oval(2,8)[rt]} \put(0,6){\framebox(2.8,1){$\Pi$}} \put(-1,6){\line(0,1){2}} \put(.5,7){\line(0,1){1}} \put(2.5,7){\line(0,1){1}} \end{picture}
\end{array}
,~~\frac{X^{(12)}(\underline{u},\underline{v})}{f(v_{12})a(v_{12})}=%
\begin{array}
[c]{c}%
{\phantom{\rule{1.5424in}{1.1734in}}}%
\end{array}%
\begin{array}
[c]{c}%
\unitlength3.8mm\begin{picture}(10,9)(0,0) \thicklines\put(3.3,0){$u_1$} \put(5.8,0){$u_n$} \put(2.3,3.8){$v_m$} \put(1.3,2.2){$v_3$} \put(5.3,6.2){$v_1$} \put(5.3,5.2){$v_2$} \put(3.4,4.8){1} \put(3.4,5.8){1} \put(4.4,7.2){1} \put(6.9,7.2){1} \put(8.,0){1} \put(9.,0){1} \put(9.9,1){1} \put(9.9,3.7){1} \put(5.3,2.4){$\dots$} \put(1.,5){${\dots}$} \put(4.5,0){\line(0,1){7}} \put(7,0){\line(0,1){7}} \put(9.5,6){\oval(18,9)[lb]} \put(9.5,6){\oval(14,4)[lb]} \put(.2,8.2){$L$} \put(.5,8.6){\oval(5.2,1.2)} \put(-1,8){\oval(1,2)[b]} \put(4,1){\oval(8,8)[rt]} \put(4,1){\oval(10,10)[rt]} \put(0,6){\framebox(2.8,1){$\Pi$}} \put(.5,7){\line(0,1){1}} \put(2.5,7){\line(0,1){1}} \end{picture}
\end{array}
\]

\end{proof}

Note that for the on-shell Bethe ansatz the equations for the unwanted terms
look quite similar as (\ref{uwa1}) - and (\ref{uwa32}) (apart from the fact
that there are no shifts) and their cancellation is equivalent to the Bethe
ansatz equations
\[
\chi_{i}(\underline{u},\underline{v})=\frac{a_{2}(\underline{u},v_{i})}%
{a_{1}(\underline{u},v_{i})}\prod_{\underset{k\neq i}{k=1}}^{m}\frac
{a(v_{ki})}{a(v_{ik})}=\prod_{k=1}^{n}\frac{\left(  v_{i}+\frac{1}{2}\right)
-u_{k}-\frac{1}{2}}{\left(  v_{i}+\frac{1}{2}\right)  -u_{k}+\frac{1}{2}}%
\prod_{\underset{k\neq i}{k=1}}^{m}\frac{v_{ik}+1}{v_{ik}-1}=1\,.
\]

\subsection{The unwanted terms\label{a3}}

In our approach of the algebraic Bethe ansatz lemma \ref{l2} replaces
commutation rules of the entries of the monodromy matrix. In the conventional
approach one derives them from the Yang-Baxter algebra relations (\ref{TTS})
and uses them for the algebraic Bethe ansatz.

\paragraph{The unwanted terms $uw^{A}$:}

We start with (\ref{wa}) and use Yang-Baxter relations and lemma \ref{l2} in
the form of (\ref{Sa})

\noindent$w^{A}=L(\underline{v})%
\begin{array}
[c]{c}%
{\phantom{\rule{1.6172in}{1.35in}}}%
\unitlength4mm\begin{picture}(10,8)
\thicklines\put(3.3,0){$u_1'$}
\put(5.8,0){$u_n$}
\put(5.3,2.8){$v_i$}
\put(5.3,5.3){$u_1$}
\put(3.4,6.2){1}
\put(4.4,6.2){1}
\put(6.9,6.2){1}
\put(8.,0){1}
\put(9.2,1){1}
\put(9.2,2){1}
\put(9.2,3.7){1}
\put(5.4,1.8){$\dots$}
\put(.3,5.5){$_{\dots}$}
\put(1.6,5.5){$_{\dots}$}
\put(4.5,0){\line(0,1){4}}
\put(7,0){\line(0,1){6}}
\put(9,6){\oval(15.5,7)[lb]}
\put(9,6){\oval(18,9)[lb]}
\put(9,6){\oval(13,4)[lb]}
\put(7,6){\oval(5,2)[lb]}
\put(4,6){\oval(1,2)[lb]}
\put(4,4){\oval(1,2)[rt]}
\put(7,1){\oval(2,8)[rt]}
\put(-.2,6){$\framebox(3,1){$\Pi$}$}
\put(0,7){\line(0,1){1}}
\put(1.25,7){\line(0,1){1}}
\put(2.5,7){\line(0,1){1}}
\end{picture}%
\end{array}
=L(\underline{v})%
\begin{array}
[c]{c}%
{\phantom{\rule{1.6717in}{1.35in}}}%
\unitlength4mm\begin{picture}(10,8)\thicklines\put(3.9,0){$u_1'$}
\put(6.8,0){$u_n$}
\put(5.9,3.8){$v_i$}
\put(5.3,5.3){$u_1$}
\put(3.4,6.2){1}
\put(5.0,6.2){1}
\put(8.0,6.2){1}
\put(9.5,.7){1}
\put(9.5,1.8){1}
\put(9.5,3.2){1}
\put(9.5,4.8){1}
\put(5.8,2.5){$\dots$}
\put(.3,5.5){$_{\dots}$}
\put(1.6,5.5){$_{\dots}$}
\put(8,0){\line(0,1){6}}
\put(9,6){\oval(15.5,5)[lb]}
\put(9,6){\oval(18,8)[lb]}
\put(9,6){\oval(13,2)[lb]}
\put(4.5,6){\oval(1.5,10)[lb]}
\put(4.5,0){\oval(1.5,2)[rt]}
\put(9,6){\oval(8,10)[lb]}
\put(-.2,6){$\framebox(3,1){$\Pi$}$}
\put(0,7){\line(0,1){1}}
\put(1.25,7){\line(0,1){1}}
\put(2.5,7){\line(0,1){1}}
\end{picture}%
\end{array}
$ \newline$=L(\underline{v})%
\begin{array}
[c]{c}%
\raisebox{-0.0069in}{\phantom{\rule{1.6717in}{1.35in}}}%
\unitlength4mm\begin{picture}(10,8)\thicklines\put(3.9,0){$u_1'$}
\put(6.8,0){$u_n$}
\put(5.9,3.8){$v_i$}
\put(5.3,5.3){$u_1$}
\put(2.8,.7){1}
\put(5.0,6.2){1}
\put(8.0,6.2){1}
\put(9.5,.7){1}
\put(9.5,1.8){1}
\put(9.5,3.2){1}
\put(9.5,4.8){1}
\put(5.8,2.5){$\dots$}
\put(.3,5.5){$_{\dots}$}
\put(1.6,5.5){$_{\dots}$}
\put(8,0){\line(0,1){6}}
\put(9,6){\oval(15.5,5)[lb]}
\put(9,6){\oval(18,8)[lb]}
\put(9,6){\oval(13,2)[lb]}
\put(3.5,0){\oval(3.5,2)[rt]}
\put(9,6){\oval(8,10)[lb]}
\put(-.2,6){$\framebox(3,1){$\Pi$}$}
\put(0,7){\line(0,1){1}}
\put(1.25,7){\line(0,1){1}}
\put(2.5,7){\line(0,1){1}}
\end{picture}%
\end{array}
+%
{\displaystyle\sum\limits_{i=1}^{m}}
c(u_{1}^{\prime}-v_{i})L(\underline{v})%
\begin{array}
[c]{c}%
{\phantom{\rule{1.6717in}{1.3102in}}}%
\unitlength4mm\begin{picture}(10,8)\thicklines\put(3.2,.4){$v_i$}
\put(6,3.8){$v_i$}
\put(6,.2){$u_1$}
\put(3.4,6.2){1}
\put(5.0,6.2){1}
\put(8.0,6.2){1}
\put(9.5,.7){1}
\put(9.5,1.8){1}
\put(9.5,3.2){1}
\put(9.5,4.8){1}
\put(5.8,5.5){$\dots$}
\put(.3,5.5){$_{\dots}$}
\put(1.6,5.5){$_{\dots}$}
\put(8,0){\line(0,1){6}}
\put(2.1,6){\oval(1.8,5)[lb]}
\put(2.1,2.5){\oval(3.,2)[rt]}
\put(9,6){\oval(18,8)[lb]}
\put(9,6){\oval(13,2)[lb]}
\put(9,6){\oval(11,5)[lb]}
\put(4.2,2.5){\oval(1.2,3)[lb]}
\put(4,0){\oval(2,2)[rt]}
\put(9,6){\oval(8,10)[lb]}
\put(-.2,6){$\framebox(3,1){$\Pi$}$}
\put(0,7){\line(0,1){1}}
\put(1.25,7){\line(0,1){1}}
\put(2.5,7){\line(0,1){1}}
\end{picture}%
\end{array}
\newline=\Psi A_{Q}-uw^{A}$.

Note that the $d(u_{1}^{\prime}-v_{i})$ contributions vanish because they
produce terms with $\Pi_{\dots\bar{1}}^{\dots}=0$ (see (\ref{m4})). This reads
in terms of formulas as
\begin{align*}
w^{A}  &  =\left(  L\Pi\right)  _{\underline{\beta}}\Omega A_{Q}\,T_{1}%
^{\beta_{m}}(\underline{u}^{\prime},v_{m})\cdots T_{1}^{\beta_{1}}%
(\underline{u}^{\prime},v_{1})a(u_{1}-v_{m})\dots a(u_{1}-v_{1})\\
&  =\left(  L\Pi\right)  _{\underline{\beta}}\Omega\,\left(  R_{0m}%
(u_{1}^{\prime}-v_{m})\dots R_{01}(u_{1}^{\prime}-v_{1})\right)
_{\gamma,\underline{\beta}^{\prime}}^{\underline{\beta},1}T_{1}^{\beta
_{m}^{\prime}}(\underline{u},v_{m})\cdots T_{1}^{\beta_{1}^{\prime}%
}(\underline{u},v_{1})\left(  T_{Q}\right)  _{1}^{\gamma}\\
&  =\left(  L\Pi\right)  _{\underline{\beta}}\Omega T_{1}^{\beta_{m}}%
(z_{m})\cdots T_{1}^{\beta_{1}}(z_{1})A_{Q}-uw^{A}\,
\end{align*}
with%
\begin{align*}
uw^{A}  &  =-\sum_{i=1}^{m}c(u_{1}^{\prime}-v_{i})\left(  L\Pi\right)
_{\underline{\beta}}\Omega\left(  R_{0m}(v_{im})\dots\mathbf{P}_{0i}\dots
R_{01}(v_{i1})\right)  _{\gamma,\underline{\beta}^{\prime}}^{\underline{\beta
},1}\,\\
&  \times T_{1}^{\beta_{m}^{\prime}}(\underline{u},v_{m})\cdots T_{1}%
^{\beta_{1}^{\prime}}(\underline{u},v_{1})\left(  T_{Q}\right)  _{1}^{\gamma
}\,,
\end{align*}
where lemma \ref{l2} in the form of (\ref{Sa}) has been used. We commute the
$R_{ik}(v_{ik})$ (for $k<i)$ with $\Pi$ using Yang-Baxter relations and
(\ref{RPi}) and apply the $\mathring{R}_{ik}(v_{ik})$ to $L$ using
(\ref{3.1a}) such that $\mathring{R}_{ik}(v_{ik})\rightarrow a(v_{ik})$.
Further we use Yang-Baxter relations to the $R_{ik}(v_{ik})$ (for $k>i)$ and
note that for $R_{0i}(v_{ik})$ only $R_{11}^{11}(v_{ik})=a(v_{ik})$ contribute%
\begin{align}
uw^{A}  &  =-\sum_{i=1}^{m}c(u_{1}^{\prime}-v_{i})\left(  L(v_{i}%
,\underline{v}_{i})\Pi(v_{i},\underline{v}_{i})\right)  _{\gamma
\underline{\beta}_{i}}\nonumber\\
&  \times\prod_{k=1,k\neq i}^{m}a(v_{ik})\Omega A(v_{i})\left[  T_{1}%
^{\beta_{m}}(v_{m})\cdots T_{1}^{\beta_{1}}(v_{1})\right]  _{i}\left(
T_{Q}\right)  _{1}^{\gamma}\label{AC}\\
&  =\left(  uw_{C}^{A}\right)  _{\mathring{\gamma}}C_{Q}^{\mathring{\gamma}%
}+uw_{C_{2}}^{A}C_{2,Q}\,,
\end{align}
where we use the short notations $\underline{v}_{i},~\underline{\beta}_{i}$
and $\left[  T_{1}^{\beta_{m}}(v_{m})\cdots T_{1}^{\beta_{1}}(v_{1})\right]
_{i}$ which means that $v_{i},$~$\beta_{i}$ and $T_{1}^{\beta_{i}}(v_{i})$ are
missing, respectively. We can apply (\ref{1.40}) $\Omega A(v_{i}%
)=a_{1}(\underline{u},v_{i})\Omega$. The different unwanted terms are due to
different values of $\gamma$ in (\ref{AC}). For $\gamma=1$ the right hand side
in (\ref{AC}) vanishes because of $\Pi_{1\dots}^{\dots}=0$ (see (\ref{m1})),
for $\gamma=\mathring{\gamma}\neq1,\bar{1}$ this gives $uw_{C}^{A}$ and for
$\gamma=\bar{1}$ this gives $uw_{C_{2}}^{A}$.

First we calculate $\left(  uw_{C}^{A}\right)  _{\mathring{\gamma}}$ using
(\ref{m3}) for $\gamma=\mathring{\gamma}$
\begin{align*}
\left(  uw_{C}^{A}\right)  _{\mathring{\gamma}}C_{Q}^{\mathring{\gamma}}  &
=-\sum_{i=1}^{m}c(u_{1}^{\prime}-v_{i})\left(  L(v_{i},\underline{v}_{i}%
)\Pi(v_{i},\underline{v}_{i})\right)  _{\mathring{\gamma}\underline{\beta}%
_{i}}\\
&  \times\prod_{\underset{k\neq i}{k=1}}^{m}a(v_{ik})a_{1}(\underline{u}%
,v_{i})\Omega\left[  T_{1}^{\beta_{m}}(v_{m})\cdots T_{1}^{\beta_{1}}%
(v_{1})\right]  _{i}\left(  T_{Q}\right)  _{1}^{\mathring{\gamma}}\\
&  =-\sum_{i=1}^{m}c(u_{1}^{\prime}-v_{i})X_{\mathring{\gamma}}^{(i)}%
(\underline{u},\underline{v})C_{Q}^{\mathring{\gamma}}%
\end{align*}
with $X_{\mathring{\gamma}}^{(i)}$ defined by (\ref{Xi}).%
\begin{equation}
X_{\mathring{\gamma}}^{(i)}(\underline{u},\underline{v})=L(v_{i},\underline
{v}_{i})_{\mathring{\gamma}\underline{\mathring{\beta}}_{i}}\Phi
^{\underline{\mathring{\beta}}_{i}}(\underline{u},\underline{v}_{i}%
)\prod_{\underset{k\neq i}{k=1}}^{m}a(v_{ik})a_{1}(\underline{u},v_{i})
\end{equation}

The remaining unwanted term $uw_{C_{2}}^{A}$ comes from $\gamma=\bar{1}$ in
(\ref{AC}). Using (\ref{Pia}) we get%
\begin{align*}
uw_{C_{2}}^{A}C_{2,Q}  &  =-\sum_{i=1}^{m}c(u_{1}^{\prime}-v_{i})\left(
L(v_{i},\underline{v}_{i})\Pi(v_{i},\underline{v}_{i})\right)  _{\bar
{1}\underline{\beta}_{i}}\\
&  \times\prod_{\underset{k\neq i}{k=1}}^{m}a(v_{ik})a_{1}(\underline{u}%
,v_{i})\Omega\left[  T_{1}^{\beta_{m}}(v_{m})\cdots T_{1}^{\beta_{1}}%
(v_{1})\right]  _{i}\left(  T_{Q}\right)  _{1}^{\bar{1}}\\
&  =-\sum_{i=1}^{m}c(u_{1}^{\prime}-v_{i})\sum_{\underset{j\neq i}{j=1}}%
^{m}f(v_{ij})\left(  L(v_{i},v_{j},\underline{v}_{ij})\mathbf{\mathring{C}%
}^{ij}\Pi(\underline{v}_{ij})\right)  _{\underline{\beta}_{ij}}\\
&  \times\prod_{\underset{k\neq i}{k=1}}^{m}a(v_{ik})a_{1}(\underline{u}%
,v_{i})\prod_{\underset{k\neq i,j}{k=1}}^{m}a(v_{jk})a_{1}(\underline{u}%
,v_{j})\Omega\left[  T_{1}^{\beta_{m}}(v_{m})\dots T_{1}^{\beta_{1}}%
(v_{1})\right]  _{ij}C_{2,Q}\\
&  =-\sum_{i=1}^{m}c(u_{1}^{\prime}-v_{i})\sum_{\underset{j\neq i}{j=1}}%
^{m}X^{(ij)}(\underline{v})C_{2,Q}%
\end{align*}
with $X^{(ij)}$ defined by (\ref{Xij}). Therefore we obtain $\left(
uw_{C}^{A,i}\right)  _{\mathring{\gamma}}(\underline{u},\underline{v})$ and
$uw_{C_{2}}^{A,ij}(\underline{u},\underline{v})$ in the form of (\ref{uwa1})
and (\ref{uwa2}).

\paragraph{The unwanted terms $uw^{D}$:}

For convenience we add an extra line to $\Pi$ and consider $\Pi_{\gamma
^{\prime}\underline{\beta^{\prime}}}^{\mathring{\gamma}\underline
{\mathring{\beta}}}(u_{1}^{\prime},\underline{v})$. Using $\Omega T{_{Q}{}%
}_{\mathring{\gamma}}^{\gamma}=0$ (see (\ref{1.41}) which also implies that
the $D$-wanted term vanishes) , Yang-Baxter and (\ref{Sb}) we derive \\[2mm]
$0=L(\underline{v})%
\begin{array}
[c]{c}%
{\phantom{\rule{1.4926in}{1.4502in}}}%
\unitlength4mm\begin{picture}(9,8)(1,0)\thicklines\put(3.9,0){$u'_1$}
\put(8.3,0){$u_n$}
\put(6,2.6){$v_i$}
\put(6,4.){$u_1$}
\put(9.5,4.4){$\mathring\gamma$}
\put(1.3,7.6){$\mathring\gamma$}
\put(5.8,5.5){$\dots$}
\put(2.2,4.5){$_{..}$}
\put(2.8,4.5){$_{\dots}$}
\put(8,0){\line(0,1){6}}
\put(9,6){\oval(8,2)[lb]}
\put(4.25,6){\oval(5.5,2)[lb]}
\put(4.25,0){\oval(1.5,10)[rt]}
\put(9,6){\oval(11,5)[lb]}
\put(9,6){\oval(12.5,7)[lb]}
\put(9,6){\oval(14,9)[lb]}
\put(1.,6){\framebox(3,1){$\Pi$}}
\put(1.5,7){\line(0,1){.5}}
\put(2.,7){\line(0,1){.5}}
\put(2.8,7){\line(0,1){.5}}
\put(3.5,7){\line(0,1){.5}}
\put(4.8,6.2){1}
\put(7.8,6.2){1}
\put(9.5,3.2){1}
\put(9.5,2.2){1}
\put(9.5,1.2){1}
\end{picture}%
\end{array}
=L(\underline{v})%
\begin{array}
[c]{c}%
{\phantom{\rule{1.8001in}{1.4776in}}}%
\unitlength4mm\begin{picture}(12,8)\thicklines\put(3.9,0){$u'_1$}
\put(5.5,0){$u_n$}
\put(7.3,0){1}
\put(8.3,0){1}
\put(9.3,0){1}
\put(5,3.9){$v_i$}
\put(4.5,2.3){$u_1$}
\put(10.2,1.4){$\mathring\gamma$}
\put(-.3,7.7){$\mathring\gamma$}
\put(7.5,3){$_{\dots}$}
\put(8.8,3){$_{\dots}$}
\put(4.5,6){$\dots$}
\put(6.5,0){\line(0,1){6}}
\put(3.5,1){\oval(12,9)[rt]}
\put(3.5,6){\oval(2,1)[lb]}
\put(3.5,1){\oval(10,7)[rt]}
\put(3.5,6){\oval(3.5,3)[lb]}
\put(3.5,1){\oval(8,5)[rt]}
\put(3.5,6){\oval(5,5)[lb]}
\put(10,6){\oval(13,8)[lb]}
\put(2,6){\oval(4,8)[lb]}
\put(2,0){\oval(3,4)[rt]}
\put(-.2,6){\framebox(3,1){$\Pi$}}
\put(0,7){\line(0,1){.5}}
\put(.9,7){\line(0,1){.5}}
\put(1.7,7){\line(0,1){.5}}
\put(2.5,7){\line(0,1){.5}}
\end{picture}%
\end{array}
\newline=L(\underline{v})%
\begin{array}
[c]{c}%
{\phantom{\rule{1.3795in}{1.4527in}}}%
\unitlength4mm\begin{picture}(8,8)\thicklines\put(3.9,0){$u'_1$}
\put(6.3,0){$u_n$}
\put(4.5,3,9){$v_i$}
\put(4.2,1.2){$u_1$}
\put(7.5,1.4){$\mathring\gamma$}
\put(-.3,7.6){$\mathring\gamma$}
\put(4.2,6){$\dots$}
\put(6,.5){\line(0,1){6}}
\put(7,6){\oval(9,1)[lb]}
\put(7,6){\oval(10.5,3)[lb]}
\put(7,6){\oval(12,5)[lb]}
\put(7,6.5){\oval(7,9)[lb]}
\put(2,6){\oval(4,8)[lb]}
\put(2,.5){\oval(3,3)[rt]}
\put(-.2,6){\framebox(3,1){$\Pi$}}
\put(0,7){\line(0,1){.5}}
\put(.9,7){\line(0,1){.5}}
\put(1.7,7){\line(0,1){.5}}
\put(2.5,7){\line(0,1){.5}}
\end{picture}%
\end{array}
+%
{\displaystyle\sum\limits_{i=1}^{m}}
c(v_{i}-u_{1})L(\underline{v})%
\begin{array}
[c]{c}%
{\phantom{\rule{1.8425in}{1.4502in}}}%
\unitlength4mm\begin{picture}(11,8)\thicklines\put(3.9,0){$u'_1$}
\put(6.3,0){$u_n$}
\put(4.5,3,9){$v_i$}
\put(4.2,1.2){$u_1$}
\put(10.5,1.4){$\mathring\gamma$}
\put(-.3,7.6){$\mathring\gamma$}
\put(8.2,.3){$v_i$}
\put(6.5,1.7){$\bullet$}
\put(7.5,3){$_{\dots}$}
\put(8.8,3){$_{\dots}$}
\put(4.2,6){$\dots$}
\put(6,.5){\line(0,1){6}}
\put(3.5,1){\oval(12,9)[rt]}
\put(3.5,6){\oval(2,1)[lb]}
\put(3.5,3){\oval(10,3)[rt]}
\put(10,3){\oval(3,2)[lb]}
\put(3.5,6){\oval(3.5,3)[lb]}
\put(3.5,1){\oval(8,5)[rt]}
\put(3.5,6){\oval(5,5)[lb]}
\put(7,1){\oval(3,2)[rt]}
\put(2,6){\oval(4,8)[lb]}
\put(2,.5){\oval(3,3)[rt]}
\put(7,6.5){\oval(7,9)[lb]}
\put(-.2,6){\framebox(3,1){$\Pi$}}
\put(0,7){\line(0,1){.5}}
\put(.9,7){\line(0,1){.5}}
\put(1.7,7){\line(0,1){.5}}
\put(2.5,7){\line(0,1){.5}}
\end{picture}%
\end{array}
$\newline$=\Psi\left(  D_{Q}\right)  _{\mathring{\gamma}}^{\mathring{\gamma}%
}-uw^{D}$,

\noindent where the dot in the last picture means the the spectral parmeter on
the left of it is $u_{1}$ and on the right of it is $v_{i}$. Again the
$d(v_{i}-u_{1})$-terms do not contribute because they produce terms with
$\Pi_{\dots\bar{1}}^{\dots}=0$ (see (\ref{m4})). In terms of formulas this
reads as
\begin{align*}
0  &  =L_{\underline{\mathring{\beta}}}(\underline{v})\left(  \Pi
(u_{1}^{\prime},\underline{v})R_{10}(v_{1}-u_{1}^{\prime})\cdots R_{m0}%
(v_{m}-u_{1}^{\prime})\right)  _{\underline{\beta^{\prime}}\gamma^{\prime}%
}^{\mathring{\gamma}\underline{\mathring{\beta}}}\Omega\left(  T_{Q}\right)
_{\mathring{\gamma}}^{\gamma^{\prime}}T_{1}^{\beta_{m}^{\prime}}(v_{m})\cdots
T_{1}^{\beta_{1}^{\prime}}(v_{1})\\
&  =L_{\underline{\mathring{\beta}}}(\underline{v})\Pi_{\gamma^{\prime
}\underline{\beta^{\prime}}}^{\mathring{\gamma}\underline{\mathring{\beta}}%
}(u_{1}^{\prime},\underline{v})\Omega T_{\gamma_{m}}^{\beta_{m}^{\prime}%
}(v_{m})\cdots T_{\gamma_{1}}^{\beta_{1}^{\prime}}(v_{1})\left(  T_{Q}\right)
_{\gamma^{\prime\prime}}^{\gamma^{\prime}}\left(  R_{10}(v_{1}-u_{1})\cdots
R_{m0}(v_{m}-u_{1})\right)  _{1\cdots1,\mathring{\gamma}}^{\gamma
^{\prime\prime},\underline{\gamma}}\\
&  =L_{\underline{\mathring{\beta}}}(\underline{v})\Pi_{\underline
{\beta^{\prime}}}^{\underline{\mathring{\beta}}}(\underline{z})\Omega
T_{1}^{\beta_{m}^{\prime}}(v_{m})\cdots T_{1}^{\beta_{1}^{\prime}}%
(v_{1})\left(  D_{Q}\right)  _{\mathring{\gamma}}^{\mathring{\gamma}}-uw^{D}.
\end{align*}
with%
\begin{align}
uw^{D}  &  =L_{\underline{\mathring{\beta}}}(\underline{v})\Pi_{\underline
{\beta^{\prime}}}^{\underline{\mathring{\beta}}}(\underline{v})\Omega
T_{1}^{\beta_{m}^{\prime}}(v_{m})\cdots T_{1}^{\beta_{1}^{\prime}}%
(v_{1})\left(  D_{Q}\right)  _{\mathring{\gamma}}^{\mathring{\gamma}%
}\nonumber\\
&  -L_{\underline{\mathring{\beta}}}(\underline{v})\Pi_{\gamma^{\prime
}\underline{\beta^{\prime}}}^{\mathring{\gamma}\underline{\mathring{\beta}}%
}(u_{1}^{\prime},\underline{v})\Omega T_{1}^{\beta_{m}^{\prime}}(v_{m})\cdots
T_{1}^{\beta_{1}^{\prime}}(v_{1})\left(  T_{Q}\right)  _{\mathring{\gamma}%
}^{\gamma^{\prime}}-\sum_{i=1}^{m}c(v_{i}-u_{1})\label{DC}\\
&  \times\Pi_{\gamma^{\prime}\underline{\beta^{\prime}}}^{\underline
{\mathring{\gamma}\mathring{\beta}}}(u_{1}^{\prime},\underline{v})\Omega
T_{\gamma_{m}}^{\beta_{m}^{\prime}}(v_{m})\cdots T_{\gamma_{1}}^{\beta
_{1}^{\prime}}(v_{1})\left(  T_{Q}\right)  _{\gamma^{\prime\prime}}%
^{\gamma^{\prime}}\left(  R_{10}(v_{1i})\cdots\mathbf{P}_{i0}\cdots
R_{m0}(v_{mi})\right)  _{1\cdots1,\mathring{\gamma}}^{\gamma^{\prime\prime
},\underline{\gamma}}\label{DC'}\\
&  =\left(  uw_{C}^{D}\right)  _{\mathring{\gamma}}C_{Q}^{\mathring{\gamma}%
}+uw_{C_{2}}^{D}C_{2,Q}+\left(  uw_{C_{3}}^{D}\right)  _{\mathring{\gamma
}^{\prime}}\mathbf{\mathring{C}}^{\mathring{\gamma}\mathring{\gamma}^{\prime}%
}\left(  C_{3,Q}\right)  _{\mathring{\gamma}}\,,\nonumber
\end{align}
where lemma \ref{l2} in the form of (\ref{Sb}) has been used. The different
unwanted terms are due to different values of $\gamma^{\prime}$ in (\ref{DC})
and (\ref{DC'}). For $\gamma^{\prime}\neq\bar{1}$ in (\ref{DC}) the second
term cancels the first one. For $\gamma^{\prime}=\bar{1}$ in (\ref{DC}) we get
$\left(  uw_{C_{3}}^{D}\right)  _{\mathring{\gamma}}$%

\[
\left(  uw_{C_{3}}^{D}\right)  _{\mathring{\gamma}^{\prime}}\mathbf{\mathring
{C}}^{\mathring{\gamma}\mathring{\gamma}^{\prime}}\left(  C_{3,Q}\right)
_{\mathring{\gamma}}=-L_{\underline{\mathring{\beta}}}(\underline{v})\Pi
_{\bar{1}\underline{\beta^{\prime}}}^{\mathring{\gamma}\underline
{\mathring{\beta}}}(u_{1}^{\prime},\underline{v})\Omega T_{1}^{\beta
_{m}^{\prime}}(v_{m})\cdots T_{1}^{\beta_{1}^{\prime}}(v_{1})\left(
T_{Q}\right)  _{\mathring{\gamma}}^{\bar{1}}\,.
\]
Using (\ref{m6}) we get%
\begin{align*}
&  \left(  uw_{C_{3}}^{D}\right)  _{\mathring{\gamma}^{\prime}}%
\mathbf{\mathring{C}}^{\mathring{\gamma}\mathring{\gamma}^{\prime}}\\
&  =-\sum_{i=1}^{m}f(u_{1}^{\prime}-v_{i})\left(  L(\underline{v})\mathring
{R}_{ii-1}\cdots\mathring{R}_{i1}\mathbf{\mathring{C}}^{0i}\Pi(\underline
{v}_{i})e_{i}R_{im}\cdots R_{ii+1}\right)  _{\underline{\beta}^{\prime}%
}^{\mathring{\gamma}}\Omega T_{1}^{\beta_{m}^{\prime}}(v_{m})\cdots
T_{1}^{\beta_{1}^{\prime}}(v_{1})\\
&  =-\sum_{i=1}^{m}f(u_{1}^{\prime}-v_{i})\left(  L(v_{i},\underline{v}%
_{i})\mathbf{1}_{i}\Pi(\underline{v}_{i})\right)  _{\mathring{\gamma
}\underline{^{\prime}\beta^{\prime}}_{i}}\mathbf{\mathring{C}}^{\mathring
{\gamma}\mathring{\gamma}^{\prime}}\prod_{\underset{k\neq i}{k=1}}^{m}%
a(v_{ik})\Omega A(v_{i})\left[  T_{1}^{\beta_{m}^{\prime}}(v_{m})\cdots
T_{1}^{\beta_{1}^{\prime}}(v_{1})\right]  _{i}\,,
\end{align*}
where as above we have replaced the R-matrices by $a(v_{ik})$. Finally we
obtain with (\ref{Xi})%
\[
\left(  uw_{C_{3}}^{D}\right)  _{\mathring{\gamma}}=-\sum_{i=1}^{m}%
f(u_{1}^{\prime}-v_{i})X_{\mathring{\gamma}}^{(i)}(\underline{u},\underline
{v})\,.
\]
Therefore we obtain $\left(  uw_{C_{3}}^{D,i}\right)  _{\mathring{\gamma}%
}(\underline{u},\underline{v})$ in the form of (\ref{uwd3}). The remaining
unwanted terms are due to (\ref{DC'})%
\begin{multline*}
\left(  uw_{C}^{D}\right)  _{\mathring{\gamma}}C_{Q}^{\mathring{\gamma}%
}+uw_{C_{2}}^{D}C_{2,Q}=-\sum_{i=1}^{m}c(v_{i}-u_{1})\\
\times L_{\underline{\mathring{\beta}}}(\underline{v})\Pi_{\gamma^{\prime
}\underline{\beta^{\prime}}}^{\underline{\mathring{\gamma}\mathring{\beta}}%
}(u_{1}^{\prime},\underline{v})\Omega T_{\gamma_{m}}^{\beta_{m}^{\prime}%
}(v_{m})\cdots T_{\gamma_{1}}^{\beta_{1}^{\prime}}(v_{1})\left(  T_{Q}\right)
_{\gamma^{\prime\prime}}^{\gamma^{\prime}}\left(  R_{10}(v_{1i})\cdots
\mathbf{P}_{i0}\cdots R_{m0}(v_{mi})\right)  _{1\cdots1,\mathring{\gamma}%
}^{\gamma^{\prime\prime},\underline{\gamma}}\\
=-\sum_{i=1}^{m}c(v_{i}-u_{1})\left(  L(v_{i}^{\prime},\underline{v}_{i}%
)\Pi(u_{1}^{\prime},\underline{v}_{i})\right)  _{\gamma^{\prime}%
\underline{\beta}_{i}}\Omega\left[  T_{1}^{\beta_{m}^{\prime}}(v_{m})\cdots
T_{1}^{\beta_{1}^{\prime}}(v_{1})\right]  _{i}\prod_{\underset{k\neq i}{k=1}%
}^{m}a(v_{ki})a_{2}(v_{i})\left(  T_{Q}\right)  _{1}^{\gamma^{\prime}}%
\end{multline*}
where (\ref{RPi}), (\ref{3.1a}), (\ref{3.2a}) and $\Omega T_{\mathring{\gamma
}}^{\beta_{i}^{\prime}}(v_{i})=\delta_{\mathring{\gamma}}^{\beta_{i}^{\prime}%
}a_{2}(v_{i})\Omega$ have been used. For $\gamma^{\prime}=1$ this vanishes
because of (\ref{m1}). For $\gamma^{\prime}=\mathring{\gamma}\neq1,\bar{1}$
this gives%
\begin{align*}
\left(  uw_{C}^{D}\right)  _{\mathring{\gamma}}  &  =-\sum_{i=1}^{m}%
c(v_{i}-u_{1})\left(  L(v_{i}^{\prime},\underline{v}_{i})\mathbf{1}_{i}%
\Pi(\underline{v}_{i})\right)  _{\mathring{\gamma}\underline{\beta}_{i}}%
\Omega\left[  T_{1}^{\beta_{m}^{\prime}}(v_{m})\cdots T_{1}^{\beta_{1}%
^{\prime}}(v_{1})\right]  _{i}\prod_{\underset{k\neq i}{k=1}}^{m}%
a(v_{ki})a_{2}(v_{i})\\
&  =-\sum_{i=1}^{m}c(v_{i}-u_{1})X_{\mathring{\gamma}}^{(i)}(\underline
{u},\underline{v}^{(i)})\chi_{i}(\underline{u},\underline{v})
\end{align*}
where $\underline{v}^{(i)}$ means that $v_{i}$ is replaced by $v_{i}^{\prime
}=v_{i}+2/\nu$ and $\chi_{i}(\underline{u},\underline{v})$ is defined by
(\ref{chi}). For $\gamma^{\prime}=\bar{1}$ we get using (\ref{m6}) and
$c(v_{i}-u_{1})f(u_{1}^{\prime}-v_{j})=\left(  c(u_{1}-v_{i})+f(u_{1}^{\prime
}-v_{j})\right)  f(v_{ij}^{\prime})$%
\begin{align*}
uw_{C_{2}}^{D}  &  =-\sum_{i=1}^{m}c(v_{i}-u_{1})\left(  L(v_{i}^{\prime
},\underline{v}_{i})\Pi(u_{1}^{\prime},\underline{v}_{i})\right)  _{\bar
{1}\underline{\beta}_{i}}\Omega\left[  T_{1}^{\beta_{m}^{\prime}}(v_{m})\cdots
T_{1}^{\beta_{1}^{\prime}}(v_{1})\right]  _{i}\prod_{\underset{k\neq i}{k=1}%
}^{m}a(v_{ki})a_{2}(v_{i})\\
&  =-\sum_{i=1}^{m}c(v_{i}-u_{1})\sum_{\underset{j\neq i}{j=1}}^{m}%
f(u_{1}^{\prime}-v_{j})\left(  L(v_{i}^{\prime},v_{j},\underline{v}%
_{ij})\mathbf{\mathring{C}}^{ij}\Pi(\underline{v}_{ij})\right)  _{\underline
{\beta}_{i}}\\
&  \times\Omega\left[  T_{1}^{\beta_{m}^{\prime}}(v_{m})\cdots T_{1}%
^{\beta_{1}^{\prime}}(v_{1})\right]  _{ij}\prod_{\underset{k\neq i}{k=1}}%
^{m}a(v_{ki})\prod_{\underset{k\neq i,j}{k=1}}^{m}a(v_{jk})a_{2}(v_{i}%
)a_{1}(v_{j})\\
&  =-\sum_{i=1}^{m}\sum_{\underset{j\neq i}{j=1}}^{m}\left(  c(u_{1}%
-v_{i})+f(u_{1}^{\prime}-v_{j})\right)  X^{(ij)}(\underline{v}^{(i)})\chi
_{i}(\underline{u},\underline{v})
\end{align*}
with $X^{(ij)}$ given by (\ref{Xij}). Therefore we obtain $\left(
uw_{C}^{D,i}\right)  _{\mathring{\gamma}}(\underline{u},\underline{v})$ and
$uw_{C_{2}}^{D,ij}(\underline{u},\underline{v})$ in the form of (\ref{uwd1})
and (\ref{uwd2}).

\paragraph{The unwanted terms $uw^{A_{3}}$:}

Using $\Omega T{_{Q}{}}_{\bar{1}}^{\gamma}=0$ (see (\ref{1.41}) which also
implies that the $A_{3}$-wanted term vanishes) and Yang-Baxter relations we derive

\noindent$0=%
\begin{array}
[c]{c}%
{\phantom{\rule{1.4676in}{1.2825in}}}%
\unitlength4mm\begin{picture}(10,7)\thicklines\put(3.9,0){$u'_1$}
\put(8.3,0){$u_n$}
\put(6,4.2){$u_1$}
\put(6,2.7){$v_i$}
\put(9.5,4.6){$\bar1$}
\put(9.5,3.1){1}
\put(9.5,2.1){1}
\put(9.5,1.1){1}
\put(4.8,6.1){1}
\put(7.8,6.1){1}
\put(5.8,5.5){$\dots$}
\put(2.8,4.5){$_{\dots}$}
\put(9,6){\oval(8,2)[lb]}
\put(1.5,0){\oval(7,10)[rt]}
\put(8,0){\line(0,1){6}}
\put(1,4.6){$\bar1$}
\put(9,6){\oval(11,5)[lb]}
\put(9,6){\oval(12.5,7)[lb]}
\put(9,6){\oval(14,9)[lb]}
\put(1.8,6){$\framebox(2,1){$\Pi$}$}
\put(2.,7){\line(0,1){.5}}
\put(2.8,7){\line(0,1){.5}}
\put(3.5,7){\line(0,1){.5}}
\end{picture}%
\end{array}
=%
\begin{array}
[c]{c}%
{\phantom{\rule{1.6754in}{1.2699in}}}%
\unitlength4mm\begin{picture}(12,8)(1,0)\thicklines\put(2.9,0){$u'_1$}
\put(5.6,0){$u_n$}
\put(7.4,0){$1~\,1~\,1$}
\put(10,6.5){\oval(13,9)[lb]}
\put(2,1){\oval(3,2)[rt]}
\put(7.6,4.4){$v_i$}
\put(4.5,1.2){$u_1$}
\put(6.5,1.2){$\gamma$}
\put(10.6,1.8){$\bar1$}
\put(7.5,3){$_{\dots}$}
\put(8.8,3){$_{\dots}$}
\put(4.2,6){$\dots$}
\put(6,1){\line(0,1){5.5}}
\put(3.5,1){\oval(12,9)[rt]}
\put(3.5,6){\oval(2,1)[lb]}
\put(3.5,1){\oval(10,7)[rt]}
\put(3.5,6){\oval(3.5,3)[lb]}
\put(3.5,1){\oval(8,5)[rt]}
\put(3.5,6){\oval(5,5)[lb]}
\put(1.4,1.7){$\bar1$}
\put(.8,6){$\framebox(2,1){$\Pi$}$}
\put(.9,7){\line(0,1){.5}}
\put(1.7,7){\line(0,1){.5}}
\put(2.5,7){\line(0,1){.5}}
\end{picture}%
\end{array}
$

\noindent or in terms of formulas (with $T_{\underline{\beta},\gamma}%
^{\gamma^{\prime},\underline{\beta}^{\prime}}(\underline{v},u)=\left(
R_{10}(v_{1}-u)\dots R_{m0}(z_{m}-u)\right)  _{\underline{\beta},\gamma
}^{\gamma^{\prime},\underline{\beta}^{\prime}}$\newline and $T_{\beta}%
^{\beta^{\prime}}(\underline{u},v)=\left(  R_{10}(u_{1}-v)\dots R_{n0}%
(u_{n}-v)\right)  _{\beta}^{\beta^{\prime}}$ where the quantum space indices
are suppressed)%
\begin{align}
0  &  =\left(  L(\underline{v})\Pi(\underline{v})\right)  _{\underline{\beta
}^{\prime}}T_{\underline{\beta},\gamma}^{\bar{1},\underline{\beta}^{\prime}%
}(\underline{v},u_{1}^{\prime})\Omega T_{Q,\bar{1}}^{\gamma}(\underline
{u})T_{1}^{\beta_{m}}(\underline{u}^{\prime},v_{m})\cdots T_{1}^{\beta_{1}%
}(\underline{u}^{\prime},v_{1})\nonumber\\
&  =\left(  L(\underline{v})\Pi(\underline{v})\right)  _{\underline{\beta}%
}\Omega T_{\beta_{m}^{\prime}}^{\beta_{m}}(\underline{u},v_{m})\cdots
T_{\beta_{1}^{\prime}}^{\beta_{1}}(\underline{u},v_{1})T_{Q,\gamma}^{\bar{1}%
}(\underline{u})T_{1\cdots1,\bar{1}}^{\gamma,\underline{\beta}^{\prime}%
}(\underline{v},u_{1})\label{A3}\\
&  =\prod_{k=1}^{m}\left(  1+d(v_{k}-u_{1})\right)  \left(  \left(
L(\underline{v})\Pi(\underline{v})\right)  _{\underline{\beta}}\Omega
T_{1}^{\beta_{m}}(\underline{u},v_{m})\cdots T_{1}^{\beta_{1}}(\underline
{u},v_{1})A_{3,Q}(\underline{u})-uw^{A_{3}}\right)  \,,\nonumber
\end{align}
where the term written down comes from $\gamma=\bar{1}$ in (\ref{A3}). It has
been used that
\[
T_{1\cdots1,\bar{1}}^{\bar{1},\underline{\beta}^{\prime}}(\underline{v}%
,u_{1})=\left(  R_{10}(v_{1}-u_{1})\dots R_{m0}(v_{m}-u_{1})\right)
_{1\cdots1,\bar{1}}^{\bar{1},\underline{\beta}^{\prime}}=\prod_{k=1}%
^{m}\left(  1+d(v_{k}-u_{1})\right)  \delta_{1}^{\beta_{k}^{\prime}}\,.
\]
The different unwanted terms are due to different values of $\gamma$ in
(\ref{A3}): for $\gamma=\mathring{\gamma}\neq1,\bar{1}$ we get $\left(
uw_{C_{3}}^{A_{3}}\right)  _{\mathring{\gamma}^{\prime}}\mathbf{C}%
^{\mathring{\gamma}\mathring{\gamma}^{\prime}}=\left(  uw_{C_{3}}^{A_{3}%
}\right)  ^{\mathring{\gamma}}$%
\[
\left(  uw_{C_{3}}^{A_{3}}\right)  ^{\mathring{\gamma}}\prod_{k=1}^{m}\left(
1+d(v_{k}-u_{1})\right)  =-\left(  L(\underline{v})\Pi(\underline{v})\right)
_{\underline{\beta}}\Omega T_{\beta_{m}^{\prime}}^{\beta_{m}}(\underline
{u},v_{m})\cdots T_{\beta_{1}^{\prime}}^{\beta_{1}}(\underline{u}%
,v_{1})T_{1\cdots1,\bar{1}}^{\mathring{\gamma},\underline{\beta}^{\prime}%
}(\underline{v},u_{1})
\]
and for $\gamma=1$ we get $uw_{C_{2}}^{A_{3}}$
\[
\left(  uw_{C_{2}}^{A_{3}}\right)  \prod_{k=1}^{m}\left(  1+d(v_{k}%
-u_{1})\right)  =-\left(  L(\underline{v})\Pi(\underline{v})\right)
_{\underline{\beta}}\Omega T_{\beta_{m}^{\prime}}^{\beta_{m}}(\underline
{u},v_{m})\cdots T_{\beta_{1}^{\prime}}^{\beta_{1}}(\underline{u}%
,v_{1})T_{1\cdots1,\bar{1}}^{1,\underline{\beta}^{\prime}}(\underline{v}%
,u_{1})\,.
\]
The determination of these unwanted terms is not so direct compared to those
of the $A$- and $D$-unwanted terms. In particular for $uw_{C_{2}}^{A_{3}}$ we
use more complicated arguments.

To calculate $uw_{C_{3}}^{A_{3}}$ we use special components of the Yang Baxter
relation (\ref{TTS})
\[
R_{\alpha\beta}^{\bar{1}1}(1/\nu-1)T_{\underline{\beta}^{\prime}\bar{1}%
}^{\beta\underline{\beta}}(\underline{v},u)T_{1\cdots1,\mathring{\alpha}%
}^{\alpha\underline{\beta}^{\prime}}(\underline{v},u+1/\nu-1)=T_{\underline
{\beta}^{\prime}\alpha}^{1\underline{\beta}}(\underline{v},u+1/\nu
-1)T_{1\cdots1,\beta}^{\bar{1}\underline{\beta}^{\prime}}(\underline
{v},u)R_{\mathring{\alpha}\bar{1}}^{\beta\alpha}(1/\nu-1).
\]
Using $d(1/\nu-1)=-1,~T_{1\cdots1,\mathring{\alpha}}^{\alpha\underline{\beta
}^{\prime}}(u^{\prime})=\delta_{\mathring{\alpha}}^{\alpha}\mathbf{1}%
_{1\cdots1}^{\underline{\beta}^{\prime}}$ for $\alpha\neq1$ and \newline%
$T_{1\cdots1,\beta}^{\bar{1}\underline{\beta}^{\prime}}(u)=\delta_{\beta
}^{\bar{1}}\mathbf{1}_{1\cdots1}^{\underline{\beta}^{\prime}}\prod_{k=1}%
^{m}\left(  1+d(v_{k}-u)\right)  $ we derive%
\begin{multline*}
T_{1\cdots1\bar{1}}^{\mathring{\gamma}\underline{\beta}}(\underline
{v},u)=-T_{1\cdots1\mathring{\gamma}^{\prime}}^{1\underline{\beta}}%
(\underline{v},u+1/\nu-1)\prod_{k=1}^{m}\left(  1+d(v_{k}-u)\right)
\mathbf{C}^{\mathring{\gamma}\mathring{\gamma}^{\prime}}\\
=-\prod_{k=1}^{m}\left(  1+d(v_{k}-u)\right)  \sum_{i=1}^{m}f(u-v_{i})\left(
R_{1a}(v_{1i})\dots\mathbf{P}_{ia}\dots R_{1a}(v_{mi})\right)  _{1\cdots
1\mathring{\gamma}^{\prime}}^{1\underline{\beta}}\mathbf{C}^{\mathring{\gamma
}\mathring{\gamma}^{\prime}}%
\end{multline*}
For the last equality (\ref{Sb}) and $c(v-\left(  u+1/\nu-1\right)  )=f(u-v)$
have been used. Therefore%
\begin{align*}
\left(  uw_{C_{3}}^{A_{3}}\right)  _{\mathring{\gamma}}  &  =\sum_{i=1}%
^{m}f(u_{1}-v_{i})\left(  L(\underline{v})\Pi(\underline{v})\right)
_{\underline{\beta}}\Omega T_{\beta_{m}^{\prime}}^{\beta_{m}}(\underline
{u},v_{m})\cdots T_{\beta_{1}^{\prime}}^{\beta_{1}}(\underline{u},v_{1})\\
&  \times\left(  R_{1a}(v_{1i})\dots\mathbf{P}_{ia}\dots R_{1a}(v_{mi}%
)\right)  _{1\cdots1,\mathring{\gamma}}^{1\underline{\beta}}\\
&  =\sum_{i=1}^{m}f(u_{1}-v_{i})\left(  L(v_{i}^{\prime},\underline{v}_{i}%
)\Pi(\underline{v}_{i})\right)  _{\mathring{\gamma}\underline{\beta}_{i}}%
\prod_{\underset{k\neq i}{k=1}}^{m}a(v_{ki})a_{2}(v_{i})\Omega\left[
T_{1}^{\beta_{m}}(v_{m})\cdots T_{1}^{\beta_{1}}(v_{1})\right]  _{i}\\
&  =\sum_{i=1}^{m}f(u_{1}-v_{i})X_{\mathring{\gamma}}^{(i)}(\underline
{u},\underline{v}^{(i)})\chi_{i}(\underline{u},\underline{v})
\end{align*}
We obtain $\left(  uw_{C_{3}}^{A_{3},i}\right)  _{\mathring{\gamma}%
}(\underline{u},\underline{v})$ in the form of (\ref{uwa3}). In order to
calculate%
\begin{equation}
\left(  uw_{C_{2}}^{A_{3}}\right)  \prod_{k=1}^{m}\left(  1+d(v_{k}%
-u_{1})\right)  =-\left(  L(\underline{v})\Pi(\underline{v})\right)
_{\underline{\beta}}\Omega T_{\beta_{m}^{\prime}}^{\beta_{m}}(\underline
{u},v_{m})\cdots T_{\beta_{1}^{\prime}}^{\beta_{1}}(\underline{u}%
,v_{1})T_{1\cdots1,\bar{1}}^{1,\underline{\beta}^{\prime}}(\underline{v}%
,u_{1}) \label{A32}%
\end{equation}
we prove

\begin{lemma}
\label{l1}The unwanted term $uw_{C_{2}}^{A_{3}}$ is of the form%
\[
uw_{C_{2}}^{A_{3}}=\sum_{i=1}^{m}\sum_{j=i,j\neq i}^{m}\left(  f_{ij}%
(\underline{v})\Pi(\underline{v}_{ij})\right)  _{\underline{\beta}_{ij}}%
a_{2}(v_{i})a_{2}(v_{j})\Omega\left[  T_{1}^{\beta_{m}}(v_{m})\dots
T_{1}^{\beta_{1}}(v_{1})\right]  _{ij}%
\]

\end{lemma}

\begin{proof}
Using (\ref{Sb}) we have%
\begin{align*}
&  \Pi_{\underline{\beta}}^{\underline{\mathring{\beta}}}\Omega T_{\beta
_{m}^{\prime}}^{\beta_{m}}(\underline{u},v_{m})\cdots T_{\beta_{1}^{\prime}%
}^{\beta_{1}}(\underline{u},v_{1})T_{1\cdots1,\bar{1}}^{1,\underline{\beta
}^{\prime}}(\underline{v},u_{1})\\
&  =\sum_{j=1}^{m}d(v_{i}-u_{1})\left(  \mathring{R}(v_{ii-1})\dots
\mathring{R}(v_{i1})\Pi(v_{i},\underline{v}_{i})\right)  _{\underline{\beta}%
}^{\underline{\mathring{\beta}}}\Omega\left[  T_{1}^{\beta_{m}}(v_{m})\cdots
T_{1}^{\beta_{1}}(v_{1})\right]  _{i}T_{\bar{1}}^{\beta_{i}}(v_{i})
\end{align*}
The $c(v_{i}-u_{1})$-terms do not contribute because they produce terms like
$\Omega B\dots=0$. The Yang-Baxter equation for $R$ implies%
\begin{align}
&  \left(  \mathring{R}(v_{ji})\Pi(v_{ji})\right)  _{\dots\alpha\beta\dots
}^{\underline{\mathring{\beta}}}\Omega\left[  \dots T_{\bar{1}}^{\beta}%
(v_{i})T_{1}^{\alpha}(v_{j})\dots\right] \nonumber\\
&  =\left(  \Pi(v_{ij})\right)  _{\dots\beta\alpha\dots}^{\underline
{\mathring{\beta}}}\Omega\left[  \dots T_{\alpha^{\prime}}^{\alpha}%
(v_{j})T_{\beta^{\prime}}^{\beta}(v_{i})\dots\right]  R_{1\bar{1}}%
^{\beta^{\prime}\alpha^{\prime}}(v_{ji})\nonumber\\
&  =\left(  \Pi(v_{ij})\right)  _{\dots\beta\alpha\dots}^{\underline
{\mathring{\beta}}}\Omega\left[  \dots T_{1}^{\alpha}(v_{j})T_{\bar{1}}%
^{\beta}(v_{i})\dots\right]  (1+d(v_{ji}))\nonumber\\
&  +\left(  \Pi(v_{ij})\right)  _{\dots\beta\alpha\dots}^{\underline
{\mathring{\beta}}}\Omega\left[  \dots T_{\bar{1}}^{\alpha}(v_{j})T_{1}%
^{\beta}(v_{i})\dots\right]  (c(v_{ji})+d(v_{ji}))\nonumber\\
&  +\left(  \Pi(v_{ij})\right)  _{\dots\beta\alpha\dots}^{\underline
{\mathring{\beta}}}\Omega\left[  \dots T_{\mathring{\gamma}}^{\alpha}%
(v_{j})T_{\overline{\mathring{\gamma}}}^{\beta}(v_{i})\dots\right]
d(v_{ji})\mathbf{C}^{\overline{\mathring{\gamma}}\mathring{\gamma}}
\label{a.66}%
\end{align}
Iterating this formula we move the $T_{\bar{1}}$-operators to the left and
finally $\Pi\Omega T_{\bar{1}}\dots=0$. Therefore $\Pi_{\underline{\beta}%
}^{\underline{\mathring{\beta}}}\Omega T_{1}^{\beta_{m}}(v_{m})\cdots
T_{\bar{1}}^{\beta_{1}}(v_{1})$ is a sum of terms like%
\[
\Pi_{\underline{\beta}}^{\underline{\mathring{\beta}}}\Omega T_{1}^{\beta_{m}%
}(v_{m})\dots T_{\mathring{\gamma}}^{\beta_{i}}(v_{i})T_{\overline
{\mathring{\gamma}}}^{\beta_{j}}(v_{j})\cdots T_{1}^{\beta_{1}}(v_{1}%
)\mathbf{C}^{\overline{\mathring{\gamma}}\mathring{\gamma}}%
\]
Similar as for $T_{\bar{1}}^{\beta}$ we can move the two $T_{\mathring{\gamma
}}^{\beta_{i}}$-operators to the left using%
\begin{align*}
&  \left(  \mathring{R}(v_{ji})\Pi(v_{ji})\right)  _{\dots\alpha\beta\dots
}^{\underline{\mathring{\beta}}}\Omega\left(  \dots T_{\mathring{\gamma}%
}^{\alpha}(v_{i})T_{1}^{\beta}(v_{j})\dots\right) \\
&  =\left(  \Pi(v_{ij})\right)  _{\dots\beta\alpha\dots}^{\underline
{\mathring{\beta}}}\Omega\left(  \dots T_{\alpha^{\prime}}^{\alpha}%
(v_{j})T_{\beta^{\prime}}^{\beta}(v_{i})\dots\right)  R_{1\mathring{\gamma}%
}^{\beta^{\prime}\alpha^{\prime}}(v_{ji})\\
&  =\left(  \Pi(v_{ij})\right)  _{\dots\beta\alpha\dots}^{\underline
{\mathring{\beta}}}\Omega\left(  \left(  \dots T_{1}^{\alpha}(v_{j}%
)T_{\mathring{\gamma}}^{\beta}(v_{i})\dots\right)  +\left(  \dots
T_{\mathring{\gamma}}^{\alpha}(v_{j})T_{1}^{\beta}(v_{i})\dots\right)
c(v_{ji})\right)
\end{align*}
and finally%
\[
\left(  \Pi\right)  _{\dots\alpha\beta}^{\underline{\mathring{\beta}}}\Omega
T_{\mathring{\gamma}}^{\beta}(v_{i})T_{\overline{\mathring{\gamma}}}^{\alpha
}(v_{j})\cdots=\left(  \Pi\right)  _{\dots\alpha\beta}^{\underline
{\mathring{\beta}}}\delta_{\mathring{\gamma}}^{\beta}\delta_{\overline
{\mathring{\gamma}}}^{\alpha}a_{2}(v_{i})a_{2}(v_{j})\Omega\cdots
\]

\end{proof}

We calculate $f_{12}(\underline{v})$ and $f_{21}(\underline{v})$, the other
$f_{ij}(\underline{v})$ are due to (\ref{3.1c}) related to $f_{12}%
(\underline{v})$ by the symmetry%
\[
\left(  L(v_{ji})\Pi(v_{ji})\right)  _{\dots\beta^{\prime}\alpha^{\prime}%
\dots}^{\underline{\mathring{\beta}}}\left[  \dots T_{1}^{\alpha^{\prime}%
}(v_{i})T_{1}^{\beta^{\prime}}(v_{j})\dots\right]  =\left(  L(v_{ij}%
)\Pi(v_{ij})\right)  _{\dots\alpha\beta\dots}^{\underline{\mathring{\beta}}%
}\left[  \dots T_{1}^{\beta}(v_{j})T_{1}^{\alpha}(v_{i})\dots\right]  .
\]
We insert in (\ref{A32}) the intermediate states $\gamma=1,0,\bar{1}$ behind
the second R-matrix in the monodromy matrix $T_{1\cdots1,\bar{1}}^{1,\cdots
}(\underline{v},u_{1})$%
\begin{multline*}
uw_{C_{2}}^{A_{3}}\prod_{k=1}^{m}\left(  1+d(v_{k}-u_{1})\right)  =-\left(
L\Pi\right)  _{\underline{\beta}}\Omega T_{\beta_{m}^{\prime}}^{\beta_{m}%
}(\underline{u},v_{m})\cdots T_{\beta_{1}^{\prime}}^{\beta_{1}}(\underline
{u},v_{1})C_{2Q}\\
\times\left(  R(v_{1}-u_{1})R(v_{2}-u_{1})\right)  _{11,\gamma}^{1,\beta
_{1}^{\prime}\beta_{2}^{\prime}}\left(  R(v_{3}-u_{1})\dots R(v_{m}%
-u_{1})\right)  _{1\dots1,\bar{1}}^{\gamma,\beta_{3}^{\prime}\dots\beta
_{m}^{\prime}}\,.
\end{multline*}
\label{here}Similar as in the proof of lemma \ref{l1} one can show that for
$\gamma=1$ there are only contributions to $f_{ij}$ with $2<i<j$ and for
$\gamma=0$ there are only contributions to $f_{1j}$ or $f_{2j}$ with $2<j$. So
we have to consider only $\gamma=\bar{1}$ where we use
\[
\left(  R(v_{3}-u_{1})\dots R(v_{m}-u_{1})\right)  _{1\dots1,\bar{1}}^{\bar
{1},\beta_{3}^{\prime}\dots\beta_{m}^{\prime}}=\prod_{k=3}^{m}\left(
1+d(v_{k}-u_{1})\right)  \mathbf{1}_{1\dots1}^{\beta_{3}^{\prime}\dots
\beta_{m}^{\prime}}%
\]
to obtain for the contribution of $f_{12}$ and $f_{21}$ to $uw_{C_{2}}^{A_{3}%
}\prod_{k=1}^{m}\left(  1+d(v_{k}-u_{1})\right)  $%
\begin{multline*}
-\left(  L\Pi\right)  _{\underline{\beta}}\Omega T_{1}^{\beta_{m}}%
(v_{m})\cdots T_{1}^{\beta_{3}}(v_{3})T_{\beta_{2}^{\prime}}^{\beta_{2}}%
(v_{2})T_{\beta_{1}^{\prime}}^{\beta_{1}}(v_{1})C_{2Q}\\
\times\left(  R(v_{1}-u_{1})R(v_{2}-u_{1})\right)  _{11,\bar{1}}^{1,\beta
_{1}^{\prime}\beta_{2}^{\prime}}\prod_{k=3}^{m}\left(  1+d(v_{k}%
-u_{1})\right)
\end{multline*}
Using (\ref{Sb}) we obtain (similar as in the proof of lemma \ref{l1}%
\begin{align*}
&  \Pi_{\underline{\beta}}^{\underline{\mathring{\beta}}}\Omega T_{1}%
^{\beta_{m}}(v_{m})\cdots T_{1}^{\beta_{3}}(v_{3})T_{\beta_{2}^{\prime}%
}^{\beta_{2}}(v_{2})T_{\beta_{1}^{\prime}}^{\beta_{1}}(v_{1})\left(
R(v_{1}-u_{1})R(v_{2}-u_{1})\right)  _{11,\bar{1}}^{1,\beta_{1}^{\prime}%
\beta_{2}^{\prime}}\\
&  =c(v_{1}-u_{1})\left(  \mathring{R}(v_{21})\Pi\right)  _{\underline{\beta}%
}^{\underline{\mathring{\beta}}}\Omega T_{1}^{\beta_{m}}(v_{m})\cdots
T_{1}^{\beta_{3}}(v_{3})T_{\bar{1}}^{\beta_{2}}(v_{1})T_{1}^{\beta_{1}}%
(v_{2})\\
&  +c(v_{2}-u_{1})\Pi_{\underline{\beta}}^{\underline{\mathring{\beta}}}\Omega
T_{1}^{\beta_{m}}(v_{m})\cdots T_{1}^{\beta_{3}}(v_{3})T_{\bar{1}}^{\beta_{2}%
}(v_{2})T_{1}^{\beta_{1}}(v_{1})a(v_{12})\\
&  +d(v_{1}-u_{1})\Pi_{\underline{\beta}}^{\underline{\mathring{\beta}}}\Omega
T_{1}^{\beta_{m}}(v_{m})\cdots T_{1}^{\beta_{3}}(v_{3})T_{1}^{\beta_{2}}%
(v_{2})T_{\bar{1}}^{\beta_{1}}(v_{1})a(v_{12})\\
&  +d(v_{2}-u_{1})\left(  \mathring{R}(v_{21})\Pi\right)  _{\underline{\beta}%
}^{\underline{\mathring{\beta}}}\Omega T_{1}^{\beta_{m}}(v_{m})\cdots
T_{1}^{\beta_{3}}(v_{3})T_{1}^{\beta_{2}}(v_{1})T_{\bar{1}}^{\beta_{1}}%
(v_{2})\,.
\end{align*}
The $c$-terms do not contribute to $f_{12}$ and $f_{21}$. For the first
$d$-term we may replace (because of (\ref{a.66}))
\[
\Omega\left[  \dots T_{1}^{\alpha}(v_{2})T_{\bar{1}}^{\beta}(v_{1})\right]
\rightarrow\Omega\left[  \dots T_{\mathring{\gamma}}^{\alpha}(v_{2}%
)T_{\overline{\mathring{\gamma}}}^{\beta}(v_{1})\right]  f(v_{12}%
)\mathbf{C}^{\overline{\mathring{\gamma}}\mathring{\gamma}}+\ldots
\]
and obtain%
\begin{align*}
&  d(v_{1}-u_{1})\Pi_{\underline{\beta}}^{\underline{\mathring{\beta}}}\Omega
T_{1}^{\beta_{m}}(v_{m})\cdots T_{1}^{\beta_{3}}(v_{3})T_{1}^{\beta_{2}}%
(v_{2})T_{\bar{1}}^{\beta_{1}}(v_{1})a(v_{12})\\
&  =-d(v_{1}-u_{1})f(v_{12})\Pi_{\underline{\beta}}^{\underline{\mathring
{\beta}}}\Omega T_{1}^{\beta_{m}}(v_{m})\cdots T_{1}^{\beta_{3}}%
(v_{3})T_{\mathring{\gamma}}^{\beta_{2}}(v_{2})T_{\overline{\mathring{\gamma}%
}}^{\beta_{1}}(v_{1})a(v_{12})\mathbf{C}^{\overline{\mathring{\gamma}%
}\mathring{\gamma}}+\ldots
\end{align*}
where the missing term again does not contribute to $f_{12}$ and $f_{21}$.
Similarly for the second $d$-term
\begin{align*}
&  d(v_{2}-u_{1})\left(  \mathring{R}(v_{21})\Pi\right)  _{\underline{\beta}%
}^{\underline{\mathring{\beta}}}\Omega T_{1}^{\beta_{m}}(v_{m})\cdots
T_{1}^{\beta_{3}}(v_{3})T_{1}^{\beta_{2}}(v_{1})T_{\bar{1}}^{\beta_{1}}%
(v_{2})\\
&  =-d(v_{2}-u_{1})f(v_{21})\left(  \mathring{R}(v_{21})\Pi\right)
_{\underline{\beta}}^{\underline{\mathring{\beta}}}\Omega T_{1}^{\beta_{m}%
}(v_{m})\cdots T_{1}^{\beta_{3}}(v_{3})T_{\mathring{\gamma}}^{\beta_{2}}%
(v_{1})T_{\overline{\mathring{\gamma}}}^{\beta_{1}}(v_{2})\mathbf{C}%
^{\overline{\mathring{\gamma}}\mathring{\gamma}}+\ldots
\end{align*}
In order to move $T_{\mathring{\gamma}}^{\beta}(v_{2})T_{\overline
{\mathring{\gamma}}}^{\alpha}(v_{1})$ to the left we consider%
\begin{align*}
&  \left(  \left(  \mathring{R}(v_{32})\dots\mathring{R}(v_{m2})\right)
\left(  \mathring{R}(v_{31})\dots\mathring{R}(v_{m1})\right)  \Pi\right)
_{\underline{\beta}}^{\underline{\mathring{\beta}}}\Omega T_{\mathring{\gamma
}}^{\beta_{m}}(v_{2})T_{\overline{\mathring{\gamma}}}^{\beta_{m-1}}%
(v_{1})T_{1}^{\beta_{2}}(v_{m})\dots T_{1}^{\beta_{1}}(v_{3})\mathbf{C}%
^{\overline{\mathring{\gamma}}\mathring{\gamma}}\\
&  =\Pi_{\underline{\beta}}^{\underline{\mathring{\beta}}}\Omega T_{\beta
_{m}^{\prime}}^{\beta_{m}}(v_{m})\dots T_{\beta_{3}^{\prime}}^{\beta_{3}%
}(v_{3})T_{\beta_{2}^{\prime}}^{\beta_{2}}(v_{2})T_{\beta_{1}^{\prime}}%
^{\beta_{1}}(v_{1})\\
&  \times\left(  \left(  R(v_{32})\dots R(v_{m2})\right)  \left(
R(v_{31})\dots R(v_{m1})\right)  \right)  _{1\dots1,\overline{\mathring
{\gamma}}\mathring{\gamma}}^{\underline{\beta}^{\prime}}\mathbf{C}%
^{\overline{\mathring{\gamma}}\mathring{\gamma}}\\
&  =\Pi_{\underline{\beta}}^{\underline{\mathring{\beta}}}\Omega T_{1}%
^{\beta_{m}}(v_{m})\dots T_{1}^{\beta_{3}}(v_{3})T_{\mathring{\gamma}}%
^{\beta_{2}}(v_{2})T_{\overline{\mathring{\gamma}}}^{\beta_{1}}(v_{1}%
)\mathbf{C}^{\overline{\mathring{\gamma}}\mathring{\gamma}}+\ldots
\end{align*}
therefore%
\begin{align*}
&  d(v_{1}-u_{1})\left(  L\Pi\right)  _{\underline{\beta}}\Omega T_{1}%
^{\beta_{m}}(v_{m})\cdots T_{1}^{\beta_{3}}(v_{3})T_{1}^{\beta_{2}}%
(v_{2})T_{\bar{1}}^{\beta_{1}}(v_{1})a(v_{12})\\
&  =-\left(  L(\underline{v})\left(  \mathring{R}(v_{32})\dots\mathring
{R}(v_{m2})\right)  \left(  \mathring{R}(v_{31})\dots\mathring{R}%
(v_{m1})\right)  \Pi\left(  \underline{v}_{12}\right)  \mathbf{1}%
_{1}\mathbf{1}_{2}\right)  _{\underline{\beta}_{12}\beta_{1}\beta_{2}}\\
&  \times d(v_{1}-u_{1})f(v_{12})a(v_{12})\Omega T_{\mathring{\gamma}}%
^{\beta_{2}}(v_{2})T_{\overline{\mathring{\gamma}}}^{\beta_{1}}(v_{1}%
)T_{1}^{\beta_{m}}(v_{m})\dots T_{1}^{\beta_{3}}(v_{3})\mathbf{C}%
^{\overline{\mathring{\gamma}}\mathring{\gamma}}+\ldots\\
&  =-d(v_{1}-u_{1})f(v_{12})a(v_{12})\left(  L(v_{1}^{\prime},v_{2}^{\prime
},\underline{v}_{12})\mathbf{C}^{12}\Pi\left(  \underline{v}_{12}\right)
\right)  _{\underline{\beta}_{12}}\prod_{k=3}^{m}a(v_{k2})a(v_{k1})\\
&  \times a_{2}(v_{2})a_{2}(v_{1})\Omega T_{1}^{\beta_{m}}(v_{m})\dots
T_{1}^{\beta_{3}}(v_{3})+\ldots\\
&  =-d(v_{1}-u_{1})X^{(12)}(v_{1}^{\prime},v_{2}^{\prime},\underline{v}%
_{12})\frac{a_{2}(v_{2})a_{2}(v_{1})\prod_{k=3}^{m}a(v_{k2})a(v_{k1})}%
{a_{1}(v_{1}^{\prime})a_{1}(v_{2}^{\prime})\prod_{k=3}^{m}a(v_{2k}^{\prime
})a(v_{1k}^{\prime})}+\ldots\\
&  =-d(v_{1}-u_{1})X^{(12)}(v_{1}^{\prime},v_{2}^{\prime},\underline{v}%
_{12})\chi_{1}(\underline{u},\underline{v}^{(2)})\chi_{2}(\underline
{u},\underline{v})+\ldots
\end{align*}
where the missing terms again do not contribute to $f_{12}$ or $f_{21}$.
Similarly
\begin{align*}
&  d(v_{2}-u_{1})\left(  L\mathring{R}(v_{21})\Pi\right)  _{\underline{\beta}%
}\Omega T_{1}^{\beta_{m}}(v_{m})\cdots T_{1}^{\beta_{3}}(v_{3})T_{1}%
^{\beta_{2}}(v_{1})T_{\bar{1}}^{\beta_{1}}(v_{2})\\
&  =-\left(  L(\underline{v})\left(  \mathring{R}(v_{21})\mathring{R}%
(v_{31})\dots\mathring{R}(v_{m1})\right)  \left(  \mathring{R}(v_{32}%
)\dots\mathring{R}(v_{m2})\right)  \mathbf{C}^{21}\Pi\left(  \underline
{v}_{12}\right)  \right)  _{\underline{\beta}_{12}}^{\underline{\mathring
{\beta}}}\\
&  \times d(v_{2}-u_{1})f(v_{21})a_{2}(v_{2})a_{2}(v_{1})\Omega T_{1}%
^{\beta_{m}}(v_{m})\dots T_{1}^{\beta_{3}}(v_{3})+\ldots\\
&  =-d(v_{2}-u_{1})f(v_{21})\left(  L(v_{2}^{\prime},v_{1}^{\prime}%
,\underline{v}_{12})\mathbf{C}^{21}\Pi\left(  \underline{v}_{12}\right)
\right)  _{\underline{\beta}_{12}}a(v_{21})\prod_{k=3}^{m}a(v_{k2})a(v_{k1})\\
&  \times a_{2}(v_{2})a_{2}(v_{1})\Omega T_{1}^{\beta_{m}}(v_{m})\dots
T_{1}^{\beta_{3}}(v_{3})+\ldots\\
&  =-d(v_{2}-u_{1})X^{(21)}(v_{2}^{\prime},v_{1}^{\prime},\underline{v}%
_{12})\chi_{2}(\underline{u},\underline{v}^{(1)})\chi_{1}(\underline
{u},\underline{v})+\ldots
\end{align*}
Note that $\chi_{j}(\underline{u},\underline{v}^{(i)})\chi_{i}(\underline
{u},\underline{v})=\chi_{i}(\underline{u},\underline{v}^{(j)})\chi
_{j}(\underline{u},\underline{v})$ and%
\[
X^{(21)}(v_{2}^{\prime},v_{1}^{\prime},\underline{v}_{12})=-X^{(12)}%
(v_{1}^{\prime},v_{2}^{\prime},\underline{v}_{12})
\]
because of the identities%
\begin{align*}
L(v_{2}^{\prime},v_{1}^{\prime},\underline{v}_{12})\mathbf{\mathring{C}}^{21}
&  =\frac{\mathring{R}_{0}(v_{21})}{a(v_{21})}L(v_{1}^{\prime},v_{2}^{\prime
},\underline{v}_{12})\mathbf{\mathring{C}}^{12}\\
f(v)\mathring{R}_{0}(v)  &  =-a(-v)f(-v)\,.
\end{align*}
Finally
\begin{align*}
&  uw_{C_{2}}^{A_{3},12}+uw_{C_{2}}^{A_{3},21}\\
&  =-\left(  \frac{d(v_{1}-u_{1})}{\left(  1+d(v_{1}-u_{1})\right)  \left(
1+d(v_{2}-u_{1})\right)  }-\frac{d(v_{2}-u_{1})}{\left(  1+d(v_{1}%
-u_{1})\right)  \left(  1+d(v_{2}-u_{1})\right)  }\right) \\
&  \times X^{(12)}(v_{1}^{\prime},v_{2}^{\prime},\underline{v}_{12})\chi
_{1}(\underline{u},\underline{v}^{(2)})\chi_{2}(\underline{u},\underline{v})\\
&  =\left(  f(u_{1}-v_{1})-f(u_{1}-v_{2})\right)  X^{(12)}(v_{1}^{\prime
},v_{2}^{\prime},\underline{v}_{12})\chi_{1}(\underline{u},\underline{v}%
^{(2)})\chi_{2}(\underline{u},\underline{v})\\
&  =\left(  f(u_{1}-v_{1})X^{(12)}(v_{1}^{\prime},v_{2}^{\prime},\underline
{v}_{12})+f(u_{1}-v_{2})X^{(21)}(v_{2}^{\prime},v_{1}^{\prime},\underline
{v}_{12})\right)  \chi_{1}(\underline{u},\underline{v}^{(2)})\chi
_{2}(\underline{u},\underline{v})
\end{align*}
we set
\begin{align*}
uw_{C_{2}}^{A_{3},12}  &  =-f(u_{1}-v_{2})X^{(12)}(v_{1}^{\prime}%
,v_{2}^{\prime},\underline{v}_{12})\chi_{1}(\underline{u},\underline{v}%
^{(2)})\chi_{2}(\underline{u},\underline{v})\\
uw_{C_{2}}^{A_{3},21}  &  =-f(u_{1}-v_{1})X^{(21)}(v_{2}^{\prime}%
,v_{1}^{\prime},\underline{v}_{12})\chi_{2}(\underline{u},\underline{v}%
^{(1)})\chi_{1}(\underline{u},\underline{v})\\
uw_{C_{2}}^{A_{3},ij}  &  =-f(u_{1}-v_{j})X^{(ij)}(v_{i}^{\prime}%
,v_{j}^{\prime},\underline{v}_{ij})\chi_{i}(\underline{u},\underline{v}%
^{(j)})\chi_{j}(\underline{u},\underline{v})
\end{align*}
which satisfy the desired symmetry. Therefore we obtain $uw_{C_{2}}^{A_{3}%
,ij}(\underline{u},\underline{v})$ in the form of (\ref{uwa32}).

\section{Proof of Theorem
\protect\ref{HW}%
}

\label{a5}

\begin{proof}

\begin{enumerate}
\item The weights (\ref{w}) of the reference state $\Omega$ (\ref{omega}) are
\[
w=\left(  n=n_{0},0,\dots,0\right)
\]
In level $k=1,\dots,\left[  \left(  N-3\right)  /2\right]  $ of the Bethe
ansatz the weights are changed as
\[
w_{k}\rightarrow w_{k}-n_{k},~w_{k+1}\rightarrow w_{k+1}+n_{k}.
\]
This means the states $\Phi_{\underline{\alpha}}^{\underline{\mathring{\beta}%
}}(\underline{u},\underline{v})$ of (\ref{Phi}) are eigenvectors of the
weights. Using in addition (\ref{w3}) for $O(3)$ and (\ref{w4}) for $O(4)$ we
obtain $w=$%
\[
(w_{1},\dots,w_{\left[  N/2\right]  })=\left\{
\begin{array}
[c]{lll}%
\left(  n-n_{1},\dots,n_{\left[  N/2\right]  -1}-n_{\left[  N/2\right]
},n_{\left[  N/2\right]  }\right)  & \text{for} & N~\text{odd}\\
\left(  n-n_{1},\dots,n_{\left[  N/2\right]  -2}-n_{-}-n_{+},n_{-}%
-n_{+}\right)  & \text{for} & N~\text{even.}%
\end{array}
\right.
\]

\item The proof of the highest weight property%
\[
\Psi(\underline{v})M_{\mathring{\gamma}}^{1}=\Psi(\underline{v})M_{\bar{1}%
}^{\mathring{\gamma}}=\Psi(\underline{v})M_{\bar{1}}^{1}=0
\]
uses similar techniques as the derivation of the unwanted terms.

i) We use $\Omega B_{\mathring{\gamma}}(v)$, Yang-Baxter relations and apply
lemma \ref{l2} for $v\rightarrow\infty$
\begin{align*}
0  &  =%
\begin{array}
[c]{c}%
{\phantom{\rule{1.6889in}{1.2671in}}}%
\unitlength4mm\begin{picture}(12,8)\thicklines\put(3.3,2){$v$}
\put(3.3,0){$u_1$}
\put(5.5,0){$u_n$}
\put(7.3,0){1}
\put(8.3,0){1}
\put(9.3,0){1}
\put(5,3.9){$v_i$}
\put(10.2,1.4){$\mathring\gamma$}
\put(7.5,3){$_{\dots}$}
\put(8.8,3){$_{\dots}$}
\put(5,6){$\dots$}
\put(3.3,6.3){1}
\put(4.3,6.3){1}
\put(6.3,6.3){1}
\put(6.5,0){\line(0,1){6}}
\put(4.5,0){\line(0,1){6}}
\put(3.5,1){\oval(12,9)[rt]}
\put(3.5,6){\oval(2,1)[lb]}
\put(3.5,1){\oval(10,7)[rt]}
\put(3.5,6){\oval(3.5,3)[lb]}
\put(3.5,1){\oval(8,5)[rt]}
\put(3.5,6){\oval(5,5)[lb]}
\put(10,6){\oval(13,8)[lb]}
\put(.5,6){$\framebox(2.5,1){$\Pi$}$}
\put(.9,7){\line(0,1){.5}}
\put(1.7,7){\line(0,1){.5}}
\put(2.5,7){\line(0,1){.5}}
\end{picture}%
\end{array}
=%
\begin{array}
[c]{c}%
{\phantom{\rule{1.519in}{1.2671in}}}%
\unitlength4mm\begin{picture}(12,8)\thicklines\put(3.3,1.8){1}
\put(3.3,0){$u_1$}
\put(5.5,0){$u_n$}
\put(7.3,0){1}
\put(8.3,0){1}
\put(9.3,0){1}
\put(5,3.9){$v_i$}
\put(6.8,1.4){$\mathring\gamma$}
\put(7.5,3){$_{\dots}$}
\put(8.8,3){$_{\dots}$}
\put(5,6){$\dots$}
\put(4.3,6.3){1}
\put(6.3,6.3){1}
\put(6.5,0){\line(0,1){6}}
\put(4.5,0){\line(0,1){6}}
\put(3.5,1){\oval(12,9)[rt]}
\put(3.5,6){\oval(2,1)[lb]}
\put(3.5,1){\oval(10,7)[rt]}
\put(3.5,6){\oval(3.5,3)[lb]}
\put(3.5,1){\oval(8,5)[rt]}
\put(3.5,6){\oval(5,5)[lb]}
\put(3.8,2){\line(1,0){3}}
\put(.5,6){$\framebox(2.5,1){$\Pi$}$}
\put(.9,7){\line(0,1){.5}}
\put(1.7,7){\line(0,1){.5}}
\put(2.5,7){\line(0,1){.5}}
\end{picture}%
\end{array}
+O(v^{-2})\\
&  +\sum_{i=1}^{m}c(v-v_{i})%
\begin{array}
[c]{c}%
{\phantom{\rule{1.519in}{1.2671in}}}%
\unitlength4mm\begin{picture}(12,8)\thicklines\put(3.3,0){$u_1$}
\put(5.5,0){$u_n$}
\put(7.3,0){1}
\put(8.3,0){1}
\put(9.3,0){1}
\put(5,3.9){$v_i$}
\put(3.3,1.4){$\mathring\gamma$}
\put(7.5,3){$_{\dots}$}
\put(8.8,3){$_{\dots}$}
\put(5,6){$\dots$}
\put(3.3,6.3){1}
\put(4.3,6.3){1}
\put(6.3,6.3){1}
\put(6.5,0){\line(0,1){6}}
\put(4.5,0){\line(0,1){6}}
\put(3.5,1){\oval(12,9)[rt]}
\put(3.5,6){\oval(2,1)[lb]}
\put(4.5,1){\oval(8,7)[rt]}
\put(4.5,6){\oval(2,3)[lb]}
\put(2.7,6){\oval(2,3)[lb]}
\put(2.7,2.5){\oval(2,4)[rt]}
\put(3.5,1){\oval(8,5)[rt]}
\put(3.5,6){\oval(5,5)[lb]}
\put(.5,6){$\framebox(2.5,1){$\Pi$}$}
\put(.9,7){\line(0,1){.5}}
\put(1.7,7){\line(0,1){.5}}
\put(2.5,7){\line(0,1){.5}}
\end{picture}%
\end{array}
+\sum_{i=1}^{m}c(v_{i}-v)%
\begin{array}
[c]{c}%
{\phantom{\rule{1.6889in}{1.2403in}}}%
\unitlength4mm\begin{picture}(12,8)\thicklines\put(6,1.4){$1$}
\put(7.3,0){1}
\put(8.3,0){1}
\put(9.3,0){1}
\put(4.2,3.9){$v_i$}
\put(10.2,1.4){$\mathring\gamma$}
\put(7.5,3){$_{\dots}$}
\put(8.8,3){$_{\dots}$}
\put(4.2,6){$\dots$}
\put(3.6,6.3){1}
\put(5.6,6.3){1}
\put(5.8,0){\line(0,1){6}}
\put(3.8,0){\line(0,1){6}}
\put(3.5,1){\oval(12,9)[rt]}
\put(3.5,6){\oval(2,1)[lb]}
\put(3.5,3){\oval(10,3)[rt]}
\put(10,3){\oval(3,2)[lb]}
\put(6.5,1){\oval(4,2)[rt]}
\put(3.5,6){\oval(3.5,3)[lb]}
\put(3.5,1){\oval(8,5)[rt]}
\put(3.5,6){\oval(5,5)[lb]}
\put(.5,6){$\framebox(2.5,1){$\Pi$}$}
\put(.9,7){\line(0,1){.5}}
\put(1.7,7){\line(0,1){.5}}
\put(2.5,7){\line(0,1){.5}}
\end{picture}%
\end{array}
\end{align*}
Multiplied with $L(\underline{v})$ this reads in terms of formulas as%
\begin{align*}
0  &  =\left(  L(\underline{v})\Pi(\underline{v})\right)  _{\underline{\beta}%
}\Omega B_{\mathring{\gamma}}(v)\,T_{1}^{\beta_{m}^{\prime}}(v_{m})\cdots
T_{1}^{\beta_{1}^{\prime}}(v_{1})\\
&  =\left(  L(\underline{v})\Pi(\underline{v})\right)  _{\underline{\beta}%
}\Omega\,\left(  R_{0m}(v-v_{m})\dots R_{01}(v-v_{1})\right)  _{\gamma
,\underline{\beta}^{\prime}}^{\underline{\beta},1}T_{1}^{\beta_{m}%
^{\prime\prime}}(v_{m})\cdots T_{1}^{\beta_{1}^{\prime\prime}}(v_{1}%
)T_{\gamma^{\prime}}^{\gamma}(v)\\
&  \times\left(  R_{01}(v_{1}-v)\dots R_{0m}(v_{m}-v)\right)  _{1\dots
1,\mathring{\gamma}}^{\gamma\underline{\beta}^{\prime\prime}}+O(v^{-2})\,.
\end{align*}
With equations (\ref{4.1}), (\ref{4.3}) and using similar techniques as for
the derivation of $uw_{C}^{A}$ and $uw_{C}^{D}$ above we obtain%
\[
0=\Psi(\underline{v})M_{\mathring{\gamma}}^{1}-\sum_{i=1}^{m}X_{\mathring
{\gamma}}^{(i)}(\underline{u},\underline{v})+\sum_{i=1}^{m}X_{\mathring
{\gamma}}^{(i)}(\underline{u},\underline{v}^{(i)})\chi_{i}(\underline
{u},\underline{v})
\]
with $X_{\mathring{\gamma}}^{(i)}$ and $\chi_{i}$ defined in (\ref{Xi}) and
(\ref{chi}). After multiplication with $g(\underline{u},\underline{v})$ and
summation over the $\underline{v}$ the terms cancel each other because of
$\chi_{i}(\underline{u},\underline{v})g(\underline{u},\underline
{v})=g(\underline{u},\underline{v}^{(i)})$.

ii) We consider%
\begin{multline*}%
\begin{array}
[c]{c}%
{\phantom{\rule{1.5064in}{1.2646in}}}%
\unitlength4mm\begin{picture}(10,7)\thicklines\put(3.9,0){$u_1$}
\put(8.3,0){$u_n$}
\put(6,4.2){$v$}
\put(6,2.7){$v_i$}
\put(9.5,4.6){$\bar1$}
\put(9.5,3.1){1}
\put(9.5,2.1){1}
\put(9.5,1.1){1}
\put(4.8,6.1){1}
\put(7.8,6.1){1}
\put(5.8,5.5){$\dots$}
\put(2.8,4.5){$_{\dots}$}
\put(1.5,5){\line(1,0){7.5}}
\put(8,0){\line(0,1){6}}
\put(5,0){\line(0,1){6}}
\put(.8,4.6){$\mathring\gamma$}
\put(9,6){\oval(11,5)[lb]}
\put(9,6){\oval(12.5,7)[lb]}
\put(9,6){\oval(14,9)[lb]}
\put(1.8,6){$\framebox(2,1){$\Pi$}$}
\put(2.,7){\line(0,1){.5}}
\put(2.8,7){\line(0,1){.5}}
\put(3.5,7){\line(0,1){.5}}
\end{picture}%
\end{array}
=\sum_{i=1}^{m}d(v_{i}-v)%
\begin{array}
[c]{c}%
{\phantom{\rule{1.5067in}{1.2649in}}}%
\unitlength4mm\begin{picture}(10,7)\thicklines\put(3.9,0){$u_1$}
\put(8.3,0){$u_n$}
\put(6,4.2){$v$}
\put(6,2.7){$v_i$}
\put(4.2,4.6){$\bar1$}
\put(9.5,3.1){1}
\put(9.5,2.1){1}
\put(9.5,1.1){1}
\put(4.8,6.1){1}
\put(7.8,6.1){1}
\put(5.8,5.5){$\dots$}
\put(8,0){\line(0,1){6}}
\put(5,0){\line(0,1){6}}
\put(.8,4.6){$\mathring\gamma$}
\put(9,6){\oval(11,5)[lb]}
\put(4,4){\oval(2.5,2)[lt]}
\put(9,4){\oval(12.5,3)[lb]}
\put(1.5,6){\oval(2.5,2)[rb]}
\put(9,6){\oval(14,9)[lb]}
\put(1.8,6){$\framebox(2,1){$\Pi$}$}
\put(2.,7){\line(0,1){.5}}
\put(2.8,7){\line(0,1){.5}}
\put(3.5,7){\line(0,1){.5}}
\end{picture}%
\end{array}
+O(v^{-2})\\
=%
\begin{array}
[c]{c}%
{\phantom{\rule{1.519in}{1.2671in}}}%
\unitlength4mm\begin{picture}(10,8)\thicklines\put(3.3,1.8){$\mathring\gamma$}
\put(3.3,0){$u_1$}
\put(5.5,0){$u_n$}
\put(7.3,0){1}
\put(8.3,0){1}
\put(9.3,0){1}
\put(5,3.9){$v_i$}
\put(6.9,1.4){$\bar1$}
\put(7.5,3){$_{\dots}$}
\put(8.8,3){$_{\dots}$}
\put(5,6){$\dots$}
\put(4.3,6.3){1}
\put(6.3,6.3){1}
\put(6.5,0){\line(0,1){6}}
\put(4.5,0){\line(0,1){6}}
\put(3.5,1){\oval(12,9)[rt]}
\put(3.5,6){\oval(2,1)[lb]}
\put(3.5,1){\oval(10,7)[rt]}
\put(3.5,6){\oval(3.5,3)[lb]}
\put(3.5,1){\oval(8,5)[rt]}
\put(3.5,6){\oval(5,5)[lb]}
\put(3.8,2){\line(1,0){3}}
\put(.5,6){$\framebox(2.5,1){$\Pi$}$}
\put(.9,7){\line(0,1){.5}}
\put(1.7,7){\line(0,1){.5}}
\put(2.5,7){\line(0,1){.5}}
\end{picture}%
\end{array}
+\sum_{i=1}^{m}d(v_{i}-v)%
\begin{array}
[c]{c}%
{\phantom{\rule{1.6328in}{1.2671in}}}%
\unitlength4mm\begin{picture}(12,8)\thicklines\put(3.3,0){$u_1$}
\put(5.5,0){$u_n$}
\put(7.3,0){1}
\put(8.3,0){1}
\put(9.3,0){1}
\put(5,3.9){$v_i$}
\put(6.7,1.4){$\mathring\gamma$}
\put(10,1.4){$\bar1$}
\put(7.5,3){$_{\dots}$}
\put(8.8,3){$_{\dots}$}
\put(5,6){$\dots$}
\put(4.3,6.3){1}
\put(6.3,6.3){1}
\put(6.5,0){\line(0,1){6}}
\put(4.5,0){\line(0,1){6}}
\put(3.5,1){\oval(12,9)[rt]}
\put(3.5,6){\oval(2,1)[lb]}
\put(3.5,3){\oval(10,3)[rt]}
\put(7.2,3){\oval(2.6,3)[rb]}
\put(10,1){\oval(3,1.5)[lt]}
\put(3.5,6){\oval(3.5,3)[lb]}
\put(3.5,1){\oval(8,5)[rt]}
\put(3.5,6){\oval(5,5)[lb]}
\put(.5,6){$\framebox(2.5,1){$\Pi$}$}
\put(.9,7){\line(0,1){.5}}
\put(1.7,7){\line(0,1){.5}}
\put(2.5,7){\line(0,1){.5}}
\end{picture}%
\end{array}
+O(v^{-2})
\end{multline*}

Multiplied with $L(\underline{v})$ this reads in terms of formulas as%
\begin{align*}
&  \left(  L(\underline{v})\Pi(\underline{v})\right)  _{\underline{\beta}%
}\left(  R_{10}(v_{1}-v)\dots R_{m0}(v_{m}-v)\right)  _{\underline{\beta
}^{\prime},\gamma}^{\mathring{\gamma}\underline{\beta}}\Omega T_{\bar{1}%
}^{\gamma}(v)\,T_{1}^{\beta_{m}^{\prime}}(v_{m})\cdots T_{1}^{\beta
_{1}^{\prime}}(v_{1})\\
&  =\left(  L(\underline{v})\Pi(\underline{v})\right)  _{\underline{\beta}%
}\left(  R_{10}(v_{1}-v)\dots R_{m0}(v_{m}-v)\right)  _{\bar{1},\underline
{\beta}^{\prime}}^{\mathring{\gamma}\underline{\beta}}\Omega\,T_{1}^{\beta
_{m}^{\prime}}(v_{m})\cdots T_{1}^{\beta_{1}^{\prime}}(v_{1})+O(v^{-2})\\
&  =\Psi(\underline{v})T_{\bar{1}}^{\mathring{\gamma}}(v)+O(v^{-2})\\
&  +\left(  L(\underline{v})\Pi(\underline{v})\right)  _{\underline{\beta}%
}\Omega\,T_{\beta_{m}^{\prime}}^{\beta_{m}}(v_{m})\cdots T_{\beta_{1}^{\prime
}}^{\beta_{1}}(v_{1})\left(  R_{10}(v_{1}-v)\dots R_{m0}(v_{m}-v)\right)
_{1\dots1,\bar{1}}^{\mathring{\gamma}\underline{\beta}^{\prime}}\,.
\end{align*}
It has been used that only $\gamma=\bar{1}$ contributes because of $\Omega
B_{2}=\Omega B_{3}=0$. We apply lemma \ref{l2} for $v\rightarrow\infty$. With
equations (\ref{4.1}), (\ref{4.3}) and using similar techniques as for the
derivation of $uw_{C_{3}}^{D}$ and $uw_{C_{3}}^{A_{3}}$ above we obtain%
\[
0=\Psi(\underline{v})M_{\bar{1}}^{\mathring{\gamma}}-\mathbf{C}^{\mathring
{\gamma}\mathring{\gamma}^{\prime}}\sum_{i=1}^{m}X_{\mathring{\gamma}^{\prime
}}^{(i)}(\underline{u},\underline{v})+\mathbf{C}^{\mathring{\gamma}%
\mathring{\gamma}^{\prime}}\sum_{i=1}^{m}X_{\mathring{\gamma}^{\prime}}%
^{(i)}(\underline{u},\underline{v}^{(i)})\chi_{i}(\underline{u},\underline
{v})\,.
\]
Again after multiplication with $g(\underline{u},\underline{v})$ and summation
over the $\underline{v}$ the terms cancel each other because of $\chi
_{i}(\underline{u},\underline{v})g(\underline{u},\underline{v})=g(\underline
{u},\underline{v}^{(i)})$.

iii) We consider%
\[
0=\Omega M_{\bar{1}}^{1}\dots=%
\begin{array}
[c]{c}%
{\phantom{\rule{1.6472in}{1.2574in}}}%
\unitlength4mm\begin{picture}(12,8)\thicklines\put(3.3,2){$v$}
\put(3.3,0){$u_1$}
\put(5.5,0){$u_n$}
\put(7.3,0){1}
\put(8.3,0){1}
\put(9.3,0){1}
\put(5,3.9){$v_i$}
\put(10.2,1.4){$\bar1$}
\put(7.5,3){$_{\dots}$}
\put(8.8,3){$_{\dots}$}
\put(5,6){$\dots$}
\put(3.3,6.3){1}
\put(4.3,6.3){1}
\put(6.3,6.3){1}
\put(6.5,0){\line(0,1){6}}
\put(4.5,0){\line(0,1){6}}
\put(3.5,1){\oval(12,9)[rt]}
\put(3.5,6){\oval(2,1)[lb]}
\put(3.5,1){\oval(10,7)[rt]}
\put(3.5,6){\oval(3.5,3)[lb]}
\put(3.5,1){\oval(8,5)[rt]}
\put(3.5,6){\oval(5,5)[lb]}
\put(10,6){\oval(13,8)[lb]}
\put(.5,6){$\framebox(2.5,1){$\Pi$}$}
\put(.9,7){\line(0,1){.5}}
\put(1.7,7){\line(0,1){.5}}
\put(2.5,7){\line(0,1){.5}}
\end{picture}%
\end{array}
=%
\begin{array}
[c]{c}%
{\phantom{\rule{1.5075in}{1.2574in}}}%
\unitlength4mm\begin{picture}(12,8)\thicklines\put(3.3,1.8){1}
\put(3.3,0){$u_1$}
\put(5.5,0){$u_n$}
\put(7.3,0){1}
\put(8.3,0){1}
\put(9.3,0){1}
\put(5,3.9){$v_i$}
\put(6.8,1.4){$\bar1$}
\put(7.5,3){$_{\dots}$}
\put(8.8,3){$_{\dots}$}
\put(5,6){$\dots$}
\put(4.3,6.3){1}
\put(6.3,6.3){1}
\put(6.5,0){\line(0,1){6}}
\put(4.5,0){\line(0,1){6}}
\put(3.5,1){\oval(12,9)[rt]}
\put(3.5,6){\oval(2,1)[lb]}
\put(3.5,1){\oval(10,7)[rt]}
\put(3.5,6){\oval(3.5,3)[lb]}
\put(3.5,1){\oval(8,5)[rt]}
\put(3.5,6){\oval(5,5)[lb]}
\put(3.8,2){\line(1,0){3}}
\put(.5,6){$\framebox(2.5,1){$\Pi$}$}
\put(.9,7){\line(0,1){.5}}
\put(1.7,7){\line(0,1){.5}}
\put(2.5,7){\line(0,1){.5}}
\end{picture}%
\end{array}
\]%
\begin{align*}
&  +\sum_{i=1}^{m}c(v-v_{i})%
\begin{array}
[c]{c}%
{\phantom{\rule{1.5075in}{1.2574in}}}%
\unitlength4mm\begin{picture}(12,8)\thicklines\put(3.3,0){$u_1$}
\put(5.5,0){$u_n$}
\put(7.3,0){1}
\put(8.3,0){1}
\put(9.3,0){1}
\put(5,3.9){$v_i$}
\put(3.3,1.4){$\bar1$}
\put(7.5,3){$_{\dots}$}
\put(8.8,3){$_{\dots}$}
\put(5,6){$\dots$}
\put(3.3,6.3){1}
\put(4.3,6.3){1}
\put(6.3,6.3){1}
\put(6.5,0){\line(0,1){6}}
\put(4.5,0){\line(0,1){6}}
\put(3.5,1){\oval(12,9)[rt]}
\put(3.5,6){\oval(2,1)[lb]}
\put(4.5,1){\oval(8,7)[rt]}
\put(4.5,6){\oval(2,3)[lb]}
\put(2.7,6){\oval(2,3)[lb]}
\put(2.7,2.5){\oval(2,4)[rt]}
\put(3.5,1){\oval(8,5)[rt]}
\put(3.5,6){\oval(5,5)[lb]}
\put(.5,6){$\framebox(2.5,1){$\Pi$}$}
\put(.9,7){\line(0,1){.5}}
\put(1.7,7){\line(0,1){.5}}
\put(2.5,7){\line(0,1){.5}}
\end{picture}%
\end{array}
+\sum_{i=1}^{m}d(v_{i}-v)%
\begin{array}
[c]{c}%
{\phantom{\rule{1.6194in}{1.2579in}}}%
\unitlength4mm\begin{picture}(12,8)\thicklines\put(3.3,0){$u_1$}
\put(5.5,0){$u_n$}
\put(7.3,0){1}
\put(8.3,0){1}
\put(9.3,0){1}
\put(5,3.9){$v_i$}
\put(6.7,1.4){$ 1$}
\put(10,1.4){$\bar1$}
\put(7.5,3){$_{\dots}$}
\put(8.8,3){$_{\dots}$}
\put(5,6){$\dots$}
\put(4.3,6.3){1}
\put(6.3,6.3){1}
\put(6.5,0){\line(0,1){6}}
\put(4.5,0){\line(0,1){6}}
\put(3.5,1){\oval(12,9)[rt]}
\put(3.5,6){\oval(2,1)[lb]}
\put(3.5,3){\oval(10,3)[rt]}
\put(7.2,3){\oval(2.6,3)[rb]}
\put(10,1){\oval(3,1.5)[lt]}
\put(3.5,6){\oval(3.5,3)[lb]}
\put(3.5,1){\oval(8,5)[rt]}
\put(3.5,6){\oval(5,5)[lb]}
\put(.5,6){$\framebox(2.5,1){$\Pi$}$}
\put(.9,7){\line(0,1){.5}}
\put(1.7,7){\line(0,1){.5}}
\put(2.5,7){\line(0,1){.5}}
\end{picture}%
\end{array}
\\
&  +\sum_{i=1}^{m}c(v_{i}-v)%
\begin{array}
[c]{c}%
{\phantom{\rule{1.6472in}{1.23in}}}%
\unitlength4mm\begin{picture}(12,8)\thicklines\put(6,1.4){$1$}
\put(7.3,0){1}
\put(8.3,0){1}
\put(9.3,0){1}
\put(4.2,3.9){$v_i$}
\put(10.2,1.4){$\bar1$}
\put(7.5,3){$_{\dots}$}
\put(8.8,3){$_{\dots}$}
\put(4.2,6){$\dots$}
\put(3.6,6.3){1}
\put(5.6,6.3){1}
\put(5.8,0){\line(0,1){6}}
\put(3.8,0){\line(0,1){6}}
\put(3.5,1){\oval(12,9)[rt]}
\put(3.5,6){\oval(2,1)[lb]}
\put(3.5,3){\oval(10,3)[rt]}
\put(10,3){\oval(3,2)[lb]}
\put(6.5,1){\oval(4,2)[rt]}
\put(3.5,6){\oval(3.5,3)[lb]}
\put(3.5,1){\oval(8,5)[rt]}
\put(3.5,6){\oval(5,5)[lb]}
\put(.5,6){$\framebox(2.5,1){$\Pi$}$}
\put(.9,7){\line(0,1){.5}}
\put(1.7,7){\line(0,1){.5}}
\put(2.5,7){\line(0,1){.5}}
\end{picture}%
\end{array}
+O(v^{-2})
\end{align*}
or in terms of formulas%
\begin{align*}
0  &  =\Omega M_{\bar{1}}^{1}\dots\\
&  =\left(  L(\underline{v})\Pi(\underline{v})\right)  _{\underline{\beta}%
}\left(  R(v-v_{m})\dots R(v-v_{1})\right)  _{\gamma,\underline{\beta}%
^{\prime}}^{\underline{\beta},1}\Omega\,T_{\beta_{m}^{\prime\prime}}%
^{\beta_{m}^{\prime}}(v_{m})\cdots T_{\beta_{1}^{\prime\prime}}^{\beta
_{1}^{\prime}}(v_{1})T_{\gamma^{\prime}}^{\gamma}(v)\\
&  \times\left(  R(v_{1}-v)\dots R(v_{m}-v)\right)  _{1\dots1,\bar{1}}%
^{\gamma^{\prime},\underline{\beta}^{\prime\prime}}\,.
\end{align*}
For $v\rightarrow\infty$ we apply lemma \ref{l2}, equations (\ref{4.1}) and
(\ref{4.3}) and obtain%
\[
0=\Psi(\underline{v})M_{\bar{1}}^{1}-\sum_{i=1}^{m}X^{(ij)}(\underline
{v})+\sum_{i=1}^{m}X^{(ij)}(v_{i}^{\prime},v_{j}^{\prime},\underline{v}%
_{ij})\chi_{i}(\underline{u},\underline{v})\chi_{j}(\underline{u}%
,\underline{v})
\]
where similar techniques as above for the derivation of the unwanted termd
have bee used. Again after multiplication with $g(\underline{u},\underline
{v})$ and summation over the $\underline{v}$ the terms cancel each other
because of $\chi_{i}(\underline{u},\underline{v})\chi_{j}(\underline
{u},\underline{v})g(\underline{u},\underline{v})=g(\underline{u},\underline
{v}^{(ij)})$.

iv) Next we prove%
\[
\Psi(\underline{v})M_{\mathring{\gamma}}^{\mathring{\gamma}^{\prime}%
}=0,~1<\mathring{\gamma}^{\prime}<\mathring{\gamma}<\bar{1}\,.
\]

We consider
\begin{align*}
&  L_{\underline{\beta}^{\prime}}(\underline{v})\Pi_{\gamma\underline{\beta}%
}^{\mathring{\gamma}^{\prime}\underline{\beta}^{\prime}}(v,\underline
{v})\Omega\,T_{1}^{\beta_{m}}(w_{m})\cdots T_{1}^{\beta_{1}}(w_{1}%
)T_{\mathring{\gamma}}^{\gamma}(v)+O(v^{-2})\\
&  =\left(  L(\underline{v})\left(  T^{(1)}\right)  _{\mathring{\gamma}%
}^{\mathring{\gamma}^{\prime}}(v)\right)  _{\underline{\beta}^{\prime}}%
\Pi_{\underline{\beta}}^{\underline{\beta}^{\prime}}(\underline{v}%
)\Omega\,T_{\mathring{\gamma}}^{\mathring{\gamma}}(v)T_{1}^{\beta_{m}}%
(w_{m})\cdots T_{1}^{\beta_{1}}(w_{1})+O(v^{-2})
\end{align*}
where Yang-Baxter rules and (\ref{RPi}) have been used. We have also used that
by (\ref{Pia}) and (\ref{Sb})%
\begin{align*}
\Pi_{\gamma\underline{\beta}}^{\mathring{\gamma}^{\prime}\underline{\beta
}^{\prime}}(v,\underline{v})  &  =\delta_{\gamma}^{\mathring{\gamma}^{\prime}%
}\Pi_{\underline{\beta}}^{\underline{\beta}^{\prime}}(\underline{v}%
)+O(v^{-1})\\
\left(  R(w_{1}-v)\dots R(w_{m}-v)\right)  _{1\dots m,0}  &  =\mathbf{1}%
_{1\dots m}\mathbf{1}_{0}+O(v^{-1}).
\end{align*}
For $v\rightarrow\infty$ the highest weight condition $L(\underline{v})\left(
M^{(1)}\right)  _{\mathring{\gamma}}^{\mathring{\gamma}^{\prime}}=0$ implies
the claim.

\item The highest weight properties of the weights are obtained as follows.
The commutation relation relation (\ref{4.8}) reads in the complex basis as%
\[
\lbrack M_{\alpha}^{\alpha^{\prime}},M_{\beta}^{\beta^{\prime}}]=-\delta
_{\alpha}^{\beta^{\prime}}M_{\beta}^{\alpha^{\prime}}+\mathbf{C}%
^{\alpha^{\prime}\beta^{\prime}}\left(  \mathbf{C}M\right)  _{\alpha\beta
}+M_{\alpha}^{\beta^{\prime}}\delta_{\beta}^{\alpha^{\prime}}-\left(
M\mathbf{C}\right)  ^{\beta^{\prime}\alpha^{\prime}}\mathbf{C}_{\alpha\beta
}\,.
\]
In particular for $\beta\neq\alpha,\bar{\alpha}$%
\[
\lbrack M_{\alpha}^{\beta},M_{\beta}^{\alpha}]=M_{\alpha}^{\alpha}-M_{\beta
}^{\beta}=M_{\alpha}^{\alpha}+M_{\bar{\beta}}^{\bar{\beta}}\,.
\]
Because of $\left(  M_{\alpha}^{\beta}\right)  ^{\dag}=M_{\beta}^{\alpha}$%
\[
0\leq M_{\alpha}^{\beta}\left(  M_{\alpha}^{\beta}\right)  ^{\dag}=M_{\alpha
}^{\beta}M_{\beta}^{\alpha}=M_{\beta}^{\alpha}M_{\alpha}^{\beta}+M_{\alpha
}^{\alpha}-M_{\beta}^{\beta}\,.
\]
Applying this to highest weight co-vectors with%
\[
0=\Psi M_{\beta}^{\alpha}~\text{for }\alpha<\beta
\]
we obtain for the weights (\ref{w})%
\[
0\leq w_{\alpha}-w_{\beta}~\text{for }\alpha<\beta\leq N/2\,.
\]
In addition if $N$ is even%
\begin{align*}
0  &  \leq w_{\alpha}+w_{\bar{\beta}}~\text{for }\alpha\leq N/2<\beta\neq
\bar{\alpha}\\
&  \Rightarrow w_{1}\geq w_{2}\geq\dots\geq w_{N/2-1}\geq|w_{N/2}|
\end{align*}
and if $N$ is odd
\begin{align*}
0  &  \leq w_{\alpha}~\text{for }\alpha\leq N/2~\text{because }\Psi M_{0}%
^{0}=0\\
&  \Rightarrow w_{1}\geq w_{2}\geq\dots\geq w_{N/2}\geq0\,.
\end{align*}

\end{enumerate}
\end{proof}


\end{document}